\documentclass[12pt,a4paper]{article}

\usepackage[T1]{fontenc}
\usepackage{setspace}
\usepackage{bm}
\usepackage{amsthm}
\usepackage{mathrsfs}  
\usepackage{microtype}
\usepackage[hmargin=2.5cm, vmargin=2.5cm]{geometry}
\usepackage[inline]{enumitem}
\setlist{nosep} 
\usepackage[titletoc,title]{appendix}
\usepackage{graphicx}
\usepackage{amssymb,amsmath}
\usepackage{dsfont}
\usepackage[official]{eurosym}
\usepackage{natbib}
\usepackage[
            colorlinks=false,hidelinks,
            citecolor=black ]{hyperref} 
\usepackage{csquotes}            
\usepackage{caption}
\usepackage{pdfrender,xcolor}   
\usepackage{adjustbox}
\usepackage[flushleft]{threeparttable}
\usepackage{siunitx}
\usepackage{booktabs}
\usepackage{bigints}
\usepackage{color}
\color{black}
\newlist{mylist}{enumerate*}{1}
\setlist[mylist]{label=(\roman*)}   
        
\numberwithin{equation}{section}

\newtheorem{defn}{Definition}[section]
\newtheorem{theorem}{Theorem}[section]

\newtheorem{proposition}[theorem]{Proposition}
\newtheorem{corollary}[theorem]{Corollary}

\newtheorem{remark}{Remark}

\makeatletter
\newcommand*{\rom}[1]{\expandafter\@slowromancap\romannumeral #1@}

\newcommand{\cond}{\;\middle\vert\;} 

\newcommand{\EE}{\mathds{E}}  

\newcommand{\eF}{\mathcal{F}}

\newcommand{\Ind}{\mathds{1}} 
\newcommand{\eL}{\mathcal{L}}

\newcommand{\Oh}{\mathcal{O}} 
\newcommand{\PP}{\mathds{P}}  
\newcommand{\QQ}{\mathds{Q}}  
\newcommand{\RR}{\mathds{R}}  
\newcommand{\NN}{\mathds{N}}
\newcommand{\DD}{\mathds{D}}
\newcommand{\Var}{\mathds{V}\kern-2pt\text{ar}} 
\newcommand{\Cov}{\mathds{C}\kern-2pt\text{ov}} 

\newcommand{\XX}{\mathds{X}}  

\let\oldsqrt\sqrt
\def\sqrt{\mathpalette\APsqrt}
\def\APsqrt#1#2{%
\setbox0=\hbox{$#1\oldsqrt{#2}$}\dimen0=\ht0%
\advance\dimen0-0.2\ht0%
\setbox2=\hbox{\vrule height\ht0 depth -\dimen0}%
{\box0\lower0.48pt\box2}} 
\usepackage[font=footnotesize]{caption}

\allowdisplaybreaks
\begin{document}
{\singlespacing
\title{\textbf{{\Large Near-Optimal Dynamic Asset Allocation in Financial Markets with Trading Constraints}}\footnote{We are very grateful to the participants at the 23rd International Congress on Insurance: Mathematics and Economics, the Netspar Pension Day, the Vienna Congress on Mathematical Finance, and the Workshop `Fair Valuation' in Insurance for their helpful comments and suggestions.}}
\author{
    Thijs Kamma\footnote{Corresponding author. Mailing address: t.kamma@maastrichtuniversity.nl. PO Box 616, 6200 MD Maastricht, The Netherlands. Phone: +31 (0) 43 388 3695.}\\{\small Dept. of Quantitative Economics}\\
    {\small Maastricht University }\\
    {\small NETSPAR }
    \and
    Antoon Pelsser\\
    {\small Dept. of Quantitative Economics}\\
    {\small Maastricht University}\\
    {\small  NETSPAR }
    \and
}
\date{  \today}
\maketitle
\vspace{-1cm}
\begin{abstract}
{\footnotesize \color{black}We develop a dual-control method for approximating investment strategies in incomplete environments that emerge from the presence of trading constraints. Convex duality enables the approximate technology to generate lower and upper bounds on the optimal value function. The mechanism rests on closed-form expressions pertaining to the portfolio composition, from which we are able to derive the near-optimal asset allocation explicitly. In a real financial market, we illustrate the accuracy of our approximate method on a dual CRRA utility function that characterises the preferences of a finite-horizon investor. Negligible duality gaps and insignificant annual welfare losses substantiate accuracy of the technique.}
\end{abstract}
\maketitle
}
%

\section{Introduction}
\textsc{The impossibility to acquire closed-form solutions} to a myriad of portfolio choice problems advocates the design of approximate methods. This failure to arrive at explicit expressions depends on the agent's preference qualifications and technical assumptions on the subject of undiversifiability in the financial environment, cf. \citet{kim1996dynamic}, \citet{wachter2002portfolio} and \citet{liu2006portfolio}. To avoid this inability, we propose a dual-control technique for approximating the dynamic asset allocation in markets that embed general trading constraints and an expansive set of preference conditions. The technology that we delineate in this article is tractable, since it is based in closed-form expressions essential to the portfolio composition that we collect through Malliavin calculus. We assess the accuracy of our approximating mechanism by means of a utility loss criterion, which emanates as a result of the available duality bounds. Convex duality in relation to artificial markets namely enables us to provide lower and upper bounds on the truly optimal value function. Insignificant duality gaps and welfare losses varying between 1 and 5 basis points allude to near-optimality of the method.

Constrained consumption and portfolio choice problems are typically difficult to solve. Therefore, it is conventional to resort to an examination of their corresponding duals. In the influential papers by \citet{karatzas1991martingale}, \citet{cvitanic1992convex}, and \citet {xu1992duality} the authors show that these primal and dual problems reconcile through a barrier cone which relates the choice of, respectively, a controlled allocation of assets to the market prices of risk. In fact, the dual problem induces a so-called fictitious economic environment, wherein the objective concerns minimisation over the shadow prices of non-traded or partially traded uncertainty. Adequately minimising the dual recovers the true market, couples the set of shadow prices via the cone to feasible and optimal decision rules, and ensures persistence of strong duality.\footnote{See Proposition 11.4 in \citet{karatzas1991martingale} or Proposition 12.1 in \citet{cvitanic1992convex}.} Despite that, the resultant shadow prices ordinarily  outline backward-forward equations that do not permit analytical recovery of their solutions, making the call for approximation rampant.

In view of the notion that the inexistence of closed-form solutions to the optimal trading rules and consumption behaviour stems entirely from the backward-forward equations identifying the shadow prices, we confine the set of attainable dual controls to a tractable parametric family: deterministically constant or affine in the uncertainty, to name a few. Under this stricture regarding the set to which the shadow prices ought to conform, we are able to derive definite expressions for the controls that minimise the dual. Thereon, the analytical dual value function prompts an upper bound on the true one, on account of strong duality. The magnitude of the bound's deviation from the optimal value depends on the quality of the approximation inherent in the restriction of the set of feasible dual controls. In parallel, the approximate sub-optimal shadow prices give rise via the barrier cone to continuous consumption streams and investment behaviour that do not by definition comply with the trading constraints in the baseline financial market. 

To obtain portfolio decisions that are admissible in the true economic environment, we turn to the primal problem specification, and propose a construct for a candidate solution. In particular, the absence of arbitrage empowers us to express the wealth and consumption processes under an equivalent measure as local martingales, see \citet{HarriKreps:1979:Martingalesandarbitrage} and \citet{fHAR81a}. Consequently, by the Martingale Representation Theorem (MRT), we are able to gather the optimal allocation of resources. The Clark-Ocone formula allows us under mild regularity conditions to identify the integrand in the MRT as a previsible projection of a generalised Fr\'{e}chet derivative in the Wiener-direction, cf. \citet{karatzas1991extension}, \citet{ocone1991generalized}, and \citet{di2009malliavin}. Dependent on the delimitation of the set of shadow prices, the consequential allocation concurs up to the market prices with the infeasible one emerging from the dual. We then project under a case-specific metric the portfolio decisions to the admissible region, whereby we retrieve closed-form approximations to the true trading strategies. Numerical maximisation over the therein incorporated shadow prices yields the approximate analytical portfolio and engenders a lower bound on the optimal value function.\footnote{We demonstrate that the numerical routine is redundant: one may insert the approximate closed-form controls implicit in the dual into the fully analytical projected primal controls for identical performance.} 

The discrepancy between the lower and upper bounds arising from the primal and dual problems is the optimality gap and aids in quantifying the precision of the approximation. Concretely, this utilitarian difference translates itself into an annual welfare loss that grows with the level of inaccuracy of the approximation. In order to illustrate the authenticity of the technique, we cast the problem into \cite{brennan2002dynamic}'s incomplete economy (e.g. \citet{battocchio2004optimal}, \cite{sangvinatsos2005does}, \citet{cairns2006stochastic}, \citet{munk2007optimal}, \citet{DeJong:2008:Pensionfundinvestments} and \citet{DeJong:2008:Valuationpensionliabilities} supply resembling economies), which accommodates unspanned inflation risk and inhabits a finite-horizon investor. Herein, we design and introduce the \textit{dual CRRA} preference qualification, which describes the investor's utility. In this way, we enlarge the framework to one that applies to state-dependent preferences. Resulting duality gaps and welfare losses are negligibly small. 

The attributes that differentiate this study from the existing literature on approximate methods, e.g. \citet{cvitanic2003monte}, \citet{Detemple2003} \citet{BraGoySaSt:2005:simulationapproachto}, \citet{keppo2007computational}, and \citet{koijen2009can}, which prominently employ the duality bounds and convex analysis in evaluating these, cf. \citet{haugh2006evaluating}, \citet{brown2010information}, \citet{brown2011dynamic}, \citet{bick2013solving}, and \citet{ma2017dual}, are threefold: \begin{mylist}
  \item the inclusion of state-dependent preferences, \item the applicability in a set of markets that may comprise non-Markovian return dynamics, and \item the incorporation of general trading restrictions.
\end{mylist} Methodologically, \citet{bick2013solving} come closest to our method of approximation. Their study involves a CRRA agent in a Markovian model, who receives labour income. Ignoring this endowment, we extend their mode by making it applicable outside of CRRA stipulations in markets that go beyond Markovian conditions.

The remainder of the paper is structured as follows. Section \ref{sec2} introduces the financial market model. Section \ref{sec3} derives optimality conditions to the constrained investment-consumption problem. Section \ref{sec4} introduces the approximate method. Section \ref{sec5} provides a numerical illustration of the aforementioned technology. Finally, section \ref{sec6} concludes. 

\section{Model Setup}\label{sec2}
This section introduces the economic setting by close analogy with the models in \citet{brennan2002dynamic} and \cite{detemple2009dynamic}. In the spirit of \citet{cvitanic1992convex}, we generalise their economy by including general trading constraints.
\subsection{Financial Market Model}\label{sec2.1}
Define a horizon $T>0$, and consider a probability space $(\Omega,\eF_t,\left\lbrace\eF_t\right\rbrace_{t\in\left[0,T\right]},\PP)$, whereon an $\RR^{d+m}$-valued independent standard Brownian motion, $\left\lbrace W_t\right\rbrace_{t\in\left[0,T\right]}=\left\lbrace W^1_t, W^2_t\right\rbrace_{t\in\left[0,T\right]}$, exists: $W^1_t$ is $\RR^d$-valued. We split $W_t$ to lighten later derivations. The $\PP$-augmentation of $W_t$'s canonical filtration $\left\lbrace\eF_t^{W}\right\rbrace_{t\in\left[0,T\right]}$ reads $\left\lbrace \eF_t\right\rbrace_{t\in\left[0,T\right]}$.\footnote{Identify the Wiener space: $\Omega=\mathcal{C}_0\left(\left[0,T\right];\RR^{d}\right)\times\mathcal{C}_0\left(\left[0,T\right];\RR^{m}\right)$, $\eF=\mathcal{B}_T^d\otimes\mathcal{B}^{m}_T$, $\PP\left(\mathrm{d}W\right)=\PP\left(\mathrm{d}W^1\right)\otimes\PP\left(\mathrm{d}W^2\right)$, cf. \cite{detemple2005closed}; $\eF_t^{W_i}$ is the $\PP$-augmented filtration of $W^i_t$'s, $i=1,2$.} In the sequel, (in)equalities between stochastic processes are understood in a $\PP$-almost sure sense.

The financial environment, say $\mathcal{M}$, distinguishes a real and a nominal market and therefore includes a commodity price index, $\Pi_t$, that harbours both $W^1_t$ and $W^2_t$ as
\begin{equation}
\mathrm{d}\Pi_t=\Pi_t\left[\pi_t\mathrm{d}t+{\xi^{\Pi_1}_t}^{\top}\mathrm{d}W^1_t+{\xi^{\Pi_2}_t}^{\top}\mathrm{d}W^2_t\right],\ \Pi_0=1.
\end{equation}
Here, the $\eF_t$-measurable $\pi_t$ denotes the rate of inflation. Moreover, $\xi^{\Pi_i}_t$ is $\eF^{W_1}_t$-measurable for $i=1,2$. We assume that all three processes are in $\DD^{1,2}\left(\left[0,T\right]\right)$, see \citet{nualart2006malliavin}. Multiplying any asset with the reciprocal of $\Pi_t$ spawns its nominal value. Thereon,
\begin{equation}
\mathrm{d}M_t=M_t\left[-r_t\mathrm{d}t+{\phi^{M_1}_t}^{\top}\mathrm{d}W^1_t+{\phi^{M_2}_t}^{\top}\mathrm{d}W^2_t\right],\ M_0=1
\end{equation}
characterises the real pricing kernel, where $r_t$ defines the $\eF_t$-measurable real interest rate, $\phi^{M_1}_t$ is $\eF^{W_1}_t$-measurable, and $\phi^{M_2}_t$ is $\eF_t$-measurable. Suppose that $L_t$ is a traded process, then $\left\lbrace L_t\Pi_t^{-1}M_t\right\rbrace_{t\in\left[0,T\right]}$ must by construction be a $\PP$-martingale. 

As a consequence, the nominal state price density (SPD) process complies with 
\begin{equation}\label{eq:genkernel}
\mathrm{d}Z_t=Z_t\left[\left(-r_t-\pi_t +\lambda_{1,t}^{\top}{\xi^{\Pi_1}_t}+\lambda_{2,t}^{\top}{\xi^{\Pi_2}_t} \right)\mathrm{d}t-\lambda_{1,t}^{\top}\mathrm{d}W^1_t-\lambda_{2,t}^{\top}\mathrm{d}W^2_t\right],
\end{equation}
in which $\lambda_{i,t}=\xi^{\Pi_i}_t-\phi^{M_i}_t$, $i=1,2$, delineate the nominal market prices of risk, and where we let $Z_t:=\Pi_t^{-1}M_t$. We postulate that $r_t\in\DD^{1,2}\left(\left[0,T\right]\right)$ and $\widehat{\phi}_t\in\DD^{1,2}\left(\left[0,T\right]\right)^{d+m}$ given that $\widehat{\phi}_t:=\left[\phi_{t}^{M_1},\phi_{t}^{M_2}\right]^{\top}$, securing Malliavin-differentiability of $\lambda_{i,t}$, for $i=1,2$. Novikov's condition, cf. Sec 3.5.D in \citet{ks:stc}, therefore holds. From \eqref{eq:genkernel},   
\begin{equation}\label{eq:appcnormalinainster}
R_{f,t}=r_t+\pi_t-\lambda_{1,t}^{\top}{\xi^{\Pi_1}_t}-\lambda_{2,t}^{\top}{\xi^{\Pi_2}_t}
\end{equation}
holds, signifying the identity for the instantaneous nominal interest rate.

The economic model accommodates an instantaneously risk-free asset, i.e. a money market account, and $d+m$ non-dividend paying risky assets that evolve pursuant to
\begin{equation}\label{eq:sec21stockproc}
\frac{\mathrm{d}S_t}{S_t}=\left(R_{f,t}1_{d+m}+\sigma_t^{S_1}\lambda_{1,t}+\sigma_t^{S_2}\lambda_{2,t}\right)\mathrm{d}t+\sigma^{S_1}_t\mathrm{d}W^1_t+\sigma^{S_2}_t\mathrm{d}W^2_t,\ S_0=1_{d+m},
\end{equation}
where $\mathrm{diag}\big(\sigma_t^{S_i}{\sigma_t^{S_i}}^{\top}\big)\in L^1\left(\left[0,T\right]\right)^{d+m}$ for $\eF_t$-measurable volatility matrices $\sigma_t^{S_i}$ that fulfill the strong non-degeneracy assumption $\psi^{\top}\sigma_t^{S_i}{\sigma_t^{S_i}}^{\top}\psi\geq \epsilon_i\|\psi\|_{\RR^{d+m}}$ for all $\psi\in\RR^{d+m}$, some $\epsilon_i\in\RR_+$, $i=1,2$. The foregoing assumption ensures invertibility of $\left[\sigma_t^{S_1},\sigma_t^{S_2}\right]$. Further,
\begin{equation}
\mathrm{d}B_t=R_{f,t}B_t\mathrm{d}t,\ B_0=1,
\end{equation}
details the cash account. This construction relies on those in \citet{brennan2002dynamic} and \citet{detemple2009dynamic}. Modifications of the trading restrictions pertaining to $S_t$ authorise us to control $\left\lbrace W_t\right\rbrace_{t\in\left[0,T\right]}$'s degree of diversifiabability.

The wealth process in agreement with that of a finite-horizon investor obeys 
\begin{equation}\label{eq:genactwealth}
\mathrm{d}X_t=X_t\left[\left(R_{f,t}+x^{\top}_t\widehat{\sigma}_t\widehat{\lambda}_{t}\right)\mathrm{d}t+x_t^{\top}\widehat{\sigma}_t\mathrm{d}W_{t}\right]-c_t\mathrm{d}t, \  X_0\in\RR_+,
\end{equation} 
for an initial endowment $X_0$, where $c_t$ denotes the progressively $\eF_t$-measurable consumption process, and $x_t$ the $\RR^{d+m}$-valued progressively $\eF_t$-measurable fraction of $X_t$ allocated to $S_t$. Moreover, we define $\widehat{\sigma}_t:=\left[\sigma_t^{S_1},\sigma_t^{S_2}\right]$ and $\widehat{\lambda}_t:=\left[\lambda_{1,t},\lambda_{2,t}\right]^{\top}$. Note that the investor determines $\left\lbrace x_t, c_t\right\rbrace_{t\in\left[0,T\right]}$ in order to maximise an expected utility criterion. We call a pair $\left\lbrace x_t,c_t\right\rbrace_{t\in\left[0,T\right]}$ \textit{admissible} if in addition to $X_t\geq 0$ $\forall$ $t\in\left[0,T\right]$, it fulfils 
\begin{equation}\label{eq:sec21assump}
x_t^{\top}\widehat{\sigma}_t{\widehat{\sigma}_t}^{\top}x_t\in L^1\left(\left[0,T\right]\right),\ \left|x^{\top}_t\widehat{\sigma}_t\widehat{\lambda}_{t}\right|\in L^1\left(\left[0,T\right]\right),\ \mathrm{and}\ \ c_t\in L^2\left(\Omega\times\left[0,T\right]\right).
\end{equation}
We denote the class of all such admissible trading-consumption pairs by $\mathcal{A}_{X_0}$.

Ultimately, suppose that $K\subseteq\RR^{d+m}$ describes a non-empty closed and convex set that comprises the investment constraints imposed on $x_t$ in a $\mathrm{d}t\otimes\PP\text{-a.e.}$ sense. Accordingly, let us define $\widehat{\mathcal{A}}_{X_0}$ as the set of all $\left\lbrace x_t, c_t\right\rbrace_{t\in\left[0,T\right]}\in\mathcal{A}_{X_0}$ such that $x_t\in K$ holds $\mathrm{d}t\otimes\PP\text{-a.e.}$ We remark that all the exploited spaces conform to their conventional definitions.\footnote{In fact, $L^p\left(\left[0,T\right]\right)$ is the Lebesgue space of all $\PP\text{-a.s.}$ $p$-integrable stochastic processes, $p\geq 0$.}

\subsection{Fictitious Completion Assets}\label{sec2.2}
Assume that the market in section \ref{sec2.1} is constrained concerning admissibility of $\left\lbrace x_t\right\rbrace_{t\in\left[0,T\right]}$ via $K\subseteq\RR^{d+m}$ such that $\left\lbrace x_t, c_t\right\rbrace_{t\in\left[0,T\right]}\in\widehat{\mathcal{A}}_{X_0}$. Appealing to the results in contributions by \citet{HePears:1991:Consumptionandportfolio}, \cite{cvitanic1992convex}, and \citet{schroder2003optimal} we may complete the baseline market by instituting so-called fictitious assets. 

To that end, we introduce respectively $K$'s support function and the barrier cone:
\begin{equation}\label{eq:sec22suppfuncetc}
\delta\left(\widehat{\nu}_t\right):=\sup_{\left\lbrace x_t\right\rbrace_{t\in\left[0,T\right]}\in\widehat{\mathcal{A}}_{X_0}'}\left(-\widehat{\nu}_t^{\top}x_t\right),\quad\mathrm{and}\quad \widetilde{\mathcal{A}}_{X_0}:=\left\lbrace\widehat{\nu}_t\mid  \delta\left(\widehat{\nu}_t\right)<\infty\right\rbrace
\end{equation}
for all $\eF_t$-measurable endogenous shadow prices $\widehat{\nu}_t\in\DD^{1,2}\left(\left[0,T\right]\right)^{d+m}$, where $\delta:\RR^{d+m}\rightarrow\RR\cup\left\lbrace\infty\right\rbrace$, and $\widehat{\mathcal{A}}_{X_0}'$ represents the set containing all $\left\lbrace x_t\right\rbrace_{t\in\left[0,T\right]}$ that cohere with $\left\lbrace x_t, c_t\right\rbrace_{t\in\left[0,T\right]}\in\widehat{\mathcal{A}}_{X_0}$, i.e. all $\left\lbrace x_t\right\rbrace_{t\in\left[0,T\right]}$ meeting \eqref{eq:sec21assump} for which $x_t\in K$ is true $\mathrm{d}t\otimes\PP\text{-a.e.}$ We postulate that $K$ contains the origin, and additionally define $\mathcal{H}_{\widehat{\mathcal{A}}}:=\big\{ \widehat{\nu}_t\mid \left\Vert\delta\left(\widehat{\nu}_t\right)\right\Vert_{L^2\left(\Omega\times\left[0,T\right]\right)}<\infty\big\}$ given $\widehat{\nu}_t\in\DD^{1,2}\left(\left[0,T\right]\right)^{d+m}$. Notice that $\widehat{\nu}_t\in\mathcal{H}_{\widehat{\mathcal{A}}}$ implies that $\widehat{\nu}_t\in\widetilde{\mathcal{A}}_{X_0}$.

For any $\widehat{\nu}_t\in\mathcal{H}_{\widehat{\mathcal{A}}}$, we complete the asset menu, which constitutes the unconstrained artificial market. For that purpose, introduce the fictitious analogues of $B_t$ and $S_t$:
\begin{equation}\label{eq:sec22s}
\begin{gathered}
\mathrm{d}\widehat{B}_t=\left[R_{f,t}+\delta\left(\widehat{\nu}_t\right)\right]\widehat{B}_t\mathrm{d}t,\ \widehat{B}_0=1,\\
{\mathrm{d}\widehat{S}_t}/{\widehat{S}_t}=\left(\big(R_{f,t}+\delta\left(\widehat{\nu}_t\right)\big)1_{d+m}+\widehat{\sigma}_t\left(\widehat{\lambda}_t+\widehat{\sigma}_t^{-1}\widehat{\nu}_t\right)\right)\mathrm{d}t+\widehat{\sigma}_t\mathrm{d}W_t,\ \widehat{S}_0=1_{d+m}.
\end{gathered}
\end{equation}
In this artificial model, say $\mathcal{M}_{\widehat{\nu}}$, the nominal interest rate reads $R_{f,t}+\delta\left(\widehat{\nu}_t\right)$, and the risk-premium on the $d+m$-dimensional synthetic risky assets $\widehat{S}_t$ changes into $\widehat{\sigma}_t\big(\widehat{\lambda}_t+\widehat{\sigma}_t^{-1}\widehat{\nu}_t\big)$. All assumptions as to the baseline stock $S_t$ in \eqref{eq:sec21stockproc} remain in force here.

Let $\mathcal{E}\left(\cdot\right)$ be the so-called Dol\'{e}ans-Dade exponential. The system of equations for $\widehat{B}_t$ and $\widehat{S}_t$ alludes to an $\mathcal{M}_{\widehat{\nu}}$-specific stochastic deflator process that obeys to
\begin{equation}\label{eq:sec22prickernpert}
Z_T^{\widehat{\nu}}=Z_T\mathcal{E}\left(\left[\widehat{\sigma}_t^{-1}\widehat{\nu}_t\right]^{\top}\right)\exp\left(-\int_0^T\left(\delta\left(\widehat{\nu}_t\right)+\widehat{\lambda}_t^{\top}\widehat{\sigma}_t^{-1}\widehat{\nu}_t\right)\mathrm{d}t\right),
\end{equation}
where we separate the exogenous $Z_T$ from the endogeneity driven by the unspecified $\widehat{\nu}_t$. Clearly, the market prices of risk in $\mathcal{M}_{\widehat{\nu}}$ abide by $\widehat{\lambda}_t+\widehat{\sigma}_t^{-1}\widehat{\nu}_t$. Correspondingly, $Z_T^{\widehat{\nu}}$ separates into $\widehat{B}_T^{-1}{\mathrm{d}\QQ^{\widehat{\nu}}}/{\mathrm{d}\PP}$, where $\QQ^{\widehat{\nu}}\sim\PP$ defines the non-unique pricing measure. For this reason, $\mathrm{d}W^{\QQ^{\widehat{\nu}}}_t=\mathrm{d}W_t+\big(\widehat{\lambda}_t+\widehat{\sigma}_t^{-1}\widehat{\nu}_t\big)\mathrm{d}t$ characterises a $\QQ^{\widehat{\nu}}$-standard Brownian motion. The non-unicity of $\widehat{\nu}_t$ thus signifies its ability to perturb the pricing measure $\QQ^{\widehat{\nu}}$ and the affiliated economic scenarios alongside the stocks' risk premia.\footnote{For $\delta\left(\widehat{\nu}_t\right)\neq 0$ $\mathrm{d}t\otimes\PP\text{-a.e.}$ $\widehat{\nu}_t$ might also affect the $\mathcal{M}_{\widehat{\nu}}$'s interest rate. In the presence of non-traded risk, the latter turns into an equality. Consider for instance \citet{cvitanic1992convex} or \citet{tepla2000optimal}.} 

The dynamic wealth process for finite-horizon investors living in $\mathcal{M}_{\widehat{\nu}}$ ensues as  
\begin{equation}\label{eq:sec21ficdynconstr}
\mathrm{d}X_t^{\widehat{\nu}}=X_t^{\widehat{\nu}}\left[\left(R_{f,t}+\delta\big(\widehat{\nu}_t\big)+\widehat{x}_t^{\top}{\widehat{\sigma}_t}\left(\widehat{\lambda}_t+\widehat{\sigma}_t^{-1}\widehat{\nu}_t\right)\right)\mathrm{d}t+\widehat{x}_t^{\top}\widehat{\sigma}_t\mathrm{d}W_t\right]-\widehat{c}_t\mathrm{d}t
\end{equation}
in which $\widehat{x}_t$ denotes the $\RR^{d+m}$-valued progressively $\eF_t$-measurable proportion of $X_t^{\widehat{\nu}}$ invested in $\widehat{S}_t$, $\widehat{c}_t$ the progressively $\eF_t$-measurable consumption control, and $X_0^{\widehat{\nu}}=X_0\in\RR_+$ the initial endowment. Analogously, the individual selects $\left\lbrace \widehat{x}_t,\widehat{c}_t\right\rbrace_{t\in\left[0,T\right]}$ in an attempt to maximise expected life-time utility. Hence, we invoke the following assumptions 
\begin{equation}\label{eq:sec22assump}
\widehat{x}_t^{\top}\widehat{\sigma}_t{\widehat{\sigma}_t}^{\top}\widehat{x}_t, \left|\widehat{x}^{\top}_t\widehat{\sigma}_t\widehat{\lambda}_{t}\right|\in L^1\left(\left[0,T\right]\right),\ \mathrm{and}\ \ \widehat{c}_t\in L^2\left(\Omega\times\left[0,T\right]\right).
\end{equation}
In the market model $\mathcal{M}_{\widehat{\nu}}$, a pair $\left\lbrace \widehat{x}_t,\widehat{c}_t\right\rbrace_{t\in\left[0,T\right]}$ is admissible if apart from \eqref{eq:sec22assump}, it secures $X_t^{\widehat{\nu}}\geq 0$ $\forall$ $t\in\left[0,T\right]$. Let us denote the class of admissible trading-consumption pairs by $\widehat{\mathcal{A}}^{\widehat{\nu}}_{X_0}$. Admissibility in $\mathcal{M}_{\widehat{\nu}}$ is in general not equivalent to that in the constrained baseline financial market: that is, $\widehat{B}_t$, $\widehat{S}_t$ differ from $B_t$, $S_t$ and $K$ does not act upon $\widehat{\mathcal{A}}^{\widehat{\nu}}_{X_0}$. 

In consideration of the aforementioned disparity, we derive that the discrepancy between artificial and true wealth equals $\mathrm{d}X_t^{\widehat{\nu}}-\mathrm{d}X_t=\left(\delta\left(\widehat{\nu}_t\right)+\widehat{x}_t^{\top}\widehat{\nu}_t\right)\mathrm{d}t\geq 0$ for all $\left\lbrace \widehat{x},\widehat{c}_t\right\rbrace_{t\in\left[0,T\right]}\in\widehat{\mathcal{A}}_{X_0}\subset\widehat{\mathcal{A}}_{X_0}^{\widehat{\nu}}$ such that $X_t^{\widehat{\nu}}\geq X_t$ $\forall$ $t\in\left[0,T\right]$ given that $\widehat{x}_t=x_t$ and $\widehat{c}_t=c_t$.\footnote{We endow the fictitious controls with the hat notation so as to concretise the difference between the baseline and fictitious environments. In \cite{haugh2006evaluating}, $\widehat{x}_t=x_t$ and $\widehat{c}_t=c_t$ is set \textit{a priori}. Let us remark here that $X_t^{\widehat{\nu}}Z_t+\int_0^tc_sZ_s\mathrm{d}s$ is a $\PP$-martingale with respect to $\eF_t$ if, and only if, $\delta\left(\widehat{\nu}_t\right)+\widehat{x}_t^{\top}\widehat{\nu}_t=0$.} Intuitively, under these premises, $\mathcal{M}_{\widehat{\nu}}$ reduces to $\mathcal{M}$ given the next conditions
\begin{equation}\label{eq:sec22assumd}
\left\lbrace\widehat{x}_t\right\rbrace_{t\in\left[0,T\right]}\in\widehat{\mathcal{A}}_{X_0}',\quad\mathrm{and}\quad \delta\left(\widehat{\nu}_t\right)+\widehat{x}_t^{\top}\widehat{\nu}_t=0,
\end{equation}
and contemporaneously leads to optimality in the baseline environment, see Prop. 8.3 in \citet{cvitanic1992convex}. Procuring $\mathcal{M}$ from $\mathcal{M}_{\widehat{\nu}}$ accords with choosing $\widehat{\nu}_t\in\mathcal{H}_{\widehat{\mathcal{A}}}$ in an effort to realise \eqref{eq:sec22assumd}. In the sequel, we mathematically clarify this construction.

\subsection{Finite-horizon Investor Preferences}\label{sec2.3}
The baseline environment $\mathcal{M}$ entertains a finite-horizon investor who disposes of $X_0\in\RR_+$ at $t=0$ and who retires at $t=T$. All through $\left[0,T\right]$, the agent seeks to maximise expected lifetime utility from consumption and expected utility from terminal wealth by holding a continuously rebalanced portfolio. As a result, the individual  solves
\begin{equation}\label{eq:gendynprob}
\begin{aligned}
\sup_{\left\lbrace{x}_t,c_t\right\rbrace_{t\in\left[0,T\right]}\in\widehat{\mathcal{A}}_{X_0}} \ &\EE\left[\int_0^Te^{-\int_0^t\beta_{s}\mathrm{d}s}u\left(c_t,\Pi_t\right)\mathrm{d}t+U\left(X_T,\Pi_T\right)\right]\\ \mathrm{s.t.}\quad &\mathrm{d}X_t=-c_t\mathrm{d}t+X_t\left[\left(R_{f,t}+{x}^{\top}_t\widehat{\sigma}_t\widehat{\lambda}_{t}\right)\mathrm{d}t+{x}_t^{\top}\widehat{\sigma}_t\mathrm{d}{W}_t\right],
\end{aligned}
\end{equation}
for a deterministic process $\left\lbrace\beta_t\right\rbrace_{t\in\left[0,T\right]}$ denoting the investor's time preference, as well as two utility functions $u:\RR_+\times\RR_+\rightarrow\RR_+$ and $U:\RR_+\times\RR_+\rightarrow\RR_+$ both of which incorporate the exogenous commodity price process $\left\lbrace\Pi_t\right\rbrace_{t\in\left[0,T\right]}$ as a benchmark.

Consider $U$. We assume that $U$ is once continuously differentiable, $U\in \mathcal{C}^{1}\left(\RR_+;\RR\right)$, in the $X_T$-direction. It also satisfies the subsequent limiting Inada's conditions
\begin{equation}
\lim_{x\rightarrow \infty}U'_X\left(x,\cdot\right)= 0,\ \ \mathrm{and} \ \ \lim_{x\rightarrow 0}U'_X\left(x,\cdot\right)=\infty,
\end{equation}
wherein we let $U'_X$ and $U''_{XX}$ denote respectively the first and second derivatives of $U$ in the $X_T$-direction. In particular, we postulate $U'_X>0$ and $U''_{XX}<0$. Additionally,
\begin{equation}
\mathrm{AE}\left(U\right)=\limsup_{x\rightarrow\infty}x\frac{U'_X\left(x,y\right)}{U\left(x,y\right)}<1,
\end{equation} 
outlines for any $y\in\RR_+$ the reasonable asymptotic elasticity condition in \citet{kramkov1999asymptotic}. We impose this condition to hold in order to ensure existence of optimal solutions, feasibility of duality modes and finiteness of expectations. 

The inverse of $U'_{X}$ in the $X_T$-direction tallies with $I:\RR_+\times\RR_+\rightarrow\RR$ such that $U'_X\left(I\left(x,y\right),y\right)=x$ for all $y\in\RR_+$. We mandate that $I$ and $U'$ are piecewise continuously differentiable, $I,U'\in\mathcal{PC}\left(\RR_+;\RR\right)$, in both directions of $X_T$ and $\Pi_T$.\footnote{This description is in the mould of \cite{detemple1991asset}, whose generality we enlarge by relying on \citet{lakner2006portfolio} to retain that $I\left(F,G\right)\in\DD^{1,2}$ for $F,G\in\DD^{1,2}$, $I\in\mathcal{PC}\left(\RR_+;\RR\right)$.} Then, define 
\begin{equation}
V\left(x,y\right)=\sup_{z>0}\left(U\left(z,y\right)-xz\right)=U\left(I\left(x,y\right),y\right)-xI\left(x,y\right)
\end{equation}
for $x,y\in\RR_+$ as the convex conjugate of $U$, which is likewise in $\mathcal{PC}\left(\RR_+;\RR\right)$. Ultimately, we suppose that $U'_Y<0$, where $U'_Y$ is $U$'s first derivative in the $\Pi_T$-direction.

Regarding $u$, we enforce the same assumptions as on $U$. Moreover, $u'_X$ and $u''_{XX}$ denote the first and second derivatives of $u$ in the first argument; $u'_Y$ is $u$'s first derivative in the $\Pi_t$-direction. Let $\widetilde{I}:\RR_+\times\RR_+\rightarrow\RR$ be the inverse of $u'_X$ in the $c_t$-direction. Hence,
\begin{equation}
v\left(x,y\right)=e^{-\int_0^t\beta_{s}}u\big(\widetilde{I}\big(e^{\int_0^t\beta_{s}}x,y\big),y\big)-x\widetilde{I}\big(e^{\int_0^t\beta_{s}}x,y\big)
\end{equation}
demarcates the convex conjugate of $e^{-\int_0^t\beta_{s}}u$, resulting from $\sup_{z\in\RR_+}\big(e^{-\int_0^t\beta_{s}}u\left(z,y\right)-xz\big)$, for $x,y\in\RR_+$. We explicitly embed the price index into both utility specifications in the light of motives that relate to real wealth, cf. \citet{brennan2002dynamic}. Note that $\Pi_t$ is sufficiently general; applications including $\RR_+$-valued semi-martingales are trivial. 

\section{Portfolio Choice}\label{sec3}
We advance by probing \eqref{eq:gendynprob}. Here, we classify the problem into two cooperating ones according to the unconstrained market and the entrenched true economy therein. In order to envision the implications of convex duality, we first solve the portfolio choice problem in this unconstrained model. Afterwards, we activate the trading restrictions and apply convex duality to recover the portfolio composition in its constrained counterpart. We finalise by clarifying the ramifications of the duality concepts for non-traded risk. 

\subsection{Allocation in Unconstrained Market}\label{sec3.1}
We first explore \eqref{eq:gendynprob} in the unconstrained market $\mathcal{M}$. Drawing on the martingale method, cf. \citet{pliska1986stochastic}, \citet{karatzas1987optimal}, and \citet{CoxHuang:1989:Optimalconsumptionand,cox1991variational}, we transform the dynamic problem in \eqref{eq:gendynprob} into a static variational formulation as specified by
\begin{equation}\label{eq:sec5finalstat}
\begin{aligned}
\sup_{\left(X_T,c_t\right)\in \widehat{L^2}} &\EE\left[\int_0^T\frac{u\left(c_t,\Pi_t\right)}{e^{\int_0^t\beta_{s}}\mathrm{d}s}\mathrm{d}t+U\left(X_T,\Pi_T\right)\right]\ \  \mathrm{s.t.} \  &\EE\left[\int_0^Tc_tZ_t\mathrm{d}_t+X_TZ_T\right]\leq X_0,
\end{aligned}
\end{equation}
which is identical to the aforementioned dynamic stochastic optimal control problem, where $\widehat{L^2}:=L^2\left(\Omega\right)\times L^2\left(\Omega\times\left[0,T\right]\right)$ and $\left\lbrace Z_t\right\rbrace_{t\in\left[0,T\right]}$ lives by the nominal SPD in \eqref{eq:genkernel}.

In this reformulated static problem, the agent maximises expected utility over all {attainable} or equitably {admissible} contingent claims that involve a continuous stream of coupon payments. Hereinafter, arbitrage arguments inclusive of the self-financing condition implicit in equation \eqref{eq:genactwealth} facilitate us to recover $\left\lbrace x_t\right\rbrace_{t\in\left[0,T\right]}$, which hedges the investor against adverse shifts in the `underlyings' of $X_T$ and $\left\lbrace c_t\right\rbrace_{t\in\left[0,T\right]}$. In particular, denote by $\mathcal{D}^{W}_t$ the Malliavin derivative in the $W_t$-direction, then Theorem \ref{clarkocone} states that
\begin{equation}\label{eq:sec31coform}
X_TZ_T+\int_0^Tc_tZ_t\mathrm{d}t=X_0+\int_0^T\EE\left[\mathcal{D}^W_t\left(X_TZ_T+\int_0^Tc_sZ_s\mathrm{d}s\right)\cond\eF_t\right]^{\top}\mathrm{d}W_t.
\end{equation}

Further, the budget constraint in \eqref{eq:sec5finalstat} ensures that $\left\lbrace X_t\right\rbrace_{t\in\left[0,T\right]}$ remains self-financing. The martingale modes that we exploit in solving \eqref{eq:sec5finalstat} allow us to circumvent the typical abridgement of the generality due to the assumption of a Markovian return structure in the financial model.\footnote{Replicating arguments namely consolidate through the MRT and Clark-Ocone formula that do not enjoin Markovian dynamics. Moreover, $U$ and $u$ satisfy minimal conditions and are state-dependent due to the incorporated $\{\Pi_t\}_{t\in\left[0,T\right]}$. Both aspects considerably enlarge the generality in \citet{bick2013solving}.} By virtue of the static nature of the problem in \eqref{eq:sec5finalstat}, we are then able to resort to ordinary Lagrangian technology on infinite-dimensional Banach spaces to rescue optimal horizon wealth, consumption patterns, and the harmonious portfolio decomposition. Theorem \ref{thm1} embraces the ensuing optimal solutions to \eqref{eq:sec5finalstat}.
\begin{theorem}\label{thm1}
Consider the portfolio choice problem \eqref{eq:sec5finalstat} in the unconstrained market, for an investor with wealth dynamics \eqref{eq:genactwealth}. Then, optimal $c_t$ and $X_T$ materialise into 
\begin{equation}\label{eq:genfictoptwealth}
c_t^{\mathrm{opt}}=\widetilde{I}\left(e^{\int_0^t\beta_{s}}\mathcal{H}^{-1}\left(X_0\right)Z_t,\Pi_t\right),\quad\mathrm{and}\quad X_T^{\mathrm{opt}}=I\left(\mathcal{H}^{-1}\left(X_0\right)Z_T,\Pi_T\right),
\end{equation}
where $X_t^{\mathrm{opt}}=\EE\big[X_TZ_{t,T}+\int_t^Tc_sZ_{t,s}\mathrm{d}s \ \big| \ \eF_t\big]$ and $Z_{t,T}:=Z_t^{-1}Z_T$, for $t\in\left[0,T\right]$. Additionally, $\mathcal{H}^{-1}\left(X_0\right)$ characterises the multiplier, given the decreasing function $\mathcal{H}:\RR_+\rightarrow\RR_+$
\begin{equation}\label{eq:thm1lagrmult1}
\mathcal{H}\left(\eta\right)=\EE\left[\int_0^T\widetilde{I}\left(\eta Z_t,\Pi_t\right)Z_t\mathrm{d}t+I\left(\eta Z_T,\Pi_T\right)Z_T\right]=X_0, \ \eta\in\RR_+
\end{equation}
in which $X_0\in\RR_+$ represents the fixed initial endowment: the budget constraint binds. The congruent optimal portfolio rules originate from the following identity
\begin{equation}\label{eq:thm1optrules}
\begin{aligned}
x_t^{\mathrm{opt}}&=\frac{{\widehat{\sigma}_t^{\top^{-1}}}}{X_t^{\mathrm{opt}}}\Bigg(\widetilde{x}_t^{m,\mathrm{opt}}\widehat{\lambda}_t+\EE\left[Z_t^{-1}\left(\int_t^T{\mathcal{R}^1}_{c,u}^{-1}{\mathcal{G}^1}_{t,u}\mathrm{d}u+{\mathcal{R}^1}_X^{-1}{\mathcal{G}^1}_{t,T}\right)\cond\eF_t\right]\\&+\EE\left[\int_t^T\left(1-{\mathcal{R}^2}_{c,u}^{-1}\right){\mathcal{G}^2}_{t,u}{c}_{u}^{\mathrm{opt}}Z_{t,u}\mathrm{d}u+\left(1-{\mathcal{R}^2}_X^{-1}\right){\mathcal{G}^2}_{t,T}X_T^{\mathrm{opt}}Z_{t,T}\cond\eF_t\right]\Bigg),
\end{aligned}
\end{equation}
which is extricable by virtue of the MRT and uniqueness of the integrand in \eqref{eq:sec31coform}.
\end{theorem}
\begin{proof}
Define $\eL\big(X_T,\left\lbrace c_t\right\rbrace, \eta\big)=\EE\big[\int_0^Te^{-\int_0^t\beta_{s}}u\big(c_t,\Pi_t\big)+U\big({X_T},{\Pi_T}\big)-\eta \big(\int_0^Tc_tZ_t\mathrm{d}t+X_TZ_T^{\widehat{\lambda_2}}-X_0\big)\big]$ as the Lagrangian, for some $\eta\in\RR_+$. Then, $\eL:L^2\left(\Omega\right)\times L^2\left(\Omega\times\left[0,T\right]\right)\times\RR_+\rightarrow\RR$. From here, let $\xi_1\in L^2\left(\Omega\right)$, $\xi_2\in L^2\left(\Omega\times\left[0,T\right]\right)$ along with $\xi_3\in\RR_+$. We derive
\begin{equation}
\begin{aligned}
D_{X_T}\eL\left(X_T,\left\lbrace c_t\right\rbrace,\eta\right)\xi_1&=\Big\langle U'_{X}\left(X_T,\Pi_T\right)-\eta Z_T,\xi_1\Big\rangle_{L^2\left(\Omega\right)}=0,\\
D_{\left\lbrace c_t\right\rbrace}\eL\left(X_T,\left\lbrace c_t\right\rbrace,\eta\right)\left\lbrace \xi_2\right\rbrace&=\left\langle e^{-\int_0^t\beta_{s}}u'_X\left(c_t,\Pi_t\right)-\eta Z_t,\xi_2\right\rangle_{L^2\left(\Omega\times\left[0,T\right]\right)}=0,
\end{aligned}
\end{equation}
and $D_{\eta}\eL\left(X_T,\left\lbrace c_t\right\rbrace,\eta\right)\xi_3=\big\langle X_TZ_T+\int_0^Tc_tZ_t\mathrm{d}t-X_0,\xi_3\big\rangle_{L^2\left(\Omega\right)}=0$, where $D_{X_T}\eL:L^2\left(\Omega\right)\rightarrow\RR$, $D_{\left\lbrace c_t\right\rbrace}\eL:L^2\left(\Omega\times\left[0,T\right]\right)\rightarrow\RR$, and $D_{\eta}\eL:\RR_+\rightarrow\RR$, cf. Definition \ref{frechet}.\footnote{Here, $D_{X_T}\eL$ is bounded linear given that $\| U'_{X}\left(X_T,\Pi_T\right)\|_{L^2\left(\Omega\right)}<\infty$ holds. Any $X_T=I\left(W_T,\Pi_T\right)$ for all $W_T\in L^2\left(\Omega\right)$ satisfies this condition; $X_T^{\mathrm{opt}}\in L^2\left(\Omega\right)$ hence obeys. This carries over to $u'_X\left(c_t,\Pi_t\right)$.} Hence,
\begin{equation}\label{eq:genwealthoptstati}
U'_{X}\left(X_T,\Pi_T\right)-\eta Z_T=0, \quad\mathrm{and}\quad e^{-\int_0^t\beta_{s}}u'_X\left(c_t,\Pi_t\right)-\eta Z_t=0.
\end{equation}

Note that $D_{X_T,X_T}\eL\left\lbrace\xi_1,\xi_1\right\rbrace=\left\langle U''_{XX}\left(X_T,\Pi_T\right),\xi_1^2\right\rangle_{L^2\left(\Omega\right)}<0$ and $D_{\left\lbrace c_t\right\rbrace,\left\lbrace c_t\right\rbrace}\eL\left\lbrace\xi_2,\xi_2\right\rbrace=\left\langle e^{-\int_0^t\beta_{s}}u''_{\left\lbrace c_t\right\rbrace,\left\lbrace c_t\right\rbrace}\left(c_t,\Pi_t\right),\xi_2^2\right\rangle_{L^2\left(\Omega\times\left[0,T\right]\right)}<0$ $\forall$ $\xi\in L^2\left(\Omega\right)$, $\xi_2\in L^2\left(\Omega\times\left[0,T\right]\right)$ or
\begin{equation}\label{eq:sec31thm1ineqsyst}
\begin{aligned}
\EE\left[U\left(X_T,\Pi_T\right)\right]&\leq \EE\left[V\left(\eta Z_T,\Pi_T\right)+\eta Z_TX_T\right]\\&=\EE\left[U\left(I\left(\eta Z_T,\Pi_T\right),\Pi_T\right)-\eta \left(I\left(\eta Z_T,\Pi_T\right)Z_T+X_TZ_T\right)\right]\\&\leq \EE\left[U\left(I\left(\mathcal{H}^{-1}\left(X_0\right) Z_T,\Pi_T\right),\Pi_T\right)\right]=\EE\left[U\left(X_T^{\mathrm{opt}},\Pi_T\right)\right],
\end{aligned}
\end{equation}
and analogously $\EE\left[\int_0^Te^{-\int_0^t\beta_s\mathrm{d}s}u\left(c_t,\Pi_t\right)\right]\leq \EE\left[\int_0^Te^{-\int_0^t\beta_s\mathrm{d}s}u\left(c_t^{\mathrm{opt}},\Pi_t\right)\right]$ $\forall$ $\left(X_T,c_t\right)\in\widehat{L^2}$, which verifies optimality of the pair $\left(X_T^{\mathrm{opt}},c_t^{\mathrm{opt}}\right)$ inherent in \eqref{eq:genwealthoptstati}.

Let us then advance by deriving the optimal portfolio rules. From the MRT, 
\begin{equation}
\left(\left(x_t^{\top}\widehat{\sigma}_t-\widehat{\lambda}_t^{\top}\right)X_t^{\mathrm{opt}}Z_t\right)^{\top}=\EE\left[\mathcal{D}^W_t\left(X_T^{\mathrm{opt}}Z_T+\int_t^Tc_s^{\mathrm{opt}}Z_s\mathrm{d}s\right)\cond\eF_t\right]
\end{equation}
must hold, because $X_T^{\mathrm{opt}}Z_T^{\lambda_2}+\int_0^Tc_t^{\mathrm{opt}}Z_t\mathrm{d}t\in\DD^{1,2}\subset L^2\left(\Omega\right)$ is a $\PP$-martingale. Define
\begin{equation}\label{eq:thm1hedgcoeff}
\begin{aligned}
{\mathcal{G}^{1}}_{t,u}&:=\int_t^u\left(\mathcal{D}^W_t\pi_s-\mathcal{D}^W_t\widehat{\xi}_s\cdot \widehat{\xi}_s\right)\mathrm{d}s+\int_t^u\mathcal{D}^W_t\widehat{\xi}_s\mathrm{d}W_s+\widehat{\xi}_t,\ u\geq t\\ {\mathcal{G}^2}_{t,u}&:=-\int_t^u\mathcal{D}^{W}_t\left(R_{f,s}+\mathcal{D}^W_t\widehat{\lambda}_s\cdot\widehat{\lambda}_s\right)\mathrm{d}s-\int_t^u\mathcal{D}^W_t\widehat{\lambda}_s\mathrm{d}W_s, \ u\geq t,
\end{aligned}
\end{equation}
for $\mathcal{D}^W_tR_{f,s}=\mathcal{D}^W_tr_s+\mathcal{D}^W_t\pi_s-\mathcal{D}^W_t\widehat{\xi}_t^{\top}\widehat{\lambda}_t$, with ${\mathcal{R}^2}_{c,u}^{-1}=-{u'_X\left(c_u^{\mathrm{opt}},\Pi_u\right)}/{u''_{XX}\left(c_u^{\mathrm{opt}},\Pi_u\right)}\frac{1}{c_u^{\mathrm{opt}}}$, ${\mathcal{R}^2}_{X}^{-1}=-{U'_X\left(X_T^{\mathrm{opt}},\Pi_T\right)}/{U''_{XX}\left(X_T^{\mathrm{opt}},\Pi_T\right)}{X_T^{\mathrm{opt}}}^{-1}$. The optimal $x_t^{\mathrm{opt}}$ emerge given  
\begin{equation}
\widetilde{x}^{m,\mathrm{opt}}_t=\EE\left[\left(\int_t^T{\mathcal{R}^2}_{c,u}^{-1}{c}_{u}^{\mathrm{opt}}Z_{t,u}\mathrm{d}u+{\mathcal{R}^2}_X^{-1}X_T^{\mathrm{opt}}Z_{t,T}\right) \cond\eF_t\right]
\end{equation} with ${\mathcal{R}^1}_{c,u}^{-1}=\widetilde{I}'_Y\big(e^{\int_0^u\beta_s\mathrm{d}s}\mathcal{H}^{-1}\left(X_0\right)Z_u,\Pi_u\big){M_u}$ and ${\mathcal{R}^1}_{X}^{-1}=I'_Y\big(\mathcal{H}^{-1}\left(X_0\right)Z_T,\Pi_T\big){M_T}$.
\end{proof}

Considering the unconstrained $\mathcal{M}$ as the true economy, we break down Theorem \ref{thm1}. To erase the semblance of superfluity concerning the formulae in \eqref{eq:genfictoptwealth}, we note that this theorem sheds light on the solution technique embroiling infinite-dimensional optimisation. In applications of convex duality, this technique plays a prominent role. The formulae themselves constitute celebrated identities to which a copious strand of the literature has supplied economically intuitive implications; consider e.g. \citet{karatzas1998methods}. The expression for ${x}^{\mathrm{opt}}_t$ germinates after deriving $\left\lbrace X_t^{\mathrm{opt}},c_t^{\mathrm{opt}}\right\rbrace_{t\in\left[0,T\right]}$ from uniqueness of the integrand in \eqref{eq:sec31coform}, which dictates given $\mathcal{X}_{t,T}^{\mathrm{opt}}:=B_T^{-1}X_T^{\mathrm{opt}}+\int_t^TB_u^{-1}c_u^{\mathrm{opt}}\mathrm{d}u$ that\footnote{We employ the fact that $Z_T=B_T^{-1}\frac{\mathrm{d}\QQ}{\mathrm{d}\PP}$ for a unique equivalent martingale measure $\QQ\sim\PP$, such that $\mathrm{d}W^{\QQ}_t=\mathrm{d}W_t+\widehat{\lambda}_t\mathrm{d}t$ is the stochastic differential equation for a standard Brownian motion under $\QQ$. The displayed equality thus disentangles and clarifies the effect of the pricing measure $\QQ$ on $x_t^{\mathrm{opt}}$.}
\begin{equation}\label{eq:sec32explicdepq}
x_t^{\mathrm{opt}}=\left(B_t^{-1}X_t^{\mathrm{opt}}\widehat{\sigma}_t^{\top}\right)\EE^{\QQ}\left[\mathcal{D}^{W}_t\mathcal{X}_{t,T}^{\mathrm{opt}}-\mathcal{X}_{0,T}^{\mathrm{opt}}\int_t^T\mathcal{D}^W_t\widehat{\lambda}_s\mathrm{d}W^{\QQ}_s\cond\eF_t\right].
\end{equation}

Specifically, ${x}^{\mathrm{opt}}_t$ encases three distinct decisions that unfold themselves by computations along the lines of the decomposition manoeuvres in \citet{detemple2009dynamic}. The separation of these hedging demands submits to ${x}_t^{\mathrm{opt}}={x}_t^{\mathrm{m}}+{x}_t^Z+{x}_t^{\Pi}$, where the first rule concerns the tangency mean-variance efficient portfolio ${x}_t^{\mathrm{m}}={X_t^{\mathrm{opt}}}^{-1}\widetilde{x}_t^{m,\mathrm{opt}}\widehat{\sigma}_t^{\top^{-1}}\widehat{\lambda}_t$. See \citet{Merton:69} for like weights in line with CRRA preferences. The remaining two demands are respectively price index and nominal deflator hedges that answer to
\begin{equation}\label{eq:sec31decomp}
\begin{aligned}
{x}^{\Pi,\mathrm{opt}}&={\widehat{\sigma}_t^{\top^{-1}}}{X_t^{\mathrm{opt}}}^{-1}\EE\left[\int_t^T{\mathcal{R}^1}_{c,u}^{-1}{\mathcal{G}^1}_{t,u}Z_t^{-1}\mathrm{d}u+{\mathcal{R}^1}_X^{-1}{\mathcal{G}^1}_{t,T}Z_t^{-1}\cond\eF_t\right] \\ {x}^{Z,\mathrm{opt}}&=\frac{\widehat{\sigma}_t^{\top^{-1}}}{X_t^{\mathrm{opt}}}\EE\left[\int_t^T\left(1-{\mathcal{R}^2}_{c,u}^{-1}\right){\mathcal{G}^2}_{t,u}\frac{{c}_{u}^{\mathrm{opt}}}{Z_{u,t}}\mathrm{d}u+\left(1-{\mathcal{R}^2}_X^{-1}\right){\mathcal{G}^2}_{t,T}\frac{X_T^{\mathrm{opt}}}{Z_{T,t}}\cond\eF_t\right].
\end{aligned}
\end{equation}
The expressions in \eqref{eq:thm1hedgcoeff} provoke the hedging coefficients within these portfolio weights, whereas the `$\mathcal{R}$' terms are interdependent with quantities that categorise relative risk aversion (RRA). By partitioning the former identities into subgroups that share the same intrinsic character, we could disentangle more specific hedging motivations such as shadow price hedges or interest rate hedges, cf. \cite{detemple2014portfolio} for akin rules. 

As concerns the RRA factors, let us exemplify the matter for a ratio CRRA agent whose preferences obey $U\left(x,y\right)=\left(1-\gamma\right)^{-1}\left(\left(x/y\right)^{1-\gamma}-1\right)$ for $x,y\in\RR_+$, $\gamma>1$. Consequently,
\begin{equation}
-\frac{U'_X\left(X_T^{\mathrm{opt}},\Pi_T\right)}{U''_{XX}\left(X_T^{\mathrm{opt}},\Pi_T\right)}{X_T^{\mathrm{opt}}}^{-1}=\frac{1}{\gamma},\quad\mathrm{and}\quad I_{Y}'\big(\eta^{\mathrm{opt}} Z_T,\Pi_T\big)\frac{M_T}{X_T^{\mathrm{opt}}}=\left(1-\frac{1}{\gamma}\right)Z_T,
\end{equation}
where $\eta^{\mathrm{opt}}:=\mathcal{H}^{-1}\left(X_0\right)$, which infers that ${\mathcal{R}^1}_X^{-1}=\left(1-1/\gamma\right)X_T^{\mathrm{opt}}Z_T$ and ${\mathcal{R}^2}_X^{-1}={1}/{\gamma}$. Presuming that utility from consumption respects an identical qualification, we straightforwardly find that ${\mathcal{R}^1}_{c,u}^{-1}=\left(1-1/\gamma\right)c_u^{\mathrm{opt}}Z_u$ and ${\mathcal{R}^2}_{c,u}^{-1}={1}/{\gamma}$. In both instances, we observe the similarities between the RRA structures, ${\mathcal{R}^1}_X^{-1}=\big(1-{\mathcal{R}^2}_X^{-1}\big)X_T^{\mathrm{opt}}Z_T$ and ${\mathcal{R}^1}_{c,u}^{-1}=\big(1-{\mathcal{R}^2}_{c,u}^{-1}\big)c_u^{\mathrm{opt}}Z_u$. There is no guarantee that this proximity surpasses this example; we generalise it as part of the demand accompanying inflation risk.

Assuming that $U$ and $u$ are such that ${\mathcal{R}^1}_X^{-1}=\big(1-{\mathcal{R}^2}_X^{-1}\big)X_T^{\mathrm{opt}}Z_T$ and ${\mathcal{R}^1}_{c,u}^{-1}=\big(1-{\mathcal{R}^2}_{c,u}^{-1}\big)c_u^{\mathrm{opt}}Z_u$ hold, the optimal portfolio composition $x_t^{\mathrm{opt}}$ in \eqref{eq:thm1optrules} transforms into  
\begin{equation}\label{eq:sec31twofsep}
\begin{aligned}
x_t^{\mathrm{opt}}&=\frac{{\widehat{\sigma}_t^{\top^{-1}}}}{X_t^{\mathrm{opt}}}\Bigg(\EE\left[\left(\int_t^T{\mathcal{R}^2}_{c,u}^{-1}{c}_{u}^{\mathrm{opt}}Z_{t,u}\widehat{\lambda}_t\mathrm{d}u+{\mathcal{R}^2}_X^{-1}X_T^{\mathrm{opt}}Z_{t,T}\widehat{\lambda}_t\right) \cond\eF_t\right]\\&+\EE\left[\int_t^T\left(1-{\mathcal{R}^2}_{c,u}^{-1}\right)\sum_{i=1}^2{\mathcal{G}^i}_{t,u}\frac{{c}_{u}^{\mathrm{opt}}}{Z_{u,t}}\mathrm{d}u+\left(1-{\mathcal{R}^2}_X^{-1}\right)\sum_{i=1}^2{\mathcal{G}^i}_{t,T}\frac{X_T^{\mathrm{opt}}}{Z_{T,t}}\cond\eF_t\right]\Bigg),
\end{aligned}
\end{equation}
wherein the two-fund separation principle for $c_t$'s RRA and $X_T$'s RRA in isolation is preponderantly visible. Namely, we distinguish the two mutual funds on the basis of the the detached RRA fractions: ${\mathcal{R}^1}_{c,u}^{-1}$ and $1-{\mathcal{R}^2}_{c,u}^{-1}$ with respect to $c_t^{\mathrm{opt}}$, ${\mathcal{R}^2}_{X}^{-1}$ and $1-{\mathcal{R}^2}_{X}^{-1}$ regarding $X_T^{\mathrm{opt}}$. The requirements imposed on $u$ and $U$ under which the separation precept holds apply for instance to dual CRRA stipulations. In \citet{brennan2002dynamic}, the authors unravel this phenomenon in detail for a utility-maximizing CRRA investor.\footnote{The concepts underscoring the solutions and techniques in the fictitious $\mathcal{M}_{\widehat{\nu}}$ are of the same kind as those in the unconstrained baseline $\mathcal{M}$. Comprehending the latter is of necessity for the former.}

\subsection{Convex Duality Application}\label{sec3.2}
This section applies convex martingale duality to the dynamic asset allocation problem in $\mathcal{M}$, wherein we constrain $\left\lbrace x_t\right\rbrace_{t\in\left[0,T\right]}$ to $\widehat{\mathcal{A}}_{X_0}'$. We assemble the duality techniques in a style that is to a great extent inspired by \citet{rogers2001duality}, for the sake of merging the primal with the dual. We relegate explicit optimality statements to the next section. 

Now, consider an agent in the baseline economic environment $\mathcal{M}$ from section \ref{sec2.1}, whose wealth dynamics accommodate general investment restrictions through $K$. Therefore,
\begin{equation}\label{eq:gendynprobdfdafdaf}
\begin{aligned}
\sup_{\left\lbrace{x}_t,c_t\right\rbrace_{t\in\left[0,T\right]}\in\widehat{\mathcal{A}}_{X_0}} \ &\EE\left[\int_0^Te^{-\int_0^t\beta_{s}\mathrm{d}s}u\left(c_t,\Pi_t\right)\mathrm{d}t+U\left(X_T,\Pi_T\right)\right]\\ \mathrm{s.t.}\quad &\mathrm{d}X_t=-c_t\mathrm{d}t+X_t\left[\left(R_{f,t}+{x}^{\top}_t\widehat{\sigma}_t\widehat{\lambda}_{t}\right)\mathrm{d}t+{x}_t^{\top}\widehat{\sigma}_t\mathrm{d}{W}_t\right],
\end{aligned}
\end{equation}
for given $X_0\in\RR_+$ details the constrained portfolio choice problem. The martingale method fails to render a static problem tantamount to \eqref{eq:sec5finalstat}, due to the presence of these constraints. Notwithstanding, after applying  duality modes, the resulting dual implies the fictitious market $\mathcal{M}_{\widehat{\nu}}$, in which a static unconstrained problem surfaces.  

To facilitate applications of convex duality, we rework the dynamic budget constraint:
\begin{equation}\label{eq:sec32rewrconstr}
X_TB_T^{-1}=X_0-\int_0^Tc_tB_t^{-1}\mathrm{d}t+\int_0^TX_tB_t^{-1}{x}^{\top}_t\widehat{\sigma}_t\left(\mathrm{d}{W}_t+\widehat{\lambda}_{t}\mathrm{d}t\right)
\end{equation}
The exogenous money market account $B_t$ in this num\'{e}raire-based approach does not influence the optimality conditions, cf. the monograph by \citet{rogers2013optimal}. In the wake of this study, we then introduce a strictly positive semi-martingale process $\left\lbrace Y_t\right\rbrace_{t\in\left[0,T\right]}$, for some $Y_0\in\RR_+$, that serves as a state-wise Lagrange multiplier on the dynamic relative wealth equation \eqref{eq:sec32rewrconstr} in the constrained portfolio problem. For the sake of clarity regarding this procedure, let $\mathrm{d}Y_t=Y_t\big[\alpha_{Y,t}\mathrm{d}t+\widehat{\beta}_t\mathrm{d}W_t\big]$. Then, $\EE[X_TB_T^{-1}Y_T]$ must be equal to
\begin{equation}\label{eq:sec32newconstr}
X_0Y_0+\left\langle B_t^{-1}Y_t,\left({x}^{\top}_t\widehat{\sigma}_t\left(\widehat{\lambda}_{t}+\widehat{\beta}_t\right)+\alpha_{Y,t}\right)X_t-c_t \right\rangle_{L^2\left(\Omega\times\left[0,T\right]\right)},
\end{equation}
in consideration of integration by parts results, where the square-integrable stochastic integral $\int_0^T\big(\widehat{\beta}_t^{\top}+x_t^{\top}\widehat{\sigma}_t\big)\mathrm{d}W_t$ vanishes in expectation. In formulating the Lagrangian consonant with \eqref{eq:gendynprobdfdafdaf}, we may embed the surrogate constraint in \eqref{eq:sec32newconstr}. Complementary slackness (CS) then poses restrictions on the drift and diffusion coefficients of $Y_t$. The dual problem in turn transpires and involves minimisation over $\left\lbrace Y_t\right\rbrace_{t\in\left[0,T\right]}$.\footnote{In the next Theorem \ref{thm2}, the semi-martingale $Y_t$ finds its technically precise definition. Observe that if we pursue along these lines in the unconstrained environment $\mathcal{M}$, we obtain $Y_t=\frac{\mathrm{d}\QQ}{\mathrm{d}\PP}\big|_{\eF_t}\mathcal{H}^{-1}\left(X_0\right)$.} 

Alternatively, we could enforce $Y_T$ on the rewritten dynamic constraint \eqref{eq:sec32rewrconstr} as a Lagrange multiplier in the ordinary fashion. As a result, we arrive in expectation at
\begin{equation}
\left\langle Y_T,\delta\left({x}_t^{\top}\widehat{\sigma}_t\widehat{\lambda}_tX_tB_t^{-1}\right)+X_0\right\rangle_{L^2\left(\Omega\right)}+\left\langle Y_TB_t^{-1},-c_t+{x}_t^{\top}\widehat{\sigma}_t\widehat{\lambda}_tX_t\right\rangle_{L^2\left(\Omega\times\left[0,T\right]\right)},
\end{equation}
which we may rework by Skorokhod's duality result, cf. Theorem \ref{skorokhod}. According to this theorem, we obtain $\EE\big[Y_T\int_0^TX_tB_t^{-1}x_t^{\top}\widehat{\sigma}_t\mathrm{d}W_t\big]=\EE\big[\int_0^TX_tB_t^{-1}x_t^{\top}\widehat{\sigma}_t\EE[Y_T\mathcal{D}^W_t\log(Y_T)\ \big| \ \eF_t \big]\mathrm{d}t\big]$, reliant on the Hermitian adjoint operator of the Malliavin kernel. This approach is effectively identical to the techniques applied in \citet{rogers2001duality} as well as in \eqref{eq:sec32newconstr}, and emphasises the integration by parts recipe. Let us then turn to the duality results based on this Lagrange multiplier procedure, as subsequently set out in Theorem \ref{thm2}. 
\begin{theorem}\label{thm2}
Consider the constrained dynamic allocation problem in \eqref{eq:gendynprobdfdafdaf}. Introduce
\begin{equation}
\mathrm{d}Y_t=Y_t\left[\alpha_{Y,t}\mathrm{d}t+\beta_{1,t}^{\top}\mathrm{d}W^1_t+\beta_{2,t}^{\top}\mathrm{d}W^2_t\right],
\end{equation}
for an unspecified $Y_0\in\RR_+$, where $\left\Vert\beta_{1,t}\right\Vert_{L^1\left(\left[0,T\right]\right)^d},\left\Vert\beta_{2,t}\right\Vert_{L^1\left(\left[0,T\right]\right)^m}\in\DD^{1,2}$. In addition, let $\widehat{\beta}_t=\left[\beta_{1,t},\beta_{2,t}\right]^{\top}$ and define the drift term $\alpha_{Y,t}\in\DD^{1,2}\left(\left[0,T\right]\right)$ comprised within $Y_t$ as  
\begin{equation}\label{eq:driftdiffmultprocsthm7}
-{\alpha_{Y,t}}=\delta\Big(\widehat{\beta}_t\Big):=\sup_{\left\lbrace{x}_t\right\rbrace_{t\in\left[0,T\right]}\in\widehat{\mathcal{A}}_{X_0}'}\Big(\big(\widehat{\lambda}_t+\widehat{\beta}_t\big)^{\top}\widehat{\sigma}_t\widehat{x}_t\Big),\ \ \mathrm{\&} \ \ \widetilde{\mathcal{A}}^P_{X_0}:=\Big\{ \widehat{\beta}_t\mid \delta\left(\widehat{\beta}_t\right)<\infty\Big\}.
\end{equation}
Let $\mathcal{H}^P_{\widehat{\mathcal{A}}}:=\big\{ \widehat{\beta}_t \mid \|\delta(\widehat{\beta}_t)^{1/2}\|_{L^2\left(\Omega\times\left[0,T\right]\right)}^2<\infty\big\}$ and $n=m+d$, the dual then adheres to 
\begin{equation}\label{eq:sec5thm7generaldual}
\inf_{\widehat{\beta}_t\in\DD^{1,2}\left(\left[0,T\right]\right)^{n},Y_0\in\RR_+}\EE\left[\int_0^Tv\big(Y_tB_t^{-1},\Pi_t\big)\mathrm{d}t+V\big(Y_TB_T^{-1},\Pi_T\big)\right]+X_0Y_0\ \ \mathrm{s.t.} \ \ \widehat{\beta}_t\in \mathcal{H}^P_{\widehat{\mathcal{A}}},
\end{equation}
where $v,V:\RR_+\times\RR_+\rightarrow\RR$ denote the convex conjugates of $e^{-\int_0^t\beta_{s}\mathrm{d}s}u$ and $U$. Suppose that $\widehat{\beta}^{\mathrm{opt}}_t$ solves \eqref{eq:sec5thm7generaldual} and $\delta\big(\widehat{\beta}_t^{\mathrm{opt}}\big)=\big(\widehat{\lambda}_t+\widehat{\beta}_t^{\mathrm{opt}}\big)^{\top}\widehat{\sigma}_t{x}_t^{\mathrm{opt}}$. Optimal $X_T$ and $c_t$ then conform to
\begin{equation}\label{eq:thm2foc1}
 c_t^{\mathrm{opt}}=\widetilde{I}\left(e^{\int_0^t\beta_{s}\mathrm{d}s}Y_t^{\mathrm{opt}}B_t^{-1},\Pi_t\right)\quad\mathrm{and}\quad X_T^{\mathrm{opt}}=I\left(Y_T^{\mathrm{opt}}B_T^{-1},\Pi_T\right),
\end{equation}
in which the optimal Lagrange multiplier is obtainable from the next equality
\begin{equation}\label{eq:thm2foc2}
\widehat{\mathcal{H}}\left(Y_0^{\mathrm{opt}}\right)={Y_0^{\mathrm{opt}}}^{-1}\EE\left[\int_0^Tc_t^{\mathrm{opt}}\frac{Y_t^{\mathrm{opt}}}{B_t}\mathrm{d}t+X_T^{\mathrm{opt}}\frac{Y_T^{\mathrm{opt}}}{B_T}\right]=X_0
\end{equation}  
such that $Y_0^{\mathrm{opt}}=\widehat{\mathcal{H}}^{-1}\left(X_0\right)$. The corresponding optimal portfolio decisions ${x}_t^{\mathrm{opt}}$ consequently decompose into $\widehat{x}_t^{\mathrm{m},\mathrm{opt}}+\widehat{x}_t^{Z,\mathrm{opt}}+\widehat{x}_t^{\Pi,\mathrm{opt}}$, for $\widehat{x}_t^{\mathrm{m},\mathrm{opt}}=-\frac{1}{\widehat{\mathcal{R}}_{f,1,t}}{X_t^{\mathrm{opt}}}^{-1}\widehat{\sigma}_t^{\top^{-1}}\widehat{\beta}_{t}^{\mathrm{opt}}$ including
\begin{equation}
\begin{aligned}
\widehat{x}^{\Pi,\mathrm{opt}}&=\frac{\widehat{\sigma}_t^{\top^{{-1}}}}{X_t^{\mathrm{opt}}}\EE\left[\int_t^T\widehat{\mathcal{R}^1}_{c,u}^{-1}\widehat{\mathcal{G}^1}_{t,u}B_t{Y_t^{\mathrm{opt}}}^{-1}\mathrm{d}u+\widehat{\mathcal{R}^1}_X^{-1}\widehat{\mathcal{G}^1}_{t,T}B_t{Y_t^{-1}}^{\mathrm{opt}}\cond\eF_t\right] \\ \widehat{x}^{Z,\mathrm{opt}}&=\widehat{\sigma}_t^{\top^{-1}}\EE\left[\bigintssss_t^T\left(\widehat{\mathcal{G}^2}_{t,u}-\frac{\widehat{\mathcal{G}^2}_{t,u}}{\widehat{\mathcal{R}^2}_{c,u}}\right)\frac{\ddot{c}_{t,u}^{\mathrm{opt}}}{X_t^{\mathrm{opt}}}\mathrm{d}u+\left(\widehat{\mathcal{G}^2}_{t,T}-\frac{\widehat{\mathcal{G}^2}_{t,T}}{\widehat{\mathcal{R}^2}_X}\right)\ddot{X}_{t,T}^{\mathrm{opt}}\cond\eF_t\right]
\end{aligned}
\end{equation}
where $\widehat{R}_{f,t,T}^{-1}=-\EE\Big[\Big(\int_t^T\widehat{\mathcal{R}^2}_{c,u}^{-1}c_u^{\mathrm{opt}}\frac{Y_{t,u}^{\mathrm{opt}}}{B_{t,u}}\mathrm{d}u+\widehat{\mathcal{R}^2}_X^{-1}X_T^{\mathrm{opt}}\frac{Y_{t,T}^{\mathrm{opt}}}{B_{t,T}}\Big) \ \Big| \ \eF_t\Big]$, $\ddot{c}_{t,u}^{\mathrm{opt}}=c_u^{\mathrm{opt}}\frac{Y_{t,u}^{\mathrm{opt}}}{B_{t,u}}$, $\ddot{X}_{t,T}^{\mathrm{opt}}=\frac{1}{X_t^{\mathrm{opt}}}X_T^{\mathrm{opt}}\frac{Y_{t,T}^{\mathrm{opt}}}{B_{t,T}}$. Lastly, $X_t^{\mathrm{opt}}=\EE\Big[\int_t^Tc_u^{\mathrm{opt}}\frac{Y_{t,u}^{\mathrm{opt}}}{B_{t,u}}\mathrm{d}u+X_T^{\mathrm{opt}}\frac{Y_{t,T}^{\mathrm{opt}}}{B_{t,T}} \ \Big| \ \eF_t\Big]$, and the subsequent
\begin{equation}\label{eq:thmalternativcsconditions}
\sup_{\left\lbrace{x}_t\right\rbrace_{t\in\left[0,T\right]}\in\widehat{\mathcal{A}}'_{X_0}}\int_0^T\left(\EE\left[\widehat{S}_T{Y_T}{B_T}^{-1}\cond\eF_t\right]-\EE\left[Y_T\cond\eF_t\right]{{S}_t}{B_t}^{-1}\right)^{\top}\mathrm{d}X_t{x}_t\odot{S}_t^{-1}<\infty.
\end{equation}
holds, implying that $Y_0^{-1}B_T^{-1}Y_T$ induces the SPD in the constrained baseline market $\mathcal{M}$. 
\end{theorem}
\begin{proof}
We follow \citet{KleinRoger:2007:Dualityinoptimal} in applying convexity techniques. In this regard, let us introduce the Lagrangian for problem \eqref{eq:gendynprobdfdafdaf} and rewrite it by It\^{o}'s Lemma:
\begin{equation}
\begin{aligned}
\eL&=\EE\bigg[\int_0^T\frac{u\left(c_t,\Pi_t\right)}{e^{\int_0^t\beta_{s}\mathrm{d}s}}\mathrm{d}t+U\left(X_T,\Pi_T\right)+X_0Y_0-X_T{Y_T}{B_T}^{-1}-\int_0^Tc_t{Y_t}{B_t}^{-1}\mathrm{d}t\\&+\int_0^T\left\lbrace\left({x}_t^{\top}\widehat{\sigma}_t\widehat{\lambda}_t+\alpha_{Y,t}+{x}_t^{\top}\widehat{\sigma}_t\beta_t\right)X_t{Y_t}{B_t}^{-1}\mathrm{d}t+\left({x}_t^{\top}\widehat{\sigma}_t+\widehat{\beta}_t\right)X_t{Y_t}{B_t}^{-1}\mathrm{d}W_t\right\rbrace\bigg].
\end{aligned}
\end{equation}

Minimisation in the $X_T$-direction and in the $\left\lbrace c_t\right\rbrace_{t\in\left[0,T\right]}$-direction of $\eL$ yields\footnote{We refer to the proof of Theorem \ref{thm1} for analogous identities (FOCs) involving Fr\'{e}chet derivatives}
\begin{equation}
\begin{aligned}
D_{X_T}\eL\zeta&=\EE\left[\left(U'_X\left(X_T,\Pi_T\right)-Y_TB_T^{-1}\right)\right]\zeta=0 \\ D_{\left\lbrace c_t\right\rbrace}\eL\big\{\widetilde{\zeta}_t\big\}&=\EE\left[\int_0^T\left(e^{-\int_0^t\beta_{s}\mathrm{d}s}
u'_X\left(c_t,\Pi_t\right)-Y_tB_t^{-1}\right)\widetilde{\zeta}_t\mathrm{d}t\right]=0,
\end{aligned} 
\end{equation} 
for all $\zeta\in \RR_+$ and $\widetilde{\zeta}_t\in L^2\left(\Omega\times\left[0,T\right]\right)$, which procure the optimality conditions in \eqref{eq:thm2foc1}.

Inserting $c_t^{\mathrm{opt}}$ and $X_T^{\mathrm{opt}}$ back into $\eL$, we acquire the dual objective function up to
\begin{equation}
\begin{aligned}
\sup_{\left\lbrace{x}_t\right\rbrace_{t\in\left[0,T\right]}\in\widehat{\mathcal{A}}_{X_0}'}\EE\Big[\int_0^TX_t{Y_t}{B_t}^{-1}{x}_t^{\top}\widehat{\sigma}_t\left(\widehat{\lambda}_t+\alpha_{Y,t}+\widehat{\beta}_t\right)\mathrm{d}t\Big].
\end{aligned}
\end{equation}
 after optimising over $\left\lbrace x_t\right\rbrace_{t\in\left[0,T\right]}\in\widehat{\mathcal{A}}_{X_0}'$. CS then engenders \eqref{eq:driftdiffmultprocsthm7} and \eqref{eq:thmalternativcsconditions}.

\citet{cvitanic1992convex} show that $\widehat{\beta}^{\mathrm{opt}}_t$ ensures ${x}_t^{\mathrm{opt}}\in\widehat{\mathcal{A}}_{X_0}$. By Theorem \ref{clarkocone}: $\big({x}_t^{\top}\widehat{\sigma}_t+\widehat{\beta}_t^{\top}\big)X_t^{\mathrm{opt}}{Y_t}^{\mathrm{opt}}{B_t}^{-1}=\EE\big[\mathcal{D}^W_t\big(X_T^{\mathrm{opt}}{Y_T}^{\mathrm{opt}}{B_T}^{-1}+\int_t^Tc_u^{\mathrm{opt}}{Y_u}^{\mathrm{opt}}{B_u}^{-1}\mathrm{d}u\big)\ \big| \ \eF_t\big]^{\top}$. Let
\begin{equation}
\begin{aligned}
\widehat{\mathcal{G}^{1}}_{t,u}&:=\int_t^u\left(\mathcal{D}^W_t\pi_s-\mathcal{D}^W_t\widehat{\xi}_s\cdot \widehat{\xi}_s\right)\mathrm{d}s+\int_t^u\mathcal{D}^W_t\widehat{\xi}_s\mathrm{d}W_s+\widehat{\xi}_t,\ u\geq t\\ \widehat{\mathcal{G}^2}_{t,u}&:=-\int_t^u\mathcal{D}^{W}_t\left(R_{f,s}+\delta\left(\widehat{\beta}_t\right)+\mathcal{D}^W_t\widehat{\beta}_s\cdot\widehat{\beta}_s\right)\mathrm{d}s+\int_t^u\mathcal{D}^W_t\widehat{\beta}_s\mathrm{d}W_s, \ u\geq t,
\end{aligned}
\end{equation}
in which $\mathcal{D}^W_tR_{f,s}$ for $s\geq t$ is given in Theorem \ref{thm1}. Further, all $\widehat{\mathcal{R}^2}$ terms are equivalent to the $\mathcal{R}^2$ terms with inclusion of \eqref{eq:thm2foc1} rather than \eqref{eq:genfictoptwealth}. Accordingly,
\begin{equation}\label{eq:thm2optrulportfol}
\begin{aligned}
{x}_t^{\mathrm{opt}}&=\frac{\widehat{\sigma}_t^{\top^{-1}}}{X_t^{\mathrm{opt}}}\Bigg(\ddot{x}_t^{m,\mathrm{opt}}\widehat{\beta}_t+B_tY_t^{-1}\EE\left[\int_t^T\widehat{\mathcal{R}^1}_{c,u}^{-1}\widehat{\mathcal{G}^1}_{t,u}\mathrm{d}u+\widehat{\mathcal{R}^1}_X^{-1}\widehat{\mathcal{G}^1}_{t,T}\cond\eF_t\right]\\&+\EE\left[\int_t^T\left(1-\widehat{\mathcal{R}^2}_{c,u}^{-1}\right)\widehat{\mathcal{G}^2}_{t,u}\ddot{c}_{t,u}^{\mathrm{opt}}\mathrm{d}u+\left(1-\widehat{\mathcal{R}^2}_X^{-1}\right)\widehat{\mathcal{G}^2}_{t,T}X_T^{\mathrm{opt}}{Y_{t,T}}{B_{t,T}}^{-1}\cond\eF_t\right]\Bigg)
\end{aligned}
\end{equation}
where $\ddot{x}^{m,\mathrm{opt}}_t=-\EE\big[\big(\int_t^T\widehat{\mathcal{R}^2}_{c,u}^{-1}\ddot{c}_{u,t}^{\mathrm{opt}}\mathrm{d}u+\widehat{\mathcal{R}^2}_X^{-1}X_T^{\mathrm{opt}}\frac{Y_{t,T}^{\mathrm{opt}}}{B_{t,T}}\big)\ \big| \ \eF_t\big]$, $\widehat{\mathcal{R}^1}_{c,u}=\widetilde{I}'_Y\big(e^{\int_0^u\beta_s\mathrm{d}s}\frac{Y_u^{\mathrm{opt}}}{B_u},\Pi_u\big)\\\widehat{M_u}$, $\widehat{\mathcal{R}^1}_{X}=I'_Y\big(Y_T^{\mathrm{opt}}B_T^{-1},\Pi_T\big)\widehat{M_T}$, in which we define $\widehat{M_t}=B_t^{-1}Y_t^{\mathrm{opt}}\Pi_t$ $\forall$ $t\in\left[0,T\right]$.\footnote{Here, $X_t^{\mathrm{opt}}Y_tB_t^{-1}=\EE\big[\int_t^T{c}_{u}^{\mathrm{opt}}\mathrm{d}u+X_T^{\mathrm{opt}}{Y_{T}^{\mathrm{opt}}}{B_{T}}^{-1} \ \big| \ \eF_t\big]$, for $Y_{t,T}:=Y_t^{-1}Y_T$, $B_{t,T}:=B_t^{-1}B_T$. More, $Y_0^{\mathrm{opt}}$ is acquirable, similarly to $\mathcal{H}^{-1}\left(X_0\right)$ in Theorem \ref{thm1}, and `$\odot$' denotes the Hadamard product.}
\end{proof}

The above theorem reveals two paramount facets that lay the groundwork for the approximate method: (i) the existence of an artificial market, in which by construction the encapsulated objective function supplies an upper bound on the true value function, and (ii) the generic FOCs essential to this artificial market which succeed in ensuring optimality in the constrained baseline $\mathcal{M}$. In the interest of the two aspects, let us turn to a full specification of $\mathcal{M}_{\widehat{\nu}}$ by conciliating the results in Theorem \ref{thm2} with the conceptualisation of the fictitious market in section \ref{sec2.2} according to \citet{cvitanic1992convex}. To achieve this, parametrise $\widehat{\beta}_t=-\widehat{\lambda}_t-\widehat{\sigma}_t^{-1}\widehat{\nu}_t$ for some $\widehat{\nu}_t\in \DD^{1,2}\left(\left[0,T\right]\right)^{d+m}$, and rewrite
\begin{equation}
-{\alpha_{Y,t}}=\delta\big(\widehat{\nu}_t\big):=\sup_{\left\lbrace\widehat{x}_t\right\rbrace_{t\in\left[0,T\right]}\in\widehat{\mathcal{A}}^{P}_{X_0}}\big(-\widehat{\nu}_t^{\top}\widehat{x}_t\big),\ \ \mathrm{with}\ \ \widetilde{\mathcal{A}}^P_{X_0}:= \widetilde{\mathcal{A}}_{X_0}=\big\{ \widehat{\nu}_t\mid \delta\left(\widehat{\nu}_t\right)<\infty\big\}.
\end{equation}

The drift term $\alpha_{Y,t}$ encircled within $Y_t$ thus alters into the negative counterpart of the support function in $\mathcal{M}_{\widehat{\nu}}$, such that the corresponding barrier cone changes as well into the analogue of $\mathcal{M}_{\widehat{\nu}}$,  see \eqref{eq:sec22suppfuncetc}. Similarly, the set of feasible dual controls $\mathcal{H}_{\widehat{\mathcal{A}}}^P$ becomes $\mathcal{H}_{\widehat{\mathcal{A}}}$, and we require that $\widehat{\nu}_t\in\mathcal{H}_{\widehat{\mathcal{A}}}$, in accordance with the dual constraint \eqref{eq:sec5thm7generaldual}. Let us then observe that the $B_t^{-1}Y_t$ term is for all $t\in\left[0,T\right]$ separable as $Y_0\widehat{B}_t^{-1}{\mathrm{d}\QQ^{\widehat{\nu}}}/{\mathrm{d}\PP}\big|_{\eF_t}$ or likewise as $Y_0Z_t^{\widehat{\nu}}$, where $\widehat{B}_t$ is the fictitious money market account, $Z_t^{\widehat{\nu}}$ the artificial SPD, and $\QQ^{\widehat{\nu}}$ the non-unique pricing measure in $\mathcal{M}_{\widehat{\nu}}$. Therefore, $B_t^{-1}Y_t$ alludes to a new market, wherein $\widehat{B}_t$ serves as money market account and $\widehat{S}_t$ in \eqref{eq:sec22s} as risky asset.\footnote{Under the stricture $\delta(\widehat{\nu}_t)=-\widehat{\nu}_t^{\top}{x}_t$ for $\widehat{\nu}_t\in\mathcal{H}_{\widehat{\mathcal{A}}}$, $Y_0^{-1}Y_tB_t^{-1}$ functions as SPD in $\mathcal{M}$. If this condition is not met, then a recovery of $\{ x_t\}_{t\in\left[0,T\right]}$ in $\mathcal{M}$ fails, but succeeds for $\{\widehat{x}_t\}_{t\in\left[0,T\right]}$ in $\mathcal{M}_{\widehat{\nu}}$. The optimal decisions $\{\widehat{x}_t^{\mathrm{opt}}\}_{t\in\left[0,T\right]}$ in the fictitious $\mathcal{M}_{\widehat{\nu}}$ may therefore not be admissible in $\mathcal{M}$ for every $\widehat{\nu}_t\in\mathcal{H}_{\widehat{\mathcal{A}}}$.}

In the artificial market, the investor is unconstrained with regard to the choice of portfolio rules $\left\lbrace\widehat{x}_t\right\rbrace_{t\in\left[0,T\right]}$, effecting that the martingale method is applicable. Hence,
\begin{equation}\label{eq:fictitiousprimalproblem}
\begin{aligned}
\sup_{(X_T^{\widehat{\nu}},\widehat{c}_t)\in \widehat{L^2}} &\EE\bigg[\int_0^Te^{\int_0^t\beta_{s}\mathrm{d}s}u\left(\widehat{c}_t,\Pi_t\right)\mathrm{d}t+U\left(X_T^{\widehat{\nu}},\Pi_T\right)\bigg]\\  \mathrm{s.t.} \quad  &\EE\bigg[\int_0^T\widehat{c}_tZ_t^{\widehat{\nu}}\mathrm{d}t+X_T^{\widehat{\nu}}Z_T^{\widehat{\nu}}\bigg]\leq X_0, \ \ \mathrm{for} \ \ \widehat{\nu}_t\in\mathcal{H}_{\widehat{\mathcal{A}}}
\end{aligned}
\end{equation}
characterises the fictitious static duplicate of the primal problem in \eqref{eq:gendynprobdfdafdaf}, where $X_T^{\widehat{\nu}}$ is the terminal value of synthetic wealth in \eqref{eq:sec21ficdynconstr}, $\left\lbrace\widehat{c}_t\right\rbrace_{t\in\left[0,T\right]}$ the fictitious consumption behaviour, $Z_t^{\widehat{\nu}}$ equates to $Y_0^{-1}Y_tB_t^{-1}$, and $\widehat{\nu}_t\in\mathcal{H}_{\widehat{\mathcal{A}}}$ is left unspecified. As a consequence, optimal fictitious horizon wealth and consumption respectively live by $X_T^{\widehat{\nu},\mathrm{opt}}=I\big(\widehat{\eta}^{\mathrm{opt}} Z_T^{\widehat{\nu}},\Pi_T\big)$ and $\widehat{c}_t^{\mathrm{opt}}=\widetilde{I}\big(\widehat{\eta}^{\mathrm{opt}}e^{\int_0^t\beta_s\mathrm{d}s}Z_t^{\widehat{\nu}},\Pi_t\big)$, where the matching rules $\widehat{x}_t^{\mathrm{opt}}$ are equal to $x_t^{\mathrm{opt}}$ in \eqref{eq:thm2optrulportfol}, and the Lagrange multiplier $\widehat{\eta}^{\mathrm{opt}}\in\RR_+$ is such that the static constraint binds. 

\begin{remark}
Let $\mathcal{L}^{\widehat{\nu}}_{t,T}=\int_0^Te^{-\int_0^t\beta_{s}\mathrm{d}s}u\left(\widehat{c}_t,\Pi_t\right)\mathrm{d}t+U\left(X_T^{\widehat{\nu}},\Pi_T\right)$, and $\eta^{\widehat{\nu}}\in\RR_+$ be such that the budget constraint \eqref{eq:fictitiousprimalproblem} binds. Then, the specification in \eqref{eq:sec5thm7generaldual} is identical to 
\begin{equation}\label{eq:sec32dualversusficcompl}
J^{\mathrm{opt}}\left(X_0\right)=\inf_{\widehat{\nu}_t\in\mathcal{H}_{\widehat{\mathcal{A}}}}\left(\sup_{(X_T^{\widehat{\nu}},\widehat{c}_t)\in\widehat{L^2}}\EE\left[\mathcal{L}^{\widehat{\nu}}_{t,T}-\eta^{\widehat{\nu}}\left(\int_0^Tc_tZ_t^{\widehat{\nu}}\mathrm{d}t+X_T^{\widehat{\nu}}Z_T^{\widehat{\nu}}- X_0\right)\right]\right).
\end{equation}
\end{remark}

Inserting $X_T^{\widehat{\nu},\mathrm{opt}}$ and $\widehat{c}_t^{\mathrm{opt}}$ into the fictitious value function thus recovers the dual objective in Theorem \ref{thm2}. As a result, the dual objective harmonises with opting for the worst-case fictitious market. Similarly, for any $\widehat{\nu}_t\in\mathcal{H}_{\widehat{\mathcal{A}}}$, the objective in \eqref{eq:fictitiousprimalproblem} supplies an upper bound on \eqref{eq:sec5thm7generaldual}, implying that the dual coincides with selecting the lowest possible upper bound. By strong duality, this least-favourable choice rallies the constrained baseline financial market $\mathcal{M}$, wherein $\widehat{c}_t^{\mathrm{opt}}$, $\widehat{x}_t^{\mathrm{opt}}$ meet the admissibility requirements enclosed by $\widehat{\mathcal{A}}_{X_0}$. The link between $\widehat{\nu}_t$ and the economic scenarios is evident from the SPD \eqref{eq:sec22prickernpert}.

\subsection{Optimality in Incomplete Market}\label{sec3.3}
This section concludes by illustrating the main findings in Theorem \ref{thm2} for the case of non-traded risk. We use a terminal wealth framework that mildly adjusts the economic setting in section \ref{sec2}. Here, we aim to emphasise the benefit of an approximate method. To accomplish this, reconsider $\mathcal{M}$, set $\sigma_t^{S_1}:=\big[\sigma_t^{S_1},0_{m\times d}\big]^{\top}$, $\sigma_t^{S_2}:=\big[0_{d\times m},\sigma_t^{S_2}\big]^{\top}$, and restrict $S_t$ to its first $d$ entries such that solely  $W_t^1$ is driving $S_t$. Moreover, we impose a trading restriction through $K\subseteq\RR^{d+m}$, by fixing $K=\RR^d\times\left\lbrace 0\right\rbrace^m$: $W^2_t$ is not traded.\footnote{Practically, $K=\RR^d\times\left\lbrace 0\right\rbrace^m$ ascertains partial hedgeability of $\left\lbrace\Pi_t\right\rbrace_{t\in\left[0,T\right]}$.}

In the fictitious $\mathcal{M}_{\widehat{\nu}}$, the support function and barrier cone change into $\delta\left(\widehat{\nu}_t\right)=0<\infty,\quad\mathrm{for}\quad \widehat{\nu}_t\in\widetilde{\mathcal{A}}_{X_0}:=\left\lbrace 0\right\rbrace^d\times\DD^{1,2}\left(\left[0,T\right]\right)^m$, which identically relates to $\widehat{\nu}_t\in\mathcal{H}_{\widehat{\mathcal{A}}}=\left\lbrace 0\right\rbrace^d\times\DD^{1,2}\left(\left[0,T\right]\right)^m$. Therefore, suppose that $\widehat{\nu}_t=\big[\widehat{\nu}_{1,t},\widehat{\nu}_{2,t}\big]^{\top}$, then in $\mathcal{M}_{\widehat{\nu}}$ we must have that $\widehat{\nu}_{1,t}=0_d$ and $\widehat{\nu}_{2,t}\in\DD^{1,2}\left(\left[0,T\right]\right)^m$ hold. The artificial non-traded asset abides by 
\begin{equation}
\frac{\mathrm{d}D_t}{D_t}=\left(R_{f,t}1_m+{\sigma}_t^{S_2}\lambda_{2,t}+\widehat{\nu}_{2,t}\right)\mathrm{d}t+\sigma_t^{S_2}\mathrm{d}W^2_t, \ D_0=1_m,
\end{equation}
wherein $\widehat{\nu}_{2,t}$ is `free' to realise values in the entirety of $\DD^{1,2}\left(\left[0,T\right]\right)^m$, which coalesces with the inability to evaluate the price of unhedgeable uncertainty. These requirements spell out an artificial market wherein $S_t$ remains as a traded asset, because $S_t$ is \textit{ab initio} unaffected by the undiversifiable $\left\lbrace W^2_t\right\rbrace_{t\in\left[0,T\right]}$ on account of $\sigma_t^{S_2}=0_{d+m\times m}$, and $\widehat{\nu}_{1,t}=0_d$. Furthermore, $B_t=\widehat{B}_t$ continues to live in $\mathcal{M}_{\widehat{\nu}}$ due to $\delta\left(\widehat{\nu}_t\right)=0$, and we group $\widehat{S}_t=\left[S_t,D_t\right]^{\top}$.

To alleviate notation, let us separate $\widehat{x}_t$ as $\widehat{x}_t=\left[x_t,x_{f,t}\right]^{\top}$, in which $x_t$ is the allocation to $S_t$ in the true market $\mathcal{M}$ and $x_{f,t}$ the proportion of wealth invested in $D_t$. Then,
\begin{equation}\label{eq:sec33fictbcwealtheq}
\mathrm{d}X^{\widehat{\nu}}_t=X^{\widehat{\nu}}_t\left[\left(R_{f,t}+x_t^{\top}\sigma_t^{S_1}\lambda_{1,t}+x_{f,t}^{\top}\sigma_t^{S_2}\widehat{\lambda}_{2,t}\right)\mathrm{d}t+x_t^{\top}\sigma_t^{S_1}\mathrm{d}W^1_t+x_{f,t}^{\top}\sigma_t^{S_2}\mathrm{d}W^2_t\right],
\end{equation}
describes the artificial wealth process, where we let $\widehat{\lambda}_{2,t}=\lambda_{2,t}+{\sigma_t^{S_2}}^{-1}\widehat{\nu}_{2,t}$ for brevity. Additionally, we suppress inclusion of a consumption process, because we confine the analysis to an investor who derives utility from solely horizon wealth, for the sake of clarity. Admissibility of $\widehat{x}_t$ in the artificial environment $\mathcal{M}_{\widehat{\nu}}$ inherits its definition from \eqref{eq:sec22assump}. Note that admissibility in the constrained environment $\mathcal{M}$ holds only if $x_{f,t}=0_m$. Ultimately, we observe that the fictitious dynamic allocation problem reads as\footnote{This problem in $\mathcal{M}$ is $\sup_{\left\lbrace x_t\right\rbrace\in\widehat{\mathcal{A}}_{X_0}'}\EE\left[U\left(X_T,\Pi_T\right)\right]\ \mathrm{s.t.} \ \mathrm{d}X_t=X_t\big[\big(R_{f,t}+x_t^{\top}\sigma_t^{S_1}\lambda_{1,t}\big)\mathrm{d}t+x_t^{\top}\sigma_t^{S_1}\mathrm{d}W^1_t\big]$.}
\begin{equation}\label{eq:sec33portfolprob}
\sup_{X_T^{\widehat{\nu}}\in L^2\left(\Omega\right)}\EE\left[U\left(X_T^{\widehat{\nu}},\Pi_T\right)\right] \quad \mathrm{s.t.} \quad \EE\left[X_T^{\widehat{\nu}}Z_T^{\widehat{\nu}}\right]\leq X_0.
\end{equation}

In application to the previous setting, we proceed by breaking down Theorems \ref{thm1} and \ref{thm2} into three propositions. The first proposition concerns optimality in the artificial environment $\mathcal{M}_{\widehat{\nu}}$. The second proposition poses the dual problem along with restrictions invoked on the the multiplier process. The last proposition retains optimality in $\mathcal{M}$.
\begin{proposition}\label{prop1}
Consider the problem in \eqref{eq:sec33portfolprob}. Then, optimal horizon wealth follows
\begin{equation}\label{eq:genfictoptwealth1}
X_t^{\mathrm{opt}}=\EE\left[I\left(\mathcal{H}^{-1}\left(X_0\right)Z_T^{\widehat{\nu}},\Pi_T\right)Z_{t,T}^{\widehat{\nu}}\cond\eF_t\right], \ \ \forall \ 0\leq t\leq T
\end{equation}
where $\eta^{\mathrm{opt}}:=\widehat{\mathcal{H}}^{-1}\left(X_0\right)$ characterises the multiplier.\footnote{Here, $\widehat{\mathcal{H}}:\RR_+\rightarrow\RR_+$ follows from $\widehat{\mathcal{H}}\left(\eta\right)=\EE\left[I\left(\eta Z_T^{\widehat{\nu}},\Pi_T\right)Z_{T}^{\widehat{\nu}}\right]$, and $Z^{\widehat{\nu}}_{t,T}:={Z^{\widehat{\nu}}_{t}}^{-1}Z^{\widehat{\nu}}_{T}$.} The optimal portfolio rises from
\begin{equation}\label{eq:genoptrulesfict}
\widehat{x}_t^{\mathrm{opt}}=\left(Z_t^{\widehat{\nu}}X_t^{\mathrm{opt}}{\widehat{\sigma}}_t^{\top}\right)^{-1}\left\lbrace\EE\left[\mathcal{D}^{W}_tZ_{T}^{\widehat{\nu}}X_T^{\mathrm{opt}}\cond\eF_t\right] + \widehat{\lambda}_t\EE\left[Z_T^{\widehat{\nu}}X_T^{\mathrm{opt}}\cond\eF_t\right] \right\rbrace.
\end{equation}
Then, define the previsible projections of transformed RRA (proxies) as
\begin{equation}
\begin{aligned}
{\mathcal{R}}^2_t\left(I\left(\eta^{\mathrm{opt}} Z_T^{\widehat{\nu}},\Pi_T\right),\Pi_T\right)^{-1}&=-\EE\left[\frac{Z_T^{\widehat{\nu}}}{\EE[X_T^{\mathrm{opt}}Z_T^{\widehat{\nu}}\ | \ \eF_t]}\frac{U'_X\left(X_T^{\mathrm{opt}},\Pi_T\right)}{U''_{XX}\left(X_T^{\mathrm{opt}},\Pi_T\right)}\cond\eF_t\right]\\
{\mathcal{R}}^1_t\left(I'_Y\left(\eta^{\mathrm{opt}} Z_T^{\widehat{\nu}},\Pi_T\right),\Pi_T\right)^{-1}&=-\EE\left[{\EE\left[X_T^{\mathrm{opt}}Z_T^{\widehat{\nu}}\cond\eF_t\right]}^{-1}{I'_Y\left(\eta^{\mathrm{opt}} Z_T^{\widehat{\nu}},\Pi_T\right)\widehat{M_T}}\cond\eF_t\right],
\end{aligned}
\end{equation}
where we abbreviate ${\mathcal{R}}^1_t:={\mathcal{R}}^1_t\left(I'_Y\left(\eta^{\mathrm{opt}} Z_T^{\widehat{\nu}},\Pi_T\right),\Pi_T\right)^{-1}$ and ${\mathcal{R}}^2_t:={\mathcal{R}}^2_t\left(I\left(\eta^{\mathrm{opt}} Z_T^{\widehat{\nu}},\Pi_T\right),\Pi_T\right)$. Define $\widehat{\mathcal{R}^i}_{t,T}$ such that ${\mathcal{R}}^i_t=\EE\big[\widehat{\mathcal{R}^i}_{t,T}\ \big| \ \eF_t\big]$ for $i=1,2$. Then, $x_{f,t}^{\mathrm{opt}}$ keeps itself to 
\begin{equation}\label{eq:genficweightactu}
\begin{aligned}
{x}_{f,t}^{\mathrm{opt}}&={{\sigma}^{S_2}_t}^{\top^{-1}}\Big(\mathcal{R}^1_t\left(I'_Y\left(\eta^{\mathrm{opt}} Z_T^{\widehat{\nu}},\Pi_T\right),\Pi_T\right)^{-1}\mathcal{G}'^1_{t,T}\\&+\left(1-\mathcal{R}^2_t\left(I\left(\eta^{\mathrm{opt}} Z_T^{\widehat{\nu}},\Pi_T\right),\Pi_T\right)^{-1}\right)\mathcal{G}'^{2}_{t,T}+\mathcal{R}^2_t\left(I\left(\eta^{\mathrm{opt}} Z_T^{\widehat{\nu}},\Pi_T\right),\Pi_T\right)^{-1}{\widehat{\lambda}_{2,t}}\Big).
\end{aligned}
\end{equation}
The hedging coefficients in optimal $\widehat{x}_t^{\mathrm{opt}}$ are given $\widehat{\lambda}_{f,t}=\big[\lambda_{1,t},\widehat{\lambda}_{2,t}\big]^{\top}$ equal to
\begin{equation}\label{eq:genhedgecoff}
\begin{aligned}
\mathcal{G}^{2}_{t,T}&=\mathcal{R}^2_t\EE\left[-\widehat{\mathcal{R}^2}_{t,T}\left(-\int_t^T\left[\mathcal{D}^{W}_tR_{f,s}\mathrm{d}s+\mathcal{D}^W_t\widehat{\lambda}_{f,s}\left(\mathrm{d}W_s+\widehat{\lambda}_{f,s}\mathrm{d}s\right)\right]\right)\cond\eF_t\right]\\\mathcal{G}^{1}_{t,T}&=\mathcal{R}^1_t\EE\left[-\widehat{\mathcal{R}^1}_{t,T}\left(\int_t^T\left(\left[\mathcal{D}^W_t\pi_s-\mathcal{D}^W_t\widehat{\xi}_s \cdot\widehat{\xi}_s\right]\mathrm{d}s+\mathcal{D}^W_t\widehat{\xi}_s\mathrm{d}W_s\right)+\widehat{\xi}_t\right)\cond\eF_t\right],
\end{aligned}
\end{equation}
where $\mathcal{G}^{1}_{t,T}=\big[\mathcal{G}''^{1}_{t,T},\mathcal{G}'^{1}_{t,T}\big]^{\top}$ and $\mathcal{G}^{2}_{t,T}=\big[\mathcal{G}''^{2}_{t,T},\mathcal{G}'^{2}_{t,T}\big]^{\top}$. Hence, $\widehat{x}_{t}^{\mathrm{opt}}$ springs from the analytic for ${x}_{f,t}^{\mathrm{opt}}$: replace ${\sigma}_t^{S_2}$ by $\widehat{\sigma}_t$, and the terms $\mathcal{G}'^1_{t,T}$ and $\mathcal{G}'^2_{t,T}$ by respectively $\mathcal{G}^1_{t,T}$ and $\mathcal{G}^2_{t,T}$.\footnote{Replacing the terms in parentheses for $\mathcal{G}^1_{t,T}$ and $\mathcal{G}^2_{t,T}$ in \eqref{eq:genhedgecoff} by respectively $(\int_t^T\mathcal{D}^{W^{2}}_t\pi_s\mathrm{d}s+{\xi}^{\Pi_2}_t)$ and $(-\int_t^T(\mathcal{D}^{W^2}_tR_{f,s}\mathrm{d}s+\mathcal{D}^{W^2}_t\widehat{\lambda}_{2,s}\mathrm{d}W^{2,\QQ^{\widehat{\lambda_2}}}_s))$ generates the optimal fictitious allocation $x_{f,t}^{\mathrm{opt}}$ to $D_t$.}
\end{proposition}
\begin{proof}
The results and corresponding proof emanate trivially from Theorem \ref{thm1}.
\end{proof}

Proposition \ref{prop1} summarises the optimal solutions to the fictitious terminal wealth problem in \eqref{eq:sec33portfolprob}, which amounts to showcasing an adjusted version of combined elements that establish Theorems \ref{thm1} and \ref{thm2}. By implication, the interpretation of results roughly shadows those of the latter two. There are, however, three aspects that are different in view of the artificial market $\mathcal{M}_{\widehat{\nu}}$ incorporating the unspecified shadow prices $\widehat{\nu}_t\in\mathcal{H}_{\widehat{\mathcal{A}}}$: (i) the optimal controls that explicitly depend on $\widehat{\nu}_{2,t}$, (ii) an altered decomposition as in \eqref{eq:sec31decomp}, and (iii) the dormant existence of a $\widehat{\nu}_{2,t}$ such that $x_{f,t}^{\mathrm{opt}}=0_m$. As to item (i):
\begin{equation}
\widehat{x}_t^{\mathrm{opt}}=\left(B_t^{-1}X_t^{\mathrm{opt}}\widehat{\sigma}_t^{\top}\right)^{-1}\EE^{\QQ^{\widehat{\nu}}}\left[\mathcal{D}^W_tB_T^{-1}X_T^{\mathrm{opt}}-B_T^{-1}X_T^{\mathrm{opt}}\int_t^T\mathcal{D}^W_t\widehat{\lambda}_{f,s}\mathrm{d}W_s^{\QQ^{}\widehat{\nu}}\right]
\end{equation}
shows reminiscent of \eqref{eq:sec32explicdepq} that the optimal portfolio decisions depend on the measure $\QQ^{\widehat{\nu}}$ that induces in conjunction with $\widehat{B}_t$ in $\mathcal{M}_{\widehat{\nu}}$ the SPD. Endogeneity of $\widehat{\nu}_{2,t}$ underpinning the non-unicity of this measure $\QQ^{\widehat{\nu}}$ demonstrates that $\widehat{x}_t^{\mathrm{opt}}$ is subject to a potentially preference-specific input. The indefinite shadow prices $\widehat{\nu}_{2,t}$ thus also enjoy the power to alter $\widehat{x}_{f,t}^{\mathrm{opt}}$. The dual problem employs this quality to effectuate $x_{f,t}^{\mathrm{opt}}=0_m$.\footnote{In the third proposition, we show that minimising the dual --- choosing the least-favourable fictitious completion $\mathcal{M}_{\widehat{\nu}}$ among all feasible $\widehat{\nu}_t\in\mathcal{H}_{\widehat{\mathcal{A}}}$ --- indeed ensures $x_{f,t}^{\mathrm{opt}}=0_m$ such that $\widehat{x}_t^{\mathrm{opt}}\in\widehat{\mathcal{A}}'_{X_0}$.}

In reference to item (ii), we imitate the decomposition procedure applied to acquire \eqref{eq:sec31decomp}, which commands that $\widehat{x}_t^{\mathrm{opt}}$ is detachable as  $\widehat{x}_t^{\mathrm{opt}}=\widehat{x}_t^{\mathrm{m}}+\widehat{x}_t^Z+\widehat{x}_t^{\Pi}$, wherein
\begin{equation}\label{eq:sec31decomport}
\begin{aligned}
\widehat{x}^{Z}_t&=-\widehat{\sigma}_t^{\top^{-1}}\frac{\mathcal{R}^2_t}{1-\mathcal{R}_t^2}\EE\left[\widehat{\mathcal{R}^2}_{t,T}\left(-\int_t^T\left[\mathcal{D}^{W}_tR_{f,s}\mathrm{d}s+\mathcal{D}^W_t\widehat{\lambda}_{f,s}\mathrm{d}W_s\right]\right)\cond\eF_t\right]\\ \widehat{x}^{\Pi}_t&=-\widehat{\sigma}_t^{\top^{-1}}\EE\left[\widehat{\mathcal{R}^1}_{t,T}\left(\int_t^T\left(\left[\mathcal{D}^W_t\pi_s-\mathcal{D}^W_t\widehat{\xi}_s\cdot \widehat{\xi}_s\right]\mathrm{d}s+\mathcal{D}^W_t\widehat{\xi}_s\mathrm{d}W_s\right)+\widehat{\xi}_t\right)\cond\eF_t\right],
\end{aligned}
\end{equation}
are respectively fictitious nominal stochastic deflator and commodity price hedges. The first rule concerns the ordinary tangency mean-variance efficient portfolio $\widehat{x}_t^{\mathrm{m}}=\frac{1}{\mathcal{R}_t^2}(\widehat{\sigma_t}^{\top})^{-1}\widehat{\lambda}_t$.\footnote{Cognate to $\widehat{x}_t=\left[x_t,x_{f,t}\right]^{\top}$, the three demands separate into a fictitious `$x_{f,t}$' part and a true `$x_{t}$' part. Changing the Malliavin kernel to $\mathcal{D}^{W^1}_t$ or $\mathcal{D}^{W^2}_t$ in \eqref{eq:sec31decomport} isolates the true from fictitious demands.} The shadow prices $\widehat{\nu}_{2,t}$ affect all three disentangled portfolio demands, due to the fact that $\{\widehat{\lambda}_{f,t}\}_{t\in\left[0,T\right]}$ embeds $\{\widehat{\nu}_{2,t}\}_{t\in\left[0,T\right]}$, and so do the previsible RRA coefficients.\footnote{The two-fund separation principle is by analogy with \eqref{eq:sec31twofsep} as well visible through \eqref{eq:sec31decomport}.} 

Lastly, taking note of $\widehat{\lambda}_{2,t}$'s endogeneity, or coequally $\widehat{\nu}_{2,t}$'s endogeneity, within the expression \eqref{eq:genficweightactu} of Proposition \ref{prop1} for $\widehat{x}_t^{\mathrm{opt}}$, we find that the next identity for $\widehat{\lambda}_{2,t}$
\begin{equation}\label{eq:sec3222dualpar}
\widehat{\lambda}_{2,t}=\left(1-\mathcal{R}^2_t\left(X_T^{\mathrm{opt}},\Pi_T\right)\right)\mathcal{G}'^{2}_{t,T}-\mathcal{R}^2_t\left(X_T^{\mathrm{opt}},\Pi_T\right){\mathcal{R}^1_t\left(X_T^{\mathrm{opt}},\Pi_T\right)}^{-1}\mathcal{G}'^{1}_{t,T}
\end{equation}
gratifies according to \eqref{eq:genficweightactu} to repeal any allocation to $D_t$: contingent on $\widehat{\nu}_{2,t}$ obeying \eqref{eq:sec3222dualpar}, $x_{f,t}^{\mathrm{opt}}=0_m$ holds. Note that the RHS of the latter equality may ingrain $\{ \widehat{\lambda}_{2,t}\}_{t\in\left[0,T\right]}$, instigating non-linearity into \eqref{eq:sec3222dualpar}. In general, the expression for $\widehat{\lambda}_{2,t}$ therefore designates a full backward-forward equation, see \citet{detemple2014portfolio} for similar arguments. As this approach towards collecting $\widehat{\lambda}_{2,t}$ is provisional, we next turn to the duality concepts. 

\begin{proposition}\label{prop2}
Consider the constrained analogue of the portfolio choice problem in \eqref{eq:sec33portfolprob} in the baseline market $\mathcal{M}$. Introduce a strictly positive martingale process
\begin{equation}
\mathrm{d}Y_t = Y_t\left[\beta_{1,t}^{\top}\mathrm{d}{W}^{1}_t+\widehat{\beta}_{2,t}^{\top}\mathrm{d}W^2_t\right]
\end{equation}
for a Lagrange multiplier $Y_0\in\RR_+$, where $\beta_{1,t} = \phi^{M_1}_t-\xi^{\Pi_1}_t$ and $\|\widehat{\beta}_{2,t}\|_{L^1\left(\left[0,T\right]\right)^m}^2\in\DD^{1,2}$. Then, the following inequality holds for any $\widehat{x}_t$ s.t. $x_{f,t}=0_m$, $Y_0$ and $\widehat{\beta}_{2,t}$
\begin{equation}\label{eq:sec33prop2ineq}
\EE\left[U\left(X_T,\Pi_T\right)\right]\leq \EE\left[V\left(Y_TB_T^{-1},\Pi_T\right)\right]+Y_0X_0,
\end{equation}
where $V$ denotes $U$'s convex conjugate. Consequently, the dual variant of the unconstrained dynamic asset allocation problem unfolds itself in the following fashion
\begin{equation}\label{eq:genthm5dual}
\inf_{\widehat{\beta}_{2,t}\in\DD^{1,2}\left(\left[0,T\right]\right)^m, Y_0\in\RR_+}\EE\left[V\left(Y_TB_T^{-1},\Pi_T\right)\right]+Y_0X_0.
\end{equation}
In addition to that, because $x_{t}=0_d$ is feasible, strong duality persists:
\begin{equation}\label{eq:sec33prop2strongdual}
\sup_{\{\widehat{x}_t\}_{t\in\left[0,T\right]}\in\widehat{\mathcal{A}}'_{X_0}}\ \EE\left[U\left({X_T},{\Pi_T}\right)\right]=\inf_{\widehat{\beta}_{2,t}\in\DD^{1,2}\left(\left[0,T\right]\right)^m, Y_0\in\RR_+}\EE\left[V\left(Y_TB_T^{-1},\Pi_T\right)\right]+Y_0X_0.
\end{equation}
Suppose that $\widehat{\beta}_{2,t}^{\mathrm{opt}}$ and $Y_0^{\mathrm{opt}}$ solve \eqref{eq:genthm5dual}. Then, optimal terminal wealth complies with $X_T^{\mathrm{opt}}=I\left(Y_T^{\mathrm{opt}}B_T^{-1},\Pi_T\right)$. Complementary slackness alternatively commands that
\begin{equation}\label{eq:thm2sec3lastequalk}
\sup_{\left\lbrace{x}_t\right\rbrace_{t\in\left[0,T\right]}\in\widehat{\mathcal{A}}'_{X_0}}\int_0^T\left(\EE\left[\widehat{S}_T{Y_T}{B_T}^{-1}\cond\eF_t\right]-\EE\left[Y_T\cond\eF_t\right]{\widehat{S}_t}{B_t}^{-1}\right)^{\top}\mathrm{d}X_t\widehat{S}_t^{-1}\odot\widehat{x}_t<\infty.
\end{equation} 
indicating that $\widehat{\beta}_t=-\widehat{\lambda}_{f,t}$ such that $Y_0^{-1}Y_TB_T^{-1}=B_T^{-1}{\mathrm{d}\QQ^{\widehat{\nu}}}/{\mathrm{d}\PP}$ agrees with $Z_T^{\widehat{\nu}}$ in $\mathcal{M}_{\widehat{\nu}}$.\footnote{The CS condition $\widehat{\beta}_t=-\widehat{\lambda}_{f,t}$ descends from the fact that $x_t\in\RR^d$ $\mathrm{d}t\otimes\PP\text{-a.e.}$, which mandates $\{ Y_t\}_{t\in\left[0,T\right]}$ to operate as $\frac{\mathrm{d}\QQ}{\mathrm{d}\PP}\big|_{\eF_t}$, $0\leq t\leq T$ with regard to $\{ W^1_t\}_{t\in\left[0,T\right]}$. Concurrently, as to $\{ W^2_t\}_{t\in\left[0,T\right]}$, the stricture $x_{f,t}=0_m$ decrees that $\{ Y_t\}_{t\in\left[0,T\right]}$ is unrestrained when it comes to pricing $D_t$.}
\end{proposition}
\begin{proof}
The proof for this proposition arises out of the one for Theorem \ref{thm2}.
\end{proof}

Proposition \ref{prop2} encompasses in conformity with the results in Theorem \ref{thm2} the dual problem corresponding to its primal congener, which is essential to the fictitious specification in \eqref{eq:sec33portfolprob}. In effect, the artificial market is a by-product of the dual concretisation and assists in reconciling the baseline constrained market $\mathcal{M}$ with this fictitious economy $\mathcal{M}_{\widehat{\nu}}$. Consistent with Theorem \ref{thm2}, in the present case, $\mathcal{M}_{\widehat{\nu}}$ finds its definition from the CS condition \eqref{eq:thm2sec3lastequalk} via which $Y_0^{-1}Y_t=\mathrm{d}\QQ^{\widehat{\nu}}/\mathrm{d}\PP|_{\eF_t}$ ought to be true. Hence, with due regard for the num\'{e}raire-based avenue in \eqref{eq:sec32rewrconstr}, the dual implicates minimisation over
\begin{equation}\label{eq:sec3separablith}
Y_0^{-1}Y_t=Z_t^{\widehat{\nu}}B_t=\mathcal{E}\left(\lambda_{1,t}^{\top}\right)\exp\left(-\frac{1}{2}\int_0^T\widehat{\lambda}_{2,t}^{\top}\widehat{\lambda}_{2,t}\mathrm{d}t-\int_0^T\widehat{\lambda}_{2,t}\mathrm{d}W^2_t\right)
\end{equation}
for unspecified $\widehat{\nu}_t\in\mathcal{H}_{\widehat{\mathcal{A}}}$, which confirms that the dual objective unifies itself with opting for a probability measure in an effort to malignly affect utility levels and therefore to make trades in $D_t$ utmost unappealing. The choice for this pricing measure namely acts via $\big\{\widehat{\lambda}_{2,t}\big\}_{t\in\left[0,T\right]}$ purely upon the pricing of non-traded uncertainty $\left\lbrace W^2_t\right\rbrace_{t\in\left[0,T\right]}$.

To unfold this phenomenon, let us define $\mathcal{M}^{\beta_2}:=\{\QQ^{\beta_{2}}\sim\PP \mid \EE^{\QQ^{\beta_{2}}}\left[B_t^{-1}S_t\cond\eF_s\right]=B_s^{-1}S_s,\ \EE^{\QQ^{\beta_{2}}}\left[B_t^{-1}D_t\cond\eF_s\right]-B_s^{-1}D_s\in\RR \ \mathrm{d}t\otimes\PP\text{-a.e.}, \ \forall \ 0\leq s\leq t\leq T\}$. Consequently,
\begin{equation}\label{eq:gensec5substdualprob}
\inf_{\QQ^{{\beta_2}}\in\mathcal{M}^{\beta_2},\eta\in\RR_+}\EE\left[V\left(\eta B_T^{-1}\frac{\mathrm{d}\QQ^{\beta_{2}}}{\mathrm{d}\PP},\Pi_T\right)\right]+\eta X_0,
\end{equation}
for Lagrange multiplier $\eta:=Y_0$, unwinds a substitute for the dual in \eqref{eq:genthm5dual}, underlining the separability of the controllable Radon-Nikodym derivative from the effectively exogenous nominal interest rate.\footnote{Observe that \eqref{eq:gensec5substdualprob} lines as in \eqref{eq:sec32dualversusficcompl} up with choosing the least-favourable completion.} Assume that $\QQ^{\beta_2^{\mathrm{opt}}}\in\mathcal{M}^{\beta_2}$ is optimal to \eqref{eq:gensec5substdualprob}, then $\widehat{M_T}=\exp\big(-\int_0^T\left(R_{f,s}-\pi_s-(\widehat{\beta}_s^{\mathrm{opt}})^{\top}\widehat{\xi}_s\right)\mathrm{d}s\big)\mathcal{E}\big(\big[\phi^{M_1}_t,\widehat{\beta}^{\mathrm{opt}}_{2,t}+\xi^{\Pi_2}_t\big]^{\top}\big)$ defines the analogue of the pricing kernel $M_T$ under $x_{f,t}=0_m$. Hence, whereas $R_{f,t}$ remains exogenous, $r_t$ in the constrained market $\mathcal{M}$ contours \textit{ex-post} an effectively endogenous process.

Lastly, conducive to the approximating dual control mechanism, let us derive that
\begin{equation}\label{eq:sec33ineqnew}
\sup_{\left\lbrace\widehat{x}_t\right\rbrace_{t\in\left[0,T\right]}\in\widehat{\mathcal{A}}_{X_0}'}\EE\left[U\left(X_T,\Pi_T\right)\right]\leq \inf_{\widehat{\beta}_{2,t}\in \mathcal{P},Y_0\in\RR_+}\EE\left[V\left(Y_TB_T^{-1},\Pi_T\right)\right]+X_0Y_0
\end{equation}
where $\mathcal{P}\subseteq\DD^{1,2}\left(\left[0,T\right]\right)^m$ is closed and convex. This inequality ascertains that the repression of $\DD^{1,2}\left(\left[0,T\right]\right)^m$ to $\mathcal{P}$ begets an upper bound on the optimal value function in virtue of \eqref{eq:sec33prop2ineq} and strong duality \eqref{eq:sec33prop2strongdual}. In the same way, the restriction to $\mathcal{P}$ may procreate a lower bound in the primal problem.\footnote{The set $\widehat{\mathcal{A}}_{X_0}'$ reduces to one that complies with all $\widehat{\lambda}_{2,t}\in \mathcal{P}$. The investor is accordingly incapable of optimally protecting him or herself against injurious shifts in $\left\lbrace W^2_t\right\rbrace_{t\in\left[0,T\right]}$, implying such a lower bound.} Combining these two outcomes, contraction of $\DD^{1,2}\left(\left[0,T\right]\right)^m$ raises lower and upper bounds on the optimal value function.

We complete the exposition of this theoretical example by solving the dual \eqref{eq:gensec5substdualprob}. In light of this, let  $\mathcal{W}\big(\big\{\widehat{\beta}_{2,t}\big\},Y_0\big)=\EE\left[V\left(Y_TB_T^{-1},\Pi_T\right)\right]+Y_0X_0$, and discern that 
\begin{equation}
\begin{aligned}
\Delta\mathcal{W}=\EE\bigg[V_X'\bigg(\frac{Y_T}{B_T},\Pi_T\bigg)\frac{Y_T}{B_T}\int_0^T\left(\phi_t^{\top}\mathrm{d}W^2_t-\widehat{\beta}_{2,t}^{\top}\phi_t\mathrm{d}t\right)\bigg]+\Oh\Big(\left\Vert\phi_t\right\Vert_{\DD^{1,2}\left(\left[0,T\right]\right)^m}\Big),
\end{aligned}
\end{equation}
where $\mathcal{W}:\DD^{1,2}\left(\left[0,T\right]\right)^m\times \RR_+$ is the unconstrained objective in consonance with \eqref{eq:gensec5substdualprob}, $\phi_t$ and $\widehat{\beta}_{2,t}$  are such that $\phi_t,\widehat{\beta}_{2,t}+\phi_t\in \DD^{1,2}\left(\left[0,T\right]\right)^m$, and $V_{X}'$ denotes the derivative of $V$ in its fist argument. Further, we let $\mathcal{O}$ denote the Landau function. Then, $\Delta \mathcal{W}$ describes the effect on $\mathcal{W}$ due to a small perturbation $\phi_t$ on $\widehat{\beta}_{2,t}$. Applying the argument that small perturbations around the optimal controls must have insignificant effects on $\mathcal{W}$, setting $\Delta \mathcal{W}$ equal to zero recovers the optimal $\widehat{\beta}_{2,t}$.\footnote{We interchangeably use $\widehat{\beta}_{2,t}$ and $-\widehat{\lambda}_{2,t}$ as well as $Z_t^{\widehat{\nu}}$ and $Y_0^{-1}Y_tB_t^{-1}$; these unequivocally are equivalent.} We formalise this in Proposition \ref{prop3}.
\begin{proposition}\label{prop3}
Consider the dual optimisation problem \eqref{eq:genthm5dual} in compliance with the fictitious specification \eqref{eq:sec33portfolprob} for the optimal allocation to assets. Then,
\begin{equation}\label{eq:lambdaucharopt1}
\begin{aligned}
\widehat{\lambda}_{2,t}^{\mathrm{opt}}&=-\EE\left[B_T^{-1}Y_T^{\mathrm{opt}}X_T^{\mathrm{opt}}\cond\eF_t\right]^{-1}\EE\left[\mathcal{D}^{W^2}_tB_T^{-1}Y_T^{\mathrm{opt}}X_T^{\mathrm{opt}}\cond\eF_t\right]\\&=\left(1-\mathcal{R}^2_t\left(X_T^{\mathrm{opt}},\Pi_T\right)\right)\mathcal{G}'^{2}_{t,T}-\mathcal{R}^2_t\left(X_T^{\mathrm{opt}},\Pi_T\right){\mathcal{R}^1_t\left(X_T^{\mathrm{opt}},\Pi_T\right)}^{-1}\mathcal{G}'^{1}_{t,T}
\end{aligned}
\end{equation}
itemises an equality from which we can extract $\widehat{\lambda}_{2,t}^{\mathrm{opt}}$ that optimises the dual. Further, 
\begin{equation}\label{eq:sec33thm3lagrmult}
\widehat{\mathcal{H}}\left(\eta^{\mathrm{opt}}\right)=\EE\big[I\big(\eta^{\mathrm{opt}}Z_T^{\widehat{\nu}},\Pi_T\big)Z_T^{\widehat{\nu}}\big]=X_0
\end{equation}
qualifies the identity from which we are able to distil the unique multiplier $\eta^{\mathrm{opt}}=Y_0=\widehat{\mathcal{H}}^{-1}\left(X_0\right)$.  The optimal portfolio rules concerning the allocation to $S_t$ obey to
\begin{equation}\label{eq:genappthm6lastrules}
\begin{aligned}
{x}_t^{\mathrm{opt}}&={{{\sigma}_t^{S_1}}^{\top}}^{-1}\Big(\mathcal{R}^1_t\left(I'_Y\left(\eta Z_T^{\widehat{\nu}},\Pi_T\right),\Pi_T\right)^{-1}\mathcal{G}^1_{x,t,T}\\&+\left(1-\mathcal{R}^2_t\left(I\left(\eta Z_T^{\widehat{\nu}},\Pi_T\right),\Pi_T\right)^{-1}\right)\mathcal{G}^{2}_{x,t,T}+\mathcal{R}^2_t\left(I\left(\eta Z_T^{\widehat{\nu}},\Pi_T\right),\Pi_T\right)^{-1}{\lambda}_{1,t}\Big),
\end{aligned}
\end{equation}
in addition to $x_{f,t}=0_m$. The other conditions arising out of Proposition \ref{prop1} remain in force. In harmony with \eqref{eq:sec31decomport}, $x_t^{\mathrm{opt}}=x_t^{m,\mathrm{opt}}+x_t^{Z,\mathrm{opt}}+x_t^{\Pi,\mathrm{opt}}$ holds, in which\footnote{We employ the definitions for ${\mathcal{R}^1}_{t}$ and ${\mathcal{R}^2}_{t}$ along with $\widehat{\mathcal{R}^1}_{t,T}$ and $\widehat{\mathcal{R}^2}_{t,T}$ from Proposition \ref{prop1}.}
\begin{equation}
{{\sigma}_t^{S_1}}^{\top}{x}^{Z,\mathrm{opt}}_t=-\frac{\mathcal{R}^2_t}{1-\mathcal{R}_t^2}\EE\left[\widehat{R^2}_{t,T}\left(-\int_t^T\left[\mathcal{D}^{W^1}_tR_{f,s}\mathrm{d}s+\mathcal{D}^{W^1}_t\widehat{\lambda}_{f,s}\mathrm{d}W_s^{\QQ^{\widehat{\nu}}}\right]\right)\cond\eF_t\right]
\end{equation}
characterises the identity for the stochastic deflator hedge, and where
\begin{equation}
{{\sigma}_t^{S_1}}^{\top}{x}^{\Pi,\mathrm{opt}}_t=-\EE\left[\widehat{\mathcal{R}^1}_{t,T}\left(\int_t^T\left(\mathcal{D}^{W^1}_t\pi_s\mathrm{d}s+\mathcal{D}^{W^1}_t\widehat{\xi}_s\left(\mathrm{d}W_s-\widehat{\xi}_s\mathrm{d}s\right)\right)+\widehat{\xi}_t\right)\cond\eF_t\right]
\end{equation}
defines the equality from where we deduce the commodity price hedge; $x_t^{m,\mathrm{opt}}={\mathcal{R}_t^2}^{-1}{{\sigma_t^{S_1}}^{\top}}^{-1}\lambda_{1,t}$ spells out the mean-variance portfolio. More, the hedging coefficients read
\begin{equation}\label{eq:genappg1g2}
\begin{aligned}
\mathcal{G}^{2}_{x,t,T}&=\mathcal{R}^2_t\EE\left[-\widehat{\mathcal{R}^2}_{t,T}\left(-\int_t^T\left[\mathcal{D}^{W^1}_t\left(R_{f,s}\mathrm{d}s+\widehat{\lambda}_{f,s}\mathrm{d}W_s^{\QQ^{\widehat{\nu}}}\right)\right]\right)\cond\eF_t\right]\\\mathcal{G}^{1}_{x,t,T}&=\mathcal{R}^1_t\EE\left[-\widehat{\mathcal{R}^1}_{t,T}\left(\int_t^T\left(\mathcal{D}^{W^1}_t\pi_s\mathrm{d}s+\mathcal{D}^{W^1}_t\widehat{\xi}_s\left(\mathrm{d}W_s-\widehat{\xi}_s\mathrm{d}s\right)\right)+\xi_t^{\Pi_1}\right)\cond\eF_t\right]
\end{aligned}
\end{equation}
provoking the portfolio rules in \eqref{eq:genappthm6lastrules}, where $\mathcal{D}^{W^1}_tR_{f,s}=\mathcal{D}^{W^1}_tr_s+\mathcal{D}^{W^1}_t\pi_s-\mathcal{D}^{W^1}_t\lambda_{s}^{\top}\widehat{\xi}_s$.\footnote{Optimality of $\widehat{\lambda}_{f,t}$ depends on $K$'s definition and is thus case specific. Hence, we provide the proof.}
\end{proposition} 
\begin{proof}
Reconsider the objective function in \eqref{eq:genthm5dual}, delimited by $\mathcal{W}:\DD^{1,2}\left(\left[0,T\right]\right)^m\times\RR_+\rightarrow\RR$. Notice that $\DD^{1,2}\left(\left[0,T\right]\right)^m\subset L^2\left(\Omega\times\left[0,T\right]\right)^m$. Let $\psi_t\in \DD^{1,2}\left(\left[0,T\right]\right)^m$, then:
\begin{equation}
D_{\left\lbrace\widehat{\beta}_{2,t}\right\rbrace}\mathcal{W}\left\lbrace\psi_t\right\rbrace=\EE\left[V'_X\left(\frac{Y_T}{B_T},\Pi_T\right)\frac{Y_T}{B_T}\left\lbrace-\int_0^T\widehat{\beta}_{2,t}^{\top}\psi_t\mathrm{d}t+\int_0^T\psi_t^{\top}\mathrm{d}W^2_t\right\rbrace\right]=0.
\end{equation}

The foregoing specifies the Fr\'{e}chet derivative and ought to equal $0$. We simplify
\begin{equation}\label{eq:genappcfocagaindwnn}
\begin{aligned}
D_{\left\lbrace\widehat{\beta}_{2,t}\right\rbrace}\mathcal{W}\left\lbrace\psi_t\right\rbrace&=-\EE\Bigg[\int_0^TX_T^{\mathrm{opt}}Z_T^{\widehat{\nu}}\left(\delta\left(\psi_t\right)-\beta_{2,t}^{\top}\psi_t\mathrm{d}t\right)\Bigg]\\&=-\int_0^T\EE\bigg[\psi_t^{\top}\EE\left[\mathcal{D}^{W^2}_tX_T^{\mathrm{opt}}Z_T^{\widehat{\nu}}-\widehat{\beta}_{2,t}X_T^{\mathrm{opt}}Z_T^{\widehat{\nu}}\cond\eF_t\right]\bigg]\mathrm{d}t=0,
\end{aligned}
\end{equation}
where we use $X_T^{\mathrm{opt}}=I\left(B_T^{-1}Y_T,\Pi_T\right)$, $Z_T^{\widehat{\nu}}=B_T^{-1}Y_T$, and the Hermitian adjoint result. As a consequence, $\big\langle\EE\big[\mathcal{D}^{W^2}_tI\big(\frac{Y_T}{B_T},\Pi_T\big)\frac{Y_T}{B_T}-\widehat{\beta}_{2,t}I\big(\frac{Y_T}{B_T},\Pi_T\big)\frac{Y_T}{B_T} \ \big| \ \eF_t\big],\psi_t\big\rangle_{L^2\left(\Omega\times\left[0,T\right]\right)}=0$. By virtue of the Riesz-Fr\'{e}chet representation theorem, given $\mathcal{R}_{x,T}:=I'_X(\eta Z_T^{\widehat{\nu}},\Pi_T)\eta Z_T^{\widehat{\nu}}$, we derive 
\begin{equation}
\begin{aligned}
\widehat{\beta}_{2,t}^{\mathrm{opt}}&=\EE\left[X_T^{\mathrm{opt}}Z_T^{\widehat{\nu}}\cond\eF_t\right]^{-1}\EE\left[\mathcal{D}^{W^2}_tX_T^{\mathrm{opt}}Z_T^{\widehat{\nu}}\cond\eF_t\right]=\EE\left[X_T^{\mathrm{opt}}Z_T^{\widehat{\nu}}\cond\eF_t\right]^{-1}\\&\times\EE\left[\left(\mathcal{R}_{x,T}+X_T^{\mathrm{opt}}\right){\mathcal{D}^{W^2}_tZ_T^{\widehat{\nu}}}+ I'_Y\left(\eta Z_T^{\widehat{\nu}},\Pi_T\right)\widehat{M_T}{\mathcal{D}^{W^2}_t\log\left(\Pi_T\right)}\cond\eF_t\right]
\end{aligned}
\end{equation}
leading to $\widehat{\beta}_{2,t}^{\mathrm{opt}}=-\left(1-\mathcal{R}^2_t\left(X_T^{\mathrm{opt}},\Pi_T\right)\right)\mathcal{G}'^{2}_{t,T}+\mathcal{R}^2_t\left(X_T^{\mathrm{opt}},\Pi_T\right){\mathcal{R}^1_t\left(X_T^{\mathrm{opt}},\Pi_T\right)}^{-1}\mathcal{G}'^{1}_{t,T}$. 

The multiplier $Y_0^{\mathrm{opt}}\in\RR_+$ induces $X_T^{\mathrm{opt}}=I\left(\widehat{\mathcal{H}}^{-1}\big(X_0\right)Y_0^{-1}{Y_T}^{\mathrm{opt}}{B_T}^{-1},\Pi_T\big)$ from
\begin{equation}\label{eq:lagrmultfrechindal}
\begin{aligned}
-D_{Y_0}\mathcal{W}\kappa&=\EE\left[\left(I\left(Y_0Z_T^{\widehat{\nu}},\Pi_T\right)Z_T^{\widehat{\nu}}-X_0\right)\kappa\right]=\big\langle X_T^{\mathrm{opt}}Z_T^{\widehat{\nu}}-X_0,\kappa\big\rangle_{L^2\left(\Omega\right)}=0,
\end{aligned}
\end{equation}
for all $\kappa\in\RR_+$. Then, $\widehat{\beta}_{2,t}^{\mathrm{opt}}$ ensures that $\{\widehat{x}_t^{\mathrm{opt}}\}_{t\in\left[0,T\right]}\in\widehat{\mathcal{A}}_{X_0}'$ such that $X_T^{\mathrm{opt}}$ is both feasible and optimal in the constrained market $\mathcal{M}$, attributable to strong duality \eqref{eq:sec33prop2strongdual}.\footnote{Principally, Proposition \ref{prop3} develops from  inserting $\{\widehat{\beta}_t^{\mathrm{opt}}\}_{t\in\left[0,T\right]}$ into Proposition \ref{prop1}. }  
\end{proof}
The above theorem extricates two ingredients in connection with optimal solutions to the constrained utility maximisation problem in $\mathcal{M}$: (i) the choice for $\widehat{\nu}_{2,t}$ so that the allocation to $D_t$ abrogates in Proposition \ref{prop1}, and (ii) the optimality criteria imposed on the primal controls $\widehat{x}_t$ and $X_T$ that correspond to this expression for $\widehat{\nu}_{2,t}^{\mathrm{opt}}$. As for the former item, we note that the optimal shadow prices $\widehat{\nu}_{2,t}^{\mathrm{opt}}$ result from accurately minimising the dual. Moreover, its equivalence with \eqref{eq:sec3222dualpar} is unambiguous. As to the last element, we observe in consideration of \eqref{eq:sec33portfolprob} that the static problem in $\mathcal{M}$ defers to\footnote{Employing the solution techniques from Theorem \ref{thm1} and Proposition \ref{prop1} yields Proposition \ref{prop3}.}
\begin{equation}
\sup_{X_T\in L^2\left(\Omega\right)}\EE\left[U\left(X_T,\Pi_T\right)\right] \quad \mathrm{s.t.} \quad \EE\left[X_TZ_T^{\widehat{\nu}}\right]\leq X_0, \quad \widehat{\nu}_{2,t}=\widehat{\nu}_{2,t}^{\mathrm{opt}}.
\end{equation}

The preceding three propositions place a strong emphasis on the linking nature of the dual between the constrained primal problem in $\mathcal{M}$ and the unconstrained fictitious specification in $\mathcal{M}_{\widehat{\nu}}$ for the case of non-traded risk.\footnote{For restrictions that differ from $K=\RR^d\times\left\lbrace0\right\rbrace^m$, consider for instance \citet{cuoco1997optimal} or \citet{tepla2000optimal}.} In particular, for each choice of $\widehat{\nu}_t\in\mathcal{H}_{\widehat{\mathcal{A}}}$, we acquire by means of the dual an artificial market that assigns an upper bound to the optimal value function. Adequately minimising the dual, i.e. choosing the least advantageous $\mathcal{M}_{\widehat{\nu}}$ or the smallest upper bound, then recovers a shadow price that induces optimality in the constrained $\mathcal{M}$. It is, however, in general difficult to analytically obtain this price, cf. \eqref{eq:sec3222dualpar}. Therefore, we subsequently outline an approximate method.

\section{Approximate Method}\label{sec4}
This section develops our dual control framework for approximating trading strategies. We cultivate the technique on the grounds of the economic setup in section \ref{sec2}. To surmount the typical absence of closed-form expressions in Theorem \ref{thm2}, we confine the space of feasible $\widehat{\beta}_t$ to a closed and convex set.\footnote{Rather than restricting $\widehat{\nu}_t$ to a coequal set, we choose to apply the principle to $\widehat{\beta}_t$ for mathematical elegance. Nonetheless, the foundation of the method does not alter for impediments with respect to $\widehat{\nu}_t$.} The procedure then resides in the information that this theorem provides, such that an analytical approximation to the truly optimal investment decisions may result. Duality principles spawn lower and upper bounds on the optimal value function, allowing us to gauge the method's accuracy. We append the blueprint of the approximate method with an example linked to section \ref{sec3.3}.

\subsection{Projection of Feasible Strategies}\label{sec4.1}
We commence with the description of the method by underscoring that its substructure consists in a twofold procedure. That is, we first cage the set of dual controls $\widehat{\beta}_t$ to a closed and convex set, $\mathcal{P}\subseteq\DD^{1,2}\left(\left[0,T\right]\right)^n$, which makes up the cornerstone that invigorates the technique at the outset, cf. the analysis around \eqref{eq:sec3222dualpar}. Second, we prescribe a projection operator commensurate with a case-dependent metric that casts the infeasible dual-optimal portfolio rules into the pre-specified feasible region $\widehat{\mathcal{A}}_{X_0}'$ that concurs with the limited space of controls $\mathcal{P}$. The inequality that lies at the root of this method is
\begin{equation}\label{eq:sec41ineq1}
\sup_{\left\lbrace x_t, c_t\right\rbrace_{t\in\left[0,T\right]}\in\widehat{\mathcal{A}}^{\mathcal{P}}_{X_0}}J^L\left(X_0,\left\lbrace x_t,c_t\right\rbrace_{t\in\left[0,T\right]}\right)\leq \inf_{\widehat{\beta}_t\in\mathcal{H}_{\mathcal{A}}^P\cap \mathcal{P},Y_0\in\RR_+}J^U\left(X_0,Y_0,\big\{ \widehat{\beta}_t\big\}_{t\in\left[0,T\right]}\right),
\end{equation} 
where $\widehat{\mathcal{A}}^{\mathcal{P}}_{X_0}$ comprehends all admissible strategies that agree with $\widehat{\beta}_{t}\in\mathcal{H}^P_{\widehat{\nu}}\cap\mathcal{P}$ for $\widehat{\beta_{t}}\in\DD^{1,2}\left(\left[0,T\right]\right)^n$. Moreover, $J^L$ and $J^U$ denote respectively the primal \eqref{eq:gendynprob} and dual \eqref{eq:sec5thm7generaldual} objectives. Theorem \ref{thm2} affirms that \eqref{eq:sec41ineq1} binds if and only if $\mathcal{P}=\DD^{1,2}\left(\left[0,T\right]\right)^n$ or $\widehat{\beta}_t^{\mathrm{opt}}\in\mathcal{P}$, which suggests that the anomaly from the optimal value function, say $J^{\mathrm{opt}}$, as regards $J^L$ and $J^U$ alternates with the efficacy of the approximation to $\widehat{\beta}_t$.

Convex duality in combination with the contraction of shadow prices to $\mathcal{P}$ breaks down in general, as reported by Theorem \ref{thm2}, when it comes to assuring that $\left\lbrace\widehat{x}_t,c_t\right\rbrace_{t\in\left[0,T\right]}\in\widehat{\mathcal{A}}_{X_0}^{\mathcal{P}}$. To overstep the consequential analytical impossibility of solving the LHS of \eqref{eq:sec41ineq1}, we approximate terminal wealth $X_T^{\mathrm{opt}}$ and the congruous consumption streams $c_t^{\mathrm{opt}}$ by
\begin{equation}\label{eq:sec41approx}
X_T^{*}=X_0+\int_0^T\left(R_{f,t}X_t^{*}-\widehat{c}_t^{\mathcal{P},\mathrm{opt}}\right)\mathrm{d}t+\int_0^T\mathrm{proj}_{\widehat{\mathcal{A}}'^{\mathcal{P}}_{X_0}}\widehat{x}_t^{\mathrm{opt},\top}\widehat{\sigma}_tX_t^{*}\left(\mathrm{d}{W}_t+\widehat{\lambda}_t\mathrm{d}t\right),
\end{equation}
where we let $\mathrm{proj}_{\widehat{\mathcal{A}}'^{\mathcal{P}}_{X_0}}:L^2\left(\Omega\times\left[0,T\right]\right)^n\rightarrow \widehat{\mathcal{A}}'^{\mathcal{P}}_{X_0}$ be the projection kernel tallying with an operator that maps any $n$-dimensional state-wise, square-integrable investment plan to the region of admissible trading strategies $\left\lbrace x_t\right\rbrace_{t\in\left[0,T\right]}$ in the constrained economy $\mathcal{M}$, i.e. $\widehat{\mathcal{A}}'^{\mathcal{P}}_{X_0}$. Furthermore, we set $\widehat{c}_t^{\mathcal{P},\mathrm{opt}}=\big(\widehat{c}_t^{\mathrm{opt}}/X_t^{\mathrm{opt}}\big)|_{\widehat{\beta}_t\in\mathcal{H}_{\widehat{\mathcal{A}}}^P\cap\mathcal{P}}X_t^{*}$, wherein $\widehat{c}_t^{\mathrm{opt}}$, $\widehat{x}_t^{\mathrm{opt}}$ and $X_t^{\mathrm{opt}}$ are the optimal controls inherent in the dual of Theorem \ref{thm2}. We also note that approximate terminal wealth $X_T^{*}$ depends on its past values, which are procurable by means of $X_T^{*}|_{T=t}$.

Note that the optimality conditions in the unconstrained $\mathcal{M}_{\widehat{\nu}}$, i.e. those in Theorem \ref{thm2} for undetermined $\widehat{\beta}_t$, are analytically available. The truly optimal $X_T^{\mathrm{opt}}$ and $c_t^{\mathrm{opt}}$ for constrained $\left\lbrace\widehat{x}_t^{\mathrm{opt}}\right\rbrace_{t\in\left[0,T\right]}$ are equal to these closed-form formulae, given a $\big\{\widehat{\beta}_t\big\}_{t\in\left[0,T\right]}$ such that $\left\lbrace\widehat{x}_t,c_t\right\rbrace_{t\in\left[0,T\right]}\in\widehat{\mathcal{A}}_{X_0}$ holds.\footnote{Theorem \ref{thm2} spells out ${x}_t^{\mathrm{opt}}$ in $\mathcal{M}$ under the premise that $\widehat{\beta}_t^{\mathrm{opt}}$ essentially optimises the dual. In this analysis, we utilise these equalities for ${x}_t^{\mathrm{opt}}$ given undefined $\widehat{\beta}_t$, resembling $\widehat{x}_t^{\mathrm{opt}}$ in $\mathcal{M}_{\widehat{\nu}}$.} These initial portfolio rules are thus affected in terms of their admissibility in $\mathcal{M}$ by sole modifications of $\widehat{\beta}_t$. The elemental idea then is that under the reservation of $\DD^{1,2}\left(\left[0,T\right]\right)^n$ to $\mathcal{P}$, these decisions are practically equal to the ones in Theorem \ref{thm2} for unspecified $\widehat{\beta}_{t}\in\mathcal{H}_{\widehat{\mathcal{A}}}^P\cap\mathcal{P}$ and $Y_0\in\RR_+$. We next endeavour to compensate for remaining inaccuracies by `pruning' $\widehat{x}_t^{\mathrm{opt}}|_{\widehat{\beta}_t\in\mathcal{H}_{\widehat{\mathcal{A}}}^P\cap\mathcal{P}}$ towards an identity that meets the admissibility criteria with the help of the projection operator.\footnote{We exclusively mention $\widehat{x}_t^{\mathrm{opt}}|_{\widehat{\beta}_t\in\mathcal{H}_{\widehat{\mathcal{A}}}\cap\mathcal{P}}$ here, because $\widehat{c_t}^{\mathcal{P},\mathrm{opt}}$ does not interfere with $\widehat{x}_t$'s admissibility.}

The preceding method administers an admissible and consequently budget-feasible pair $\big\{\mathrm{proj}_{\widehat{\mathcal{A}}'^{\mathcal{P}}_{X_0}}\widehat{x}_t^{\mathrm{opt}},\widehat{c}_t^{\mathcal{P},\mathrm{opt}}\big\}_{t\in\left[0,T\right]}$ for any $\widehat{\beta}_t\in\mathcal{H}_{\widehat{\mathcal{A}}}^P\cap\mathcal{P}$. So as to determine $\widehat{\beta}_t$, consider 
\begin{equation}\label{eq:sec41dualgap}
\mathrm{D}_{\mathcal{P}}\left(\widehat{\theta}_t^{L,*},\widehat{\theta}_t^{U,*}\right):=\inf_{\widehat{\theta}_t^{U,*}\in\mathcal{H}_{\widehat{\mathcal{A}}}^P\cap\mathcal{P}\times\RR}J^U\left(X_0,\big\{ \widehat{\theta}_t^{U,*}\big\}\right)-\sup_{\widehat{\theta}_t^{L,*}\in\mathcal{H}_{\widehat{\mathcal{A}}}^P\cap\mathcal{P}\times\RR}\widehat{J^L}\left(X_0,\big\{ \widehat{\theta}_t^{L,*}\big\}\right),
\end{equation}
which delimits the smallest duality gap affiliated with the circumscription to $\mathcal{P}$  under $X_T^{*}$ in \eqref{eq:sec41approx}, where we separate the unspecified approximate controls on the primal side $\widehat{\theta}^{L,*}_t=\big(\widehat{\beta}_t^{L,*},Y_0^{L,*}\big)$ from those on the dual side $\widehat{\theta}^{U,*}_t=\big(\widehat{\beta}_t^{L,*},Y_0^{L,*}\big)$. We also let $\widehat{J^L}$ be the primal value function that ensues after insertion of $\big\{\mathrm{proj}_{\widehat{\mathcal{A}}'^{\mathcal{P}}_{X_0}}\widehat{x}_t^{\mathrm{opt}},\widehat{c}_t^{\mathcal{P},\mathrm{opt}}\big\}_{t\in\left[0,T\right]}$. This quantification of the smallest duality gap conditional on the outlined approximation underpins the notion that the approximate rules embed the undefined Lagrange multipliers, besides the shadow prices. After maximization of $\widehat{J^{L}}$, the approximations to the portfolio rules and consumption streams by $\mathrm{proj}_{\widehat{\mathcal{A}}'^{\mathcal{P}}_{X_0}}\widehat{x}_t^{\mathrm{opt}}$ and $\widehat{c}_t^{\mathcal{P},\mathrm{opt}}$, respectively, are in keeping with the most-right term in \eqref{eq:sec41dualgap} entirely identified via the successive $\widehat{\theta}^{L,*}_t$. 

Normally, we acquire $\widehat{\theta}^{U,*}_t$ analytically. Hence, as contrasted with the approach in \eqref{eq:sec41approx},
\begin{equation}\label{eq:sec53feasibprojreg}
X_T^{*}=X_0+\int_0^T\left(R_{f,t}X_t^{*}-\widehat{c}_{f,t}^{\mathcal{P},\mathrm{opt}}\right)\mathrm{d}t+\int_0^T\widehat{\mathrm{proj}}_{\widehat{\mathcal{A}}'^{\mathcal{P}}_{X_0}}\widehat{x}_t^{\mathrm{opt},\top}\widehat{\sigma}_tX_t^{*}\left(\mathrm{d}{W}_t+\widehat{\lambda}_t\mathrm{d}t\right),
\end{equation}
in which $\widehat{\mathrm{proj}}_{\widehat{\mathcal{A}}'^{\mathcal{P}}_{X_0}}\widehat{x}_t^{\mathrm{opt},\top}=\mathrm{proj}_{\widehat{\mathcal{A}}'^{\mathcal{P}}_{X_0}}\widehat{x}_t^{\mathrm{opt},\top}\big|_{\widehat{\theta}_t^{L,*}=\widehat{\theta}_t^{U,*}}$ and $\widehat{c}_{f,t}^{\mathcal{P},\mathrm{opt}}=\widehat{c}_{t}^{\mathcal{P},\mathrm{opt}}\big|_{\widehat{\theta}^{L,*}_t=\widehat{\theta}^{U,*}_t}$, formulates an alternative approximation to $X_T^{\mathrm{opt}}$ and $c_t^{\mathrm{opt}}$ that is fully analytical. Here, depending on the quality of the approximation to $\widehat{\beta}^{\mathrm{opt}}_t$ interlaced with the repression to $\mathcal{P}$, the approximation to the optimal rules authorises us in a meaningful manner to detour all possible numerical effort involved in optimising the primal objective $\widehat{J^L}$. We base the suitability as concerns the immediate injection of the analytically obtainable dual controls $\widehat{\theta}_t^{U,*}$ for $\widehat{\theta}_t^{L,*}$ into the approximate rules implicit in \eqref{eq:sec41approx} on the fact that $\widehat{\theta}^{L,*}_t=\widehat{\theta}^{U,*}_t$ in general holds only if $\mathrm{D}_P\big(\widehat{\theta}^{L,*}_t,\widehat{\theta}^{U,*}_t\big)=0$. This equality concretely demonstrates that practically identical approximate primal and dual parameters accompany trifling optimality gaps.

The exceptional instances in which the duality gap depletes to zero are when the restricted set $\mathcal{P}$ contains the truly optimal shadow prices $\widehat{\beta}_{t}$, or when we let $\mathcal{P}$ coincide with the original set $\DD^{1,2}\left(\left[0,T\right]\right)^n$. In all other cases, $\mathrm{D}_{\mathcal{P}}\big(\widehat{\theta}_t^{L,*},\widehat{\theta}_t^{U,*}\big)>0$ is true. Hence, to make the size of the duality gap tangible for both types of approximations in \eqref{eq:sec41approx} and \eqref{eq:sec53feasibprojreg}, we compute the so-called compensating variation. To that end, we examine
\begin{equation}
\widehat{J^{L,\mathrm{opt}}}\left(X_0+\mathcal{CV},\big\{ \widehat{\theta}_t^{L,*}\big\}\right)=\inf_{\widehat{\theta}_t^{U,*}\in\mathcal{H}_{\widehat{\mathcal{A}}}\times\RR}J^U\left(X_0,\big\{ \widehat{\theta}_t^{U,*}\big\}\right)+\mathrm{D}_{\mathcal{P}}\left(\widehat{\theta}_t^{L,*},\widehat{\theta}_t^{U,*}\right),
\end{equation}
wherein we let $\widehat{J^{L,\mathrm{opt}}}$ be the primal objective in \eqref{eq:sec41dualgap} resulting from either of the two approximate modes, corresponding to the initial endowment $X_0\in\RR_+$. Here, we denote the compensating variation by $\mathcal{CV}\in\RR_+$. The previous quantity defines the amount of capital that one must add to $X_0$ in order to overpass $\mathrm{D}_{\mathcal{P}}\big(\widehat{\theta}_t^{L,*},\widehat{\theta}_t^{U,*}\big)$. Economically, the magnitude of $\mathcal{CV}$ translates the utilitarian loss incurred due to the approximation into a monetary loss. The annual equivalent, $\mathcal{CV}^{\frac{1}{T}}$, abides by the interpretation of a management fee that protects the agent against shifts in the undiversifiability of $\left\lbrace W_t\right\rbrace_{t\in\left[0,T\right]}$.\footnote{We refer the reader to \citet{de2008utilitarianism} and \citet{de2009standardized} for details.}

\subsection{Stepwise Approximating Routine}\label{sec4.2}
We proceed by describing the approximate method for $\mathcal{P}\subseteq\DD^{1,2}\left(\left[0,T\right]\right)^n$ in terms of a stepwise routine involving Monte Carlo. Although the previous approximations to ${x}_t^{\mathrm{opt}}$ accommodated in Theorem \ref{thm2} are fully analytical, the approximate objective $\widehat{J^{L,\mathrm{opt}}}$ typically does not induce a closed-form formula. Considering further probable practical purposes, we amplify the routine. We first discuss the necessary notation and interconnected features. Let us rewrite the equality for approximate wealth, \eqref{eq:sec41approx} or \eqref{eq:sec53feasibprojreg}, as follows
\begin{equation}\label{eq:approxwealthprocsec42new}
\log\left(X_T^{*}\right)=\log\big(X_0\big)+\int_0^T\left(R_{f,t}-c_t^{*}-\frac{1}{2}\bar{x}_t^{\top,*}\bar{x}_t^{*}\right)\mathrm{d}t+\int_0^T\bar{x}_t^{\top,*}\left(\mathrm{d}W_{t}+\widehat{\lambda}_t\mathrm{d}t\right),
\end{equation}
where $\bar{x}_t^{\top,*}=x_t^{\top,*}\widehat{\sigma}_t^{-1}$ and $c^{*}_t=\widehat{c}_t^{\mathcal{P},\mathrm{opt}}{X_t^{*}}^{-1}$, in which $x_t^{*}$ is the approximation to $x_t^{\mathrm{opt}}$. Hence, we ensure that $X_t^*>0$ for all $0\leq t\leq T$, and eliminate the dependency on $X_T^*\big|_{T=t}$, which would otherwise affect the simulations. Note that $X_t^*$ differs from $\EE\big[\int_t^T\ddot{c}_{t,u}^{\mathrm{opt}}\mathrm{d}u+X_T^{\mathrm{opt}}{Y_{t,T}^{\mathrm{opt}}}{B_{t,T}}^{-1} \ \big| \ \eF_t\big]$ in Theorem \ref{thm2}, and that $\widehat{c}_t^{\mathcal{P},\mathrm{opt}}$ restores \textit{post factum} from $c_t^{*}X_t^{*}$. Bearing these facts in mind, we continue with the description of the method. 
\\\\
\textbf{Step 1. Initialisation of method.} We initialise an $N\in\NN$, denoting the number of paths for $\left\lbrace W_t\right\rbrace_{t\in\left[0,T\right]}$. In this way, we discretise $\Omega$ into $\omega_i\in\Omega$ for $i=1,\hdots,N$. Similarly, we fix an $M\in\NN$ representing $W_t$'s number of time increments on $\left[0,T\right]$. That is, we partition $\left[0,T\right]$ into $M$ equidistant intervals as $0=t_0<t_1<\hdots<t_M=T$ such that $\left|t_i-t_{i-1}\right|=T/M$ for $i\geq 0$. Conclusively, one may employ finite-difference methods to simulate the state variables according to the sample space induced by $N\times M$.
\\\\
\textbf{Step 2. Wealth dynamics.} Afterwards, we unscramble $X_T^{*}$ subsuming the approximate $x_t^{*}$ and $c_t^{*}$. From there, we simulate the finite-difference analogue to the completely analytical $\log\left(X_T^{*}\right)$ in \eqref{eq:approxwealthprocsec42new}. The resulting process is equipped with unspecified $\widehat{\theta}^{L,*}_t$. We endow state-dependent processes with $\omega_j$-notation. Approximate $X_T^*$ then agrees with
\begin{equation}\label{eq:approxnewmechanismlambda}
\begin{aligned}
\log\left(X_T^*\left(\omega_j\right)\right)&=\log\left(X_0\right)+\sum_{i=1}^M\bigg[{R}_{f,t_{i-1}}\left(\omega_j\right)+\bar{x}_{t_{i-1}}^{\top,*}\left(\omega_j\right)\lambda_{1,t_{i-1}}-c_t^{*}\left(\omega_j\right)\\&-\frac{1}{2}\bar{x}_{t_{i-1}}^{\top,*}\left(\omega_j\right)\bar{x}^*_{t_{i-1}}\left(\omega_j\right)\bigg]\Delta t_i+\sum_{i=1}^M\bar{x}_{t_{i-1}}^{\top,*}\left(\omega_j\right)\left(W^1_{t_{i}}\left(\omega_j\right)-W^1_{t_{i-1}}\left(\omega_j\right)\right)
\end{aligned}
\end{equation}
for $j=1,\hdots,N$, where $\Delta t_i=t_i-t_{i-1}$, and $\bar{x}^*_{t_i}$ complies with the approximating portfolio weights implicit in either ${\mathrm{proj}}_{\widehat{\mathcal{A}}'^{\mathcal{P}}_{X_0}}\widehat{x}_t^{\mathrm{opt}}$ or $\widehat{\mathrm{proj}}_{\widehat{\mathcal{A}}'^{\mathcal{P}}_{X_0}}\widehat{x}_t^{\mathrm{opt}}$. We simulate the log of $X_T^*$ so as to safeguard $X_T^*>0$, cf. \eqref{eq:approxwealthprocsec42new}. The former rules' discrete duplicates $\bar{x}^*_{t_i}$ depend on the unspecified $\widehat{\theta}^{L,*}_t$, and follow from equation \eqref{eq:thm2optrulportfol}, as non-linear transformations of the state variables. More, $X_{t_n}^{*}$ results from taking the sums in \eqref{eq:approxnewmechanismlambda} up to $n=1,\hdots,M$.
\\\\
\textbf{Step 3. Shadow price and multiplier.} We then continue with optimally determining the approximate shadow prices $\widehat{\beta}^{L,*}_{t}$ and the corresponding Lagrange multiplier $Y_0^{L,*}$, i.e. $\widehat{\theta}_{t}^{L,*}\in\mathcal{H}_{\widehat{\mathcal{A}}}^P\cap \mathcal{P}\times\RR$. In particular, we employ standard optimisation software to maximise the lower bound on $J^{\mathrm{opt}}$ engendered by the following primal value function 
\begin{equation}\label{eq:mcobjective1}
\widehat{J^{L}}=\frac{1}{N}\sum_{j=1}^N\sum_{i=1}^Mu\left(c_{t_i}^{*}\left(\omega_j\right)X_{t_i}^{*}\left(\omega_j\right),\Pi_{t_i}\left(\omega_j\right)\right)\Delta t_i+\frac{1}{N}\sum_{j=1}^NU\left(X_T^*\left(\omega_j\right),\Pi_T\left(\omega_j\right)\right),
\end{equation}
over $\big(\widehat{\beta}^{L,*}_{t},Y_0^{L,*}\big)\in\mathcal{H}_{\mathcal{A}}^P\cap \mathcal{P}\times\RR$. The unavailability of closed-form formulae in addition to analytically insuperable FOCs stimulates the usage of such numerical algorithms. Running times depend on the routine, and are as a result uninformative. Clearly, we assume at this point that we are in possession of $X_{t_n}^{*}$ for each $n=1,\hdots,M$. Simulating this $X_{t_n}^{*}$ is considered part of step 2. Notice that step 3 applies solely to ${\mathrm{proj}}_{\widehat{\mathcal{A}}'^{\mathcal{P}}_{X_0}}\widehat{x}_t^{\mathrm{opt}}$. Hence:\\
\begin{itemize}
\item For ${\mathrm{proj}}_{\widehat{\mathcal{A}}'^{\mathcal{P}}_{X_0}}\widehat{x}_t^{\mathrm{opt}}$, apply numerical maximisation of \eqref{eq:mcobjective1} initialised at $\widehat{\theta}^{L,*}_t=\widehat{\theta}^{U,*}_t$. \\
\item In case of $\widehat{\mathrm{proj}}_{\widehat{\mathcal{A}}'^{\mathcal{P}}_{X_0}}\widehat{x}_t^{\mathrm{opt}}$, we set $\widehat{\theta}^{L,*}_t=\widehat{\theta}^{U,*}_t$ \textit{a priori}, such that one may skip step 3.\\
\end{itemize} 
The first branch plainly acknowledges the non-unique nature of the controls $\widehat{\theta}_t^{L,*}$ contained within ${\mathrm{proj}}_{\widehat{\mathcal{A}}'^{\mathcal{P}}_{X_0}}\widehat{x}_t^{\mathrm{opt}}$, and therefore coincides with pursuing step 3 as outlined. Further, it adds an element that aims to precipitate any running time by initialising the method under scrutiny at the presumably analytically available dual controls $\widehat{\theta}_{t}^{U,*}$. The second item applies to $\widehat{\mathrm{proj}}_{\widehat{\mathcal{A}}'^{\mathcal{P}}_{X_0}}\widehat{x}_t^{\mathrm{opt}}$, for which all numerical computation is entirely eliminated, on account of its complete analytical characterisation via pre-fixed $\widehat{\theta}^{L,*}_t=\widehat{\theta}^{U,*}_t$.\footnote{The residual numerical effort involved with this approximate rule is due to steps 1 and 4. Step 4's running time is redundant. That of step 1 differs e.g. with dimensionality, see \citet{detemple2006asymptotic}.} Let us note that, in this step, we couple $c^{\mathcal{P},\mathrm{opt}}_t$ to ${\mathrm{proj}}_{\widehat{\mathcal{A}}'^{\mathcal{P}}_{X_0}}\widehat{x}_t^{\mathrm{opt}}$ and $c^{\mathcal{P},\mathrm{opt}}_{f,t}$ to $\widehat{\mathrm{proj}}_{\widehat{\mathcal{A}}'^{\mathcal{P}}_{X_0}}\widehat{x}_t^{\mathrm{opt}}$.\footnote{In the description, we use $c^{\mathcal{P},\mathrm{opt}}_{t}$. Replacing this approximation by $c^{\mathcal{P},\mathrm{opt}}_{f,t}$ only has an effect on step 3.}
\\\\
\textbf{Step 4. Performance evaluation.} The final step consists of evaluating the performance of the mechanism. Define $\widehat{\mathrm{D}_{\mathcal{P}}}\big(\widehat{\theta}^{L,*}_t,\widehat{\theta}^{U,*}_t\big):=\widehat{J^{U}}\left(X_0\right)-\widehat{J^{L,\mathrm{opt}}}\big(X_0\big)$, where $\widehat{J^{U}}$ and $\widehat{J^{L,\mathrm{opt}}}$ are respectively the approximate analytical dual and approximate estimated primal value functions. To recover $\mathcal{CV}$, we apply numerically facile root-finding modes to
\begin{equation}\label{eq:djdasddsd}
\widehat{J^{L,\mathrm{opt}}}\left(X_0+\mathcal{CV}\right)=J^{\mathrm{opt}}+\widehat{\mathrm{D}_{\mathcal{P}}}\left(\widehat{\theta}^{L,*}_t,\widehat{\theta}^{U,*}_t\right)=\widehat{J^{U}}\left(X_0\right),
\end{equation}
wherein we search over all $\mathcal{CV}\in\RR_+$ such that this equality holds. Observe that $\mathcal{CV}\in\RR_+$ is unique, because it linearly enacts on $X_T^{*}$ and $c_{t_n}^{*}$, which are both encompassed within the strictly increasing utility functions $u$ and $U$. The LHS of this equation also comprises the optimal $\widehat{\beta}^{L,*}_{t}$ and $Y_0^{L,*}$ from step 3 of the routine. This last step is obviously superfluous in light of the technique; it only serves to measure precision of the approximation.
\\\\
We construe the stepwise approximating routine under the presumption that the dual side of the problem renders analytical expressions for the dual controls $\widehat{\theta}^{U,*}$ and the upper value function $\widehat{J^{U}}$. However, these expressions are absent in Theorem \ref{thm2}, since the dual requires a full characterisation of $K\subseteq\RR^{d+m}$ before it deciphers FOCs for the approximate dual controls $\widehat{\theta}^{U,*}_t$. An application analogous to step 3 on the dual side discharges this postulate and engenders estimates to $\widehat{\theta}^{U,*}$ and $\widehat{J^{U}}$, which we may employ in steps 3/4. 

\subsection{Approximation in Incomplete Markets}\label{sec4.3}
This section finalises the analysis of the approximating procedure by casting the subject into the context of the economic setup and problem specification of section \ref{sec3.3}. Recall that the set of restrictions in that environment lives by $K=\RR^d\times\left\lbrace 0\right\rbrace^m$ such that $\widetilde{\mathcal{A}}_{X_0}^P=\mathcal{H}^P_{\widehat{\mathcal{A}}}=\left\lbrace -\lambda_{1,t}\right\rbrace \times\DD^{1,2}\left(\left[0,T\right]\right)^{m}$ define the spaces in reference to $\widehat{\beta}_t$. In the following analysis, we let $\mathcal{P}=\DD^{1,2}\left(\left[0,T\right]\right)^d\times\RR^m$ with the intention of portraying the mathematical comfort that the approximate technology could entail. This $\mathcal{P}$ embodies $\DD^{1,2}\left(\left[0,T\right]\right)^d$ to ensure $\beta_{1,t}=-\lambda_{1,t}$. Correspondingly, the approximate dual as per Theorem \ref{thm2} reads
\begin{equation}\label{eq:sec43approxdual}
\inf_{\widehat{\beta}_{2,t}^{U,*}\in\RR^m, Y_0^{U,*}\in\RR_+}\EE\left[V\left(Y_TB_T^{-1},\Pi_T\right)\right]+Y_0X_0,
\end{equation}
where we take advantage of $\mathcal{H}^P_{\widehat{\mathcal{A}}}\cap\DD^{1,2}\left(\left[0,T\right]\right)^d\times\RR^m=\left\lbrace -\lambda_{1,t}\right\rbrace^d\times\RR^m$. The confinement marked by $\mathcal{P}$ hence effectively ensures that the dual controls $\widehat{\beta}_{2,t}^{U,*}$ are only allowed to attain values on the real line $\RR^m$. By this regulation, we aspire to approximate $\widehat{\beta}_{2,t}^{\mathrm{opt}}=-\widehat{\lambda}_{2,t}^{\mathrm{opt}}$ in \eqref{eq:lambdaucharopt1} of Proposition \ref{prop3} by a constant in an attempt to avoid the computational strain enmeshed with solving its backward-forward equation, in a simple way.\footnote{For unconstrained $\mathcal{P}$, i.e. $\mathcal{P}=\DD^{1,2}\left(\left[0,T\right]\right)^d$, we find according to Proposition \ref{prop2} that $\beta_{1,t}=-\lambda_{1,t}$, which is completely identified. Ergo, there is no reason to confine $\beta_{1,t}$'s parameter space akin to $\widehat{\beta}_{2,t}$'s.}  

On the dual side, we are able to amass an analytic for optimal $\widehat{\beta}_{2,t}^{U,*}$. Specifically,
\begin{equation}
\begin{aligned}
D_{\left\lbrace\widehat{\beta}_{2,t}^{U,*}\right\rbrace}\widehat{\mathcal{W}}\left\lbrace\psi_t\right\rbrace=-\int_0^T\EE\left[\psi_t^{\top}\left\lbrace\EE\left[\mathcal{D}^{W^2}_tX_T^{\mathrm{opt}}Z_T^{\widehat{\nu}}\right]-\widehat{\beta}_{2,t}^{U,*}\EE\left[X_T^{\mathrm{opt}}Z_T^{\widehat{\nu}}\right]\right\rbrace\right]\mathrm{d}t=0,
\end{aligned}
\end{equation}
describes the FOC to the approximate dual \eqref{eq:sec43approxdual} in the $\widehat{\beta}_{2,t}^{U,*}$-direction, where we let $\widehat{\mathcal{W}}$ be the objective function in \eqref{eq:sec43approxdual}. As a result, the optimal approximate dual control must adhere to $\widehat{\beta}_{2,t}^{U,*}=\EE\big[X_T^{\mathrm{opt}}Z_T^{\widehat{\nu}}\big]^{-1}\EE\big[\mathcal{D}^{W^2}_tX_T^{\mathrm{opt}}Z_T^{\widehat{\nu}}\big]$, which is equal to the true FOC in Proposition \ref{prop3} for $\widehat{\beta}_{2,t}^{\mathrm{opt}}$ at its initialisation point in time $t=0$. In a like fashion, the FOC for the corresponding approximate dual multiplier $Y_0^{U,*}$ rises from the equality $\big\langle I\big(Y_0^{U,*}B_T^{-1}Y_T,\Pi_T\big){Y_0^{U,*}}^{-1}B_T^{-1}Y_T-X_0,\widehat{\kappa}\big\rangle_{L^2\left(\Omega\right)}=0$  for all $\widehat{\kappa}\in\RR_+$. In summary,
\begin{equation}\label{eq:sec43syseq}
\begin{aligned}
\EE\left[I\left(Y_0^{U,*}Z_T^{\widehat{\nu}},\Pi_T\right)Z_T^{\widehat{\nu}}\right]=X_0,\quad\mathrm{and}\quad
\widehat{\beta}_{2,t}^{U,*}=\left(1-\mathcal{R}^2_0\right)\mathcal{G}'^{2}_{0,T}-\mathcal{R}^2_t{\mathcal{R}^1_0}^{-1}\mathcal{G}'^{1}_{0,T}
\end{aligned}
\end{equation}
urges an analytical system of two non-linear equations that we are able to solve for the approximate dual shadow price $\widehat{\beta}_{2,t}^{U,*}$ and reciprocal Lagrange multiplier $Y_0^{U,*}$. In tandem with both preceding dual controls included in $\widehat{\theta}^{U,*}_t$, the approximate dual objective function $\widehat{J^{U}}\left(X_0\right)$ finds its definition in closed-form. The dual side of the method is on this account characterised by closed-form expressions for the controls and value function. 

Let us now turn to the primal side. For that purpose, consider $\widehat{x}_t^{\mathrm{opt}}$ in Proposition \ref{prop1} along with the artificial wealth equation in \eqref{eq:sec33fictbcwealtheq}. We transform this wealth process into an admissible analogue by suppressing the dependency on $x_{f,t}$, which then leads to 
\begin{equation}\label{eq:sec43approxtermwealth}
X_T^{*}=X_0+\int_0^T\left(R_{f,t}+x_t^{\mathrm{opt},*,\top}\sigma_t^{S_1}\lambda_{1,t}\right)X_t^{*}\mathrm{d}t+\int_0^Tx_t^{\mathrm{opt},*,\top}\sigma_t^{S_1}X_t^{*}\mathrm{d}W^1_t.
\end{equation}
Here, $x_t^{\mathrm{opt},*}$ contains the first $d$ analytical entries of $\widehat{x}_t^{\mathrm{opt}}$ in \eqref{eq:genoptrulesfict} for undefined $\widehat{\beta}_{t}$. This enforcement of the trading constraint guarantees feasibility of $X_t^{*}$, $0\leq t\leq T$. However, we cannot analytically determine the $\widehat{\beta}_{t}$ that optimises expected utility for these wealth dynamics. For tractability with regard to numerically determining this $\widehat{\beta}_{t}$, we confine the set to of admissible $\widehat{\beta}_{t}$ to all $\mathcal{H}_{\widehat{\mathcal{A}}}^P\cap\mathcal{P}$, for $\mathcal{P}$ as used in \eqref{eq:sec43approxdual}, such that $\beta_{1,t}=-\lambda_{1,t}$ and $\widehat{\beta}_{2,t}\in\DD^{1,2}\left(\left[0,T\right]\right)^m$. This approach coheres with the description in section \ref{sec4.1} under
\begin{equation}\label{eq:sec43projectkern}
\mathrm{proj}_{\widehat{\mathcal{A}}'^{\mathcal{P}}_{X_0}}x=\left[x_d,0_m\right]^{\top},\ \ \forall\ \  x=\left[x_d,x_m\right]\in  L^2\left(\Omega\times\left[0,T\right]\right)^{d+m},
\end{equation}
which dictates that $x_{f,t}=0_m$. Considering that the actual optimal portfolio weights $x_t^{\mathrm{opt}}$ in the baseline constrained environment $\mathcal{M}$, as given in Proposition \ref{prop3}, consist of $\widehat{x}_t^{\mathrm{opt}}$ including $\widehat{\beta}_t^{\mathrm{opt}}$ whereon $x_{f,t}=0_m$ is imposed, this projection pertaining to $\mathcal{P}$ is plausible.\footnote{The projection is case-dependent; e.g. in respect of borrowing and short-sale constraints, cf. \citet{tepla2000optimal}, i.e. $K=\left\lbrace x\in\RR^n_+\mid x^{\top}1_n\leq 1 \right\rbrace$, a reasonable projection would be ${\mathrm{proj}}_{\widehat{\mathcal{A}}'^{\mathcal{P}}_{X_0}}x=\max\left\lbrace 0_n,\|\sqrt{x}\|_{\RR^n}^{-2}x\right\rbrace$.} At this stage, the procedure is identical for ${\mathrm{proj}}_{\widehat{\mathcal{A}}'^{\mathcal{P}}_{X_0}}\widehat{x}_t^{\mathrm{opt}}$ and $\widehat{\mathrm{proj}}_{\widehat{\mathcal{A}}'^{\mathcal{P}}_{X_0}}\widehat{x}_t^{\mathrm{opt}}$.

The approximation to $X_T^{\mathrm{opt}}$ in Proposition \ref{prop3} by $X_T^{*}$ in \eqref{eq:sec43approxtermwealth} is budget-feasible and admissible in the sense of $\widehat{\mathcal{A}}_{X_0}$, since $x_t^{\mathrm{opt}}$ in the foregoing proposition is admissible for $\widehat{\beta}_{2,t}^{\mathrm{opt}}\in\DD^{1,2}\left(\left[0,T\right]\right)^m$ in \eqref{eq:lambdaucharopt1} implying that a reduction of  $\DD^{1,2}\left(\left[0,T\right]\right)^m$ to $\RR^m$ \textit{re} $\widehat{\beta}_{2,t}$'s space of controls does not hamper ${x}_{t}^{\mathrm{opt,*}}={\mathrm{proj}}_{\widehat{\mathcal{A}}'^{\mathcal{P}}_{X_0}}\widehat{x}_t^{\mathrm{opt}}\in\widehat{\mathcal{A}}_{X_0}$. These rules sprout from
\begin{equation}
\begin{aligned}
\mathcal{G}^{2}_{x,t,T}&=\mathcal{R}^{2,*}_t\EE\left[-\widehat{\mathcal{R}^{2,*}}_{t,T}\left(-\int_t^T\left[\mathcal{D}^{W^1}_t\left(R_{f,s}\mathrm{d}s+\widehat{\lambda}_{f,s}\mathrm{d}W_s^{\QQ^{\widehat{\nu}}}\right)\right]\right)\cond\eF_t\right]\\\mathcal{G}^{1}_{x,t,T}&=\mathcal{R}^{1,*}_t\EE\left[-\widehat{\mathcal{R}^{1,*}}_{t,T}\left(\int_t^T\left(\mathcal{D}^{W^1}_t\pi_s\mathrm{d}s+\mathcal{D}^{W^1}_t\widehat{\xi}_s\left(\mathrm{d}W_s-\widehat{\xi}_s\mathrm{d}s\right)\right)+\xi_t^{\Pi_1}\right)\cond\eF_t\right]
\end{aligned}
\end{equation}
inserted into ${x}_t^{\mathrm{opt},*}={{\sigma}_t^{S_1}}^{\top^{-1}}\big({\mathcal{R}^{1,*}_t}^{-1}\mathcal{G}^1_{x,t,T}+\big(1-{\mathcal{R}^{2,*}_t}^{-1}\big)\mathcal{G}^{2}_{x,t,T}+{\mathcal{R}^{2,*}_t}^{-1}{\lambda}_{1,t}\big)$, cf. Proposition \ref{prop3}. Independent of the selected approximation, either ${\mathrm{proj}}_{\widehat{\mathcal{A}}'^{\mathcal{P}}_{X_0}}\widehat{x}_t^{\mathrm{opt}}$ or $\widehat{\mathrm{proj}}_{\widehat{\mathcal{A}}'^{\mathcal{P}}_{X_0}}\widehat{x}_t^{\mathrm{opt}}$, closed-form expressions specify the approximate portfolio decisions ${x}_t^{\mathrm{opt},*}$, for any $\widehat{\theta}^{L,*}_t\in\mathcal{H}^P_{\widehat{\mathcal{A}}}\cap\mathcal{P}$. Related to the former projection, we determine $\widehat{\beta}_{2,t}^{L,*}$ and $Y_0^{L,*}$ by numerical maximisation of $\EE\left[U\left(X_T^{*},\Pi_T\right)\right]$ in the spirit of section \ref{sec4.2} in order to wholly identify $x_t^{\mathrm{opt},*}$. Proceeding from $\widehat{\mathrm{proj}}_{\widehat{\mathcal{A}}'^{\mathcal{P}}_{X_0}}\widehat{x}_t^{\mathrm{opt}}$, we take $\widehat{\theta}^{L,*}_t$ from the system motivated by \eqref{eq:sec43syseq} and insert these into ${x}_t^{\mathrm{opt},*}$ for full identification. In both cases, $\widehat{J^{L,\mathrm{opt}}}$ follows through \eqref{eq:mcobjective1}.

In ${x}_t^{\mathrm{opt},*}$, we recognise in the shape of \eqref{eq:sec31decomport} three portfolio demands submitting to 
\begin{equation}
\begin{aligned}
{x}^{Z,\mathrm{opt},*}_t&={{\sigma}_t^{S_1}}^{\top^{-1}}\frac{1}{1-\mathcal{R}_t^{2,*}}\mathcal{G}^{2}_{x,t,T},\quad\mathrm{and}\quad
{x}^{\Pi,\mathrm{opt},*}_t&={{\sigma}_t^{S_1}}^{\top^{-1}}\frac{1}{\mathcal{R}^{1,*}_t}\mathcal{G}^{2}_{x,t,T}
\end{aligned}
\end{equation}
consistent with the decomposition $x_t^{\mathrm{opt},*}=x_t^{m,\mathrm{opt},*}+x_t^{Z,\mathrm{opt}}+x_t^{\Pi,\mathrm{opt},*}$, in which the mean-variance efficient portfolio rule conforms to $x_t^{m,\mathrm{opt}}={\mathcal{R}_t^{2,*}}^{-1}{{\sigma_t^{S_1}}^{\top}}^{-1}\lambda_{1,t}$. These hedging demands are up to the RRA coefficients that incorporate deterministic constants for the shadow prices of non-traded risk equivalent to those presented in Proposition \ref{prop1}: we empower the $\mathcal{R}$ terms with an asterisk to refine the difference from the true expressions. Let us also observe that $\mathcal{D}^{W^1}_t\widehat{\lambda}_{f,s}=\big[\mathcal{D}^{W^1}_t\lambda_{1,s},0_{d\times m}\big]$, which implies the deterministic character and the total independence from $W^1_t$ of the approximate shadow prices.

\section{Numerical Illustration}\label{sec5}
We complete the examination of the approximate method with a numerical analysis that aspires to appraise the accuracy of the technique in an explicit economic framework that involves a definite problem setup. With this in view together with the antecedent scrutiny of the incomplete markets setting, we discuss results in the environment of \citet{brennan2002dynamic}.\footnote{This financial market model appears as a special case of the constrained $\mathcal{M}$ in section \ref{sec3.3}.} At the heart of this illustration lies our dual CRRA utility function, which constitutes a novel addition to the existing body of state-dependent preference conditions. This qualification replaces the CRRA function in the economic outline at hand.  

\begin{figure}[!t]
\begin{center}
\includegraphics[scale=1.13]{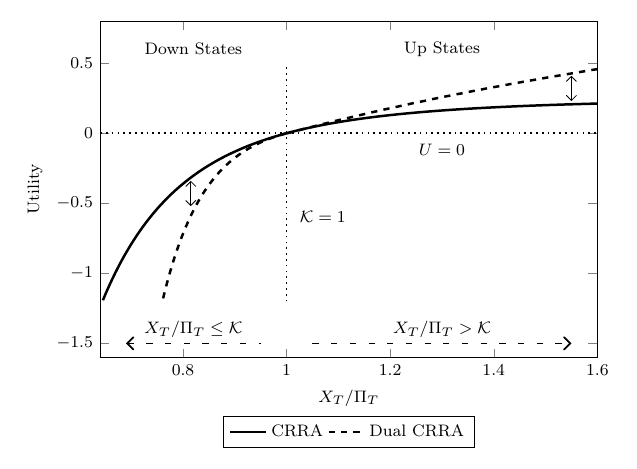}
\end{center}
\caption{\label{f1}\textbf{Isoelastic and dual CRRA utility.} The graph depicts the utility development for an isoelastic agent and a dual CRRA agent. The matching risk profiles adhere to $\gamma=5$ and $\left(\gamma_d,\gamma_u\right)=\left(10,2\right)$ for the respective CRRA and dual CRRA investors. The reference level for the dual CRRA individual reads $\mathcal{K}=1$ (vertically dotted line). At $\mathcal{K}=1$, this individual draws zero utility, $U=0$ (horizontally dotted line). Down states of the world tally with all $\omega\in\Omega$ wherein $X_T/\Pi_T\leq \mathcal{K}$. Conversely, all $\omega\in\Omega$ in which $X_T/\Pi_T> \mathcal{K}$ coincide with the up states, as separated by the vertically dotted line $\mathcal{K}=1$.}
\end{figure}

\subsection{Dual CRRA Utility Function}\label{sec5.1}
We suppose that the agent from section \ref{sec3.3} has in mind a benchmark, $\mathcal{K}\in\RR_+$, regarding his or her horizon purchasing power, $X_T/\Pi_T$. For states in which $X_T/\Pi_T\leq \mathcal{K}$, the agent draws proportionately less utility than in reversed states. Therefore, we attach $\mathcal{K}$ to $U$, on account of which the \textit{dual CRRA} function incorporates two isoelastic qualifications
\begin{equation}\label{eq:dualCRRA}
U_{\mathcal{K}}\left(X_T,\Pi_T\right)=\frac{\left(\mathcal{K}^{-1}X_T/\Pi_T\right)^{1-\gamma_d}-1}{1-\gamma_d}\Ind_{\left\lbrace\frac{X_T}{\Pi_T}\leq \mathcal{K}\right\rbrace}+\frac{\left(\mathcal{K}^{-1}X_T/\Pi_T\right)^{1-\gamma_u}-1}{1-\gamma_u}\Ind_{\left\lbrace\frac{X_T}{\Pi_T}> \mathcal{K}\right\rbrace}
\end{equation}
for coefficients of risk aversion $\gamma_d,\gamma_u\in\RR_+\setminus\left[0,1\right]$ to forgo infinite utility, cf. \cite{kramkov1999asymptotic}. In order to convene \eqref{eq:dualCRRA} with the portrayed situation, we dictate that $\gamma_d\geq\gamma_u$. Hereby, we argue that down states correspond to an enhanced apprehension regarding the attainment with reference to the exogenous $\mathcal{K}$.\footnote{This is not a necessary requirement, any $\gamma_d,\gamma_u\in\RR_+\setminus\left[0,1\right]$ ensure proper specification of \eqref{eq:dualCRRA}.} The investor under study retains total flexibility to choose $\mathcal{K}$ \textit{ante initium}, before $t=0$, resulting in $\mathcal{K}$ being fundamentally exogenous over the full course of $\left[0,T\right]$. Without loss of generality, we normalise $\mathcal{K}=1$, such that $U_{\mathcal{K}}\mid_{\mathcal{K}=1}:=U$ coalesces with the desires of an individual who wishes to maintain a constant degree of purchasing power.\footnote{The subsequent breakdown bears upon a sole terminal wealth setup; we restrict this analysis to $U$.}

Technically, we observe that $U\left(X_T,\Pi_T\right)$ is continuous at $X_T=\Pi_T$ and state-dependent by virtue of the accommodated separation of instances. Furthermore, we find that 
\begin{equation}\label{eq:dualCRRA2}
U'_X=\begin{cases}
X_T^{-\gamma_d}\Pi_T^{-1+\gamma_d},\ \mathrm{if}\quad \frac{X_T}{\Pi_T}\leq 1\\
X_T^{-\gamma_u}\Pi_T^{-1+\gamma_u},\ \mathrm{if}\quad \frac{X_T}{\Pi_T}> 1
\end{cases};\ I =\begin{cases}
X_T^{-\frac{1}{\gamma_d}}\Pi_T^{1-\frac{1}{\gamma_d}},\ \mathrm{if}\quad \frac{X_T}{\Pi_T}\geq 1\\
X_T^{-\frac{1}{\gamma_u}}\Pi_T^{1-\frac{1}{\gamma_u}},\ \mathrm{if}\quad \frac{X_T}{\Pi_T}< 1
\end{cases}
\end{equation}
define respectively marginal utility in the first argument and its inverse, both of which are continuous at $X_T=\Pi_T$ as well. The second derivative of $U$ in the first argument, $U''_{XX}$, follows straightforwardly. Note that $U$ satisfies once continuous differentiability in the $X_T$-direction. Further, the inverse of marginal utility is piecewise continuously differentiable with a single breakpoint at $X_T=\Pi_T$. The mathematical description of the dual CRRA stipulation hence aligns with section \ref{sec2.3}'s. Finally, the instalment of two CRRA utilities in \eqref{eq:dualCRRA} occasions a qualification that displays non-constant RRA:\footnote{Optimal portfolios ought to conform to a structure that concurs per instance with the CRRA demands. }
\begin{equation}\label{eq:nonCRRA}
1-I'_Y\Pi_T=-\left(X_T\frac{U''_{XX}}{U'_X}\right)^{-1}=\frac{1}{\gamma_d}\Ind_{\left\lbrace \frac{X_T}{\Pi_T}\leq 1\right\rbrace} + \frac{1}{\gamma_u}\Ind_{\left\lbrace \frac{X_T}{\Pi_T}>1\right\rbrace}.
\end{equation}
 
In Figure \ref{f1}, we exemplify the evolution of dual CRRA preferences, characterised by $\gamma_d=10$ and $\gamma_u=2$, jointly with those for a CRRA agent, rising from $\gamma_d=\gamma_u=5$ in \eqref{eq:dualCRRA}. The discrepancy between the trajectories displays itself in a salient tone for the dual CRRA individual by a comparatively intensified utility development in `up states' ($X_T/\Pi_T>\mathcal{K}$) and a relatively abated progress in `down states' ($X_T\Pi_T\leq\mathcal{K}$) of the function. Thus, this utility-maximising investor is indispensably concerned about achieving $\mathcal{K}$ in down states, whereas the same individual substantially loosens his or her anxiety in up states. 

\begin{corollary}\label{col1}
Consider $\sup_{\left\lbrace{x}_t\right\rbrace_{t\in\left[0,T\right]}\in\widehat{\mathcal{A}}_{X_0}'}\EE\left[U\left(X_T,\Pi_T\right)\right]$ subject to the dynamic constraint implicit in \eqref{eq:sec33fictbcwealtheq} for an investor with dual CRRA preferences \eqref{eq:dualCRRA}. Then, 
\begin{equation}\label{eq:sec33thm33ladgrmdsal}
\begin{aligned}
X_T^{\mathrm{opt}}&=\left(\eta Z_T^{\widehat{\nu}}\Pi_T\right)^{-\frac{1}{\gamma_d}}\Pi_T\Ind_{\left\lbrace \eta Z_T^{\widehat{\nu}}\Pi_T\geq 1\right\rbrace}+\left(\eta Z_T^{\widehat{\nu}}\Pi_T\right)^{-\frac{1}{\gamma_u}}\Pi_T\Ind_{\left\lbrace \eta Z_T^{\widehat{\nu}}\Pi_T< 1\right\rbrace}\\
\mathcal{H}\left(\eta\right)&=\eta^{-\frac{1}{\gamma_d}}\EE\left[\widehat{M_T}^{1-\frac{1}{\gamma_d}}\right]\XX_{\gamma_d}^{\lambda_2}\left(\eta\widehat{M_T}\geq {1}\right)+\eta^{-\frac{1}{\gamma_u}}\EE\left[\widehat{M_T}^{1-\frac{1}{\gamma_u}}\right]\XX_{\gamma_u}^{\lambda_2}\left(\eta\widehat{M_T}< {1}\right),
\end{aligned}
\end{equation}
in which $\XX_{\gamma_i}^{\lambda_2}\sim\PP$ are the measures generated by $\XX_{\gamma_i}^{\lambda_2}=\EE\left[\widehat{M_T}^{1-{1}/{\gamma_i}}\right]^{-1}\widehat{M_T}^{1-{1}/{\gamma_i}}$ for $i\in\mathcal{S}:=\left\lbrace d,u\right\rbrace$. The applicable, previsible RRA transformation then must follow 
\begin{equation} 
\begin{gathered}\label{eq:sec3.3thm3.38ed}
\frac{1}{\mathcal{R}^2_t}=\frac{1}{\gamma_d}\XX^{\lambda_2}_{\gamma_d}\left(Y_T^{\mathrm{opt}}\geq B_T\cond\eF_t\right)\frac{X_t^{\gamma_d}}{X_t^{\mathrm{opt}}}+\frac{1}{\gamma_u}\XX^{\lambda_2}_{\gamma_u}\left(Y_T^{\mathrm{opt}}< B_T\cond\eF_t\right)\frac{X_t^{\gamma_u}}{X_t^{\mathrm{opt}}}, \  \mathrm{where}\\ X_t^{\gamma_i}=\eta^{-1/\gamma_i}\EE\left[\widehat{M_T}^{1-1/\gamma_i}\cond\eF_t\right], \ X_t^{\mathrm{opt}}=\EE\left[Z_{t,T}^{\widehat{\nu}}X_T^{\mathrm{opt}}\cond\eF_t\right], \ \mathrm{and} \ Z_{t,T}^{\widehat{\nu}}={Z_t^{\widehat{\nu}}}^{-1}Z_T^{\widehat{\nu}},
\end{gathered}
\end{equation}
such that ${\mathcal{R}_t^1}^{-1}={\mathcal{R}_t^2}^{-1}-1$ defines the proxy for transformed RRA, for all $t\in\left[0,T\right]$.\footnote{In accordance with Theorem \ref{thm2}, we use $Y_0^{-1}B_t^{-1}Y_t=Z_t^{\widehat{\nu}}$. Additionally, $Y_t^{\mathrm{opt}}$ embeds $\{\widehat{\lambda}_{2,s}^{\mathrm{opt}}\}_{s\in\left[0,t\right]}$.}
\end{corollary} 
\begin{proof}
The statements and associated proof rise effortlessly from Proposition \ref{prop3}.
\end{proof}

Corollary \ref{col1} carries the expressions that establish the optimality conditions in line with Proposition \ref{prop3} for a dual CRRA agent \eqref{eq:dualCRRA}. Note that closed-form expressions for the optimal portfolio $x_t^{\mathrm{opt}}$ and the market fair value of wealth $X_t^{\mathrm{opt}}$ are readily obtainable, as a result. The previous corollary chiefly discloses the analytical burden that may arise with regard to a recuperation of the optimal shadow prices of non-traded risk. Let us actualise this mathematical distress by inspecting the next expression 
\begin{equation}\label{eq:sec33shadowprice2}
\widehat{\lambda}_{2,t}^{\mathrm{opt}}\left(\mathcal{G}'^{2}_{t,T}-\mathcal{G}'^{1}_{t,T}\right)^{-1}=1-X_t^{\mathrm{opt}}\bigg(\sum_{i\in\mathcal{S}}\frac{1}{\gamma_i}\XX^{\lambda_2}_{\gamma_i}\left(\mathcal{A}_i\cond\eF_t\right)X_t^{\gamma_i}\bigg)^{-1},
\end{equation}
which typifies the identity from which we can obtain the shadow prices, where we use $\mathcal{R}^2_t\left(\mathcal{R}^1_t\right)^{-1}=1-\mathcal{R}_t^2$, for $\mathcal{A}_d=\{ Y_T^{\mathrm{opt}}\geq B_T\}$, $\mathcal{A}_u=\{ Y_T^{\mathrm{opt}}< B_T\}$. The probability weights each depend on $M_T$ and non-linearly shelter $\{\widehat{\lambda}_{2,t}^{\mathrm{opt}}\}_{t\in\left[0,T\right]}$. This path dependency epitomises that a closed-form expression for $\widehat{\lambda}_{2,t}^{\mathrm{opt}}$ is not attainable. In addition, $\widehat{\lambda}_{2,t}^{\mathrm{opt}}$ is restricted to an interval that diverges with the variation between $\gamma_d$ and $\gamma_u$.\footnote{The formula in \eqref{eq:sec33shadowprice2} demonstrates that $\widehat{\lambda}_{2,t}^{\mathrm{opt}}\in[(1-\gamma_d),(1-\gamma_u)](\mathcal{G}'^{1}_{t,T}+\mathcal{G}'^{2}_{t,T})$ roughly holds for all $\omega\in\Omega$ and $t\in\left[0,T\right]$. Evidently, $\gamma_d\rightarrow\gamma_u$ could assist in reducing the size of this approximate interval.} Consequently, fixing $\mathcal{P}=\RR^m$ may provide a level-headed approximation to $\widehat{\lambda}_{2,t}^{\mathrm{opt}}$ for tenable $\gamma_d$, $\gamma_u$.  

Letting $\gamma:=\gamma_d=\gamma_u$ recovers the standard isoelastic framework and provides 
\begin{equation}
X_T^{\mathrm{opt}}=X_0\EE\left[\widehat{M_T}^{1-{1}/{\gamma}}\right]\widehat{M_T}^{-{1}/{\gamma}}\Pi_T,\ \ \mathrm{and} \ \ \widehat{\lambda}_{2,t}^{\mathrm{opt}}=\left(1-\gamma\right)\left(\mathcal{G}'^{2}_{t,T}-\mathcal{G}'^{1}_{t,T}\right),
\end{equation}
from which we infer that the shadow prices still mark off troublesome wholly forward-backward equations, impelled by the fictitious hedging coefficients $\mathcal{G}'^{2}_{t,T}$ and $\mathcal{G}'^{1}_{t,T}$. As a result, the latter equality does not allow us to withdraw $\widehat{\lambda}_{2,t}^{\mathrm{opt}}$ in closed-form. Under the supplementary premises that $\pi_t$ and $r_t$ are $\eF_t^{W_1}$-measurable, and that $\xi_t^{\Pi_2}$ defines a constant, as in for example \cite{brennan2002dynamic}, the optimal shadow prices would characterise constants as well: $\widehat{\lambda}_{2,t}^{\mathrm{opt}}=\left(1-\gamma\right)\xi^{\Pi_2}_t$. By consequence, the optimal decisions $x_t^{\mathrm{opt}}$ explicitly implant the two-fund separation principle, and break down into
\begin{equation}
x_t^{\mathrm{opt}}=\frac{1}{\gamma}{\sigma_t^{S_1}}^{\top^{-1}}\lambda_{1,t}+\left(1-\frac{1}{\gamma}\right){\sigma_t^{S_1}}^{\top^{-1}}\left(\mathcal{G}_{x,t,T}^2-\mathcal{G}_{x,t,T}^1\right),
\end{equation}
where $\mathcal{G}_{x,t,T}^j$, $j=1,2$ are completely spelled out thanks to the analytical existence of $\widehat{\lambda}^{\mathrm{opt}}_{2,t}$. This mathematical convenience provides a proper point of departure for the numerical verification of the approximation technology underpinning $\mathcal{P}\subseteq\DD^{1,2}\left(\left[0,T\right]\right)^m$.

\begin{figure}[!t]
\begin{center}
\includegraphics[scale=1.13]{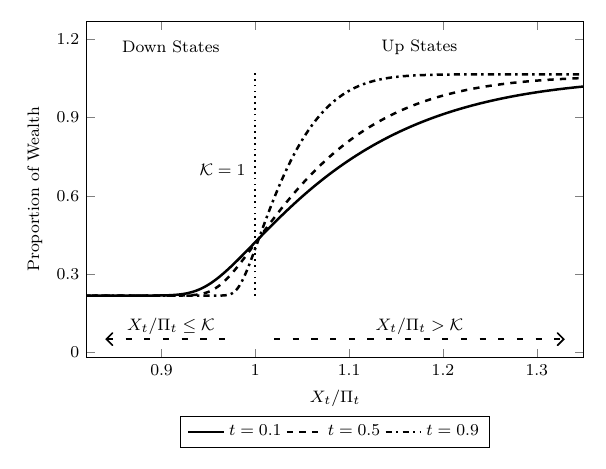}
\end{center}
\caption{\label{f2}\textbf{Optimal allocation to stock.} The figure depicts the optimal portfolio demand in relation to real wealth. To that end, the plot relies on a reduced economic setting for a dual CRRA agent characterised by $\left(\gamma_d,\gamma_u\right)=\left(10,2\right)$ and a reference level $\mathcal{K}=1$ (vertically dotted line). This economy contains a one-dimensional stock, for $\lambda_{1,t}=0.343$ and $\sigma_t^{S_1}=0.158$, wherein $r_t=\pi_t=0$ and $\Pi_t=1$, so that $X_t/\Pi_t$ reads $X_t$. Further, we set $T=1$. Corollary \ref{col1}'s statements suffice to derive the optimality conditions. Here, ${\gamma_d}^{-1}{\lambda_{1,t}}/{\sigma_t^{S_1}}$ and ${\gamma_u}^{-1}{\lambda_{1,t}}/{\sigma_t}^{S_1}$ define the portfolio decisions' floor and cap, respectively. }
\end{figure}

Let us return to the general dual CRRA case in which $\gamma_d,\gamma_u\in\RR_+\setminus\left[0,1\right]$. Revisiting Corollary \ref{col1} for such $\gamma_d,\gamma_u$, we detect an omnipresent weighting texture that characterises the optimal controls, $\widehat{x}_t^{\mathrm{opt}}$ and $X_t^{\mathrm{opt}}$, $0\leq t\leq T$. To underscore this structure and related idiosyncrasy, we \textit{pro tem} restrict the analysis to the mean-variance hedge demand:
\begin{equation}\label{eq:optruledualcrra}
x_t^{m,\mathrm{opt}}=\frac{1}{\gamma_d}{{{\sigma_t^{S_1}}^{\top}}^{-1}\lambda_{1,t}}\XX^{\lambda_2}_{\gamma_d}\left(\mathcal{A}_d\cond\eF_t\right)\frac{X_t^{\gamma_d}}{X_t^{\mathrm{opt}}}+\frac{1}{\gamma_u}{{{\sigma_t^{S_1}}^{\top}}^{-1}\lambda_{1,t}}\XX^{\lambda_2}_{\gamma_u}\left(\mathcal{A}_u\cond\eF_t\right)\frac{X_t^{\gamma_u}}{X_t^{\mathrm{opt}}}.
\end{equation}
These decisions corner wealth-dependent dynamics in consequence of the reciprocal inclusion of $X_t^{\mathrm{opt}}$, which could evade if and only if $\gamma_d=\gamma_u$. The weighting and interpolating anatomy crops up within the projected RRA as a linear combination of the individually detached CRRA demands $\frac{1}{\gamma_d}{{\sigma_t^{S_1}}^{\top}}^{-1}\lambda_{1,t}$ and $\frac{1}{\gamma_u}{{\sigma_t^{S_1}}^{\top}}^{-1}\lambda_{1,t}$ along with the previsible conditional probabilities.\footnote{Absolute continuity infers that $\frac{1}{\gamma_d}({{\sigma_t^{S_1}})^{\top}}^{-1}\lambda_{1,t}$, $\frac{1}{\gamma_u}({{\sigma_t^{S_1}})^{\top}}^{-1}\lambda_{1,t}$ emerge from $\XX^{\lambda_2}_{\gamma_d}\rightarrow 1$, $\XX^{\lambda_2}_{\gamma_u}\rightarrow 1$.} Expressly, for down states of the world, $\mathcal{A}_d$, the tangency rule in \eqref{eq:optruledualcrra} assigns more weight to $\frac{1}{\gamma_d}{{\sigma_t^{S_1}}^{\top}}^{-1}\lambda_{1,t}$ so as to prevent the investor from further dissatisfaction by executing a prudent investment strategy. In up states, $\mathcal{A}_u$, the converse is true, and $x_t$ spells out a less cautious strategy, due to achieved acquisition of $\mathcal{K}=1$. 

Per illustration of the most typical features, Figure \ref{f2} displays $x_{t}^{\mathrm{opt}}$ in an economy where the exclusive risky asset abides by one-dimensional $S_t$ that embeds a constant appreciation rate and volatility coefficient. The plot applies to an agent, personified by $\gamma_d=10$ and $\gamma_u=2$. We provide the allocation in relation to $X_t/\Pi_t$, which is deducible from 
\begin{equation}\label{eq:markvaluesec3}
X_t^{\mathrm{opt}}Z_t^{\widehat{\nu}}=\sum_{i\in\mathcal{S}}\eta^{-1/\gamma_i}\EE\left[\widehat{M_T}^{1-1/\gamma_i}\cond\eF_t\right]\XX^{\lambda_2}_{\gamma}\left(\mathcal{A}_i\cond\eF_t\right)
\end{equation}
All curves in this plot exhibit an interpolating design between the floor and cap around $\mathcal{K}=1$, resembling shifts from prudent to less cautious investment behaviour in line with Figure \ref{f1}.\footnote{The CRRA exposures $\frac{1}{\gamma_d}({{\sigma_t^{S_1}})^{\top}}^{-1}\lambda_{1,t}$, $\frac{1}{\gamma_u}({{\sigma_t^{S_1}})^{\top}}^{-1}\lambda_{1,t}$ do not always provoke the floors and caps of $x_t^{\mathrm{opt}}$ in the economy of Figure \ref{f2}. In the limit, however, the portfolio converges to these extrema.} Following the payoff of a European option, the cutoff points approach $\mathcal{K}=1$ for $t\rightarrow T$, as the investor's concerns and harmonious rapid anticipation increase towards the planning horizon. The decisions indeed show trend-chasing behaviour, cf. \citet{basak2002comparative}. Akin, though `fuzzier', shapes apply to the more general multi-dimensional instances.

\subsection{Brennan and Xia Environment}\label{sec5.2}
In order to facilitate the numerical illustration, which corroborates precision of the approximate technique and relies on \citet{brennan2002dynamic}'s economic environment, we introduce their model setup and pose congruent approximate optimality conditions in the spirit of section \ref{sec4}. To align notation, we preserve the probabilistic schema in section \ref{sec2.1}, wherein we replace $W_t=\left[W^1_t,W^2_t\right]^{\top}$ by the $\RR^4$-valued $\widehat{z}_t=\left[z_t,z_{u,t}\right]^{\top}$. Further, introduce
\begin{equation}
\begin{aligned}
\mathrm{d}r_t&=\kappa\left(\bar{r}-r_t\right)\mathrm{d}t+\sigma_r\mathrm{d}z_{r,t}, \ r_0\in\RR_+ \\ \mathrm{d}\pi_t&=\alpha\left(\bar{\pi}-\pi_t\right)\mathrm{d}t+\sigma_{\pi}\mathrm{d}z_{\pi,t}, \ \pi_0\in\RR_+
\end{aligned}
\end{equation}
as two single-factor Vasicek processes for the instantaneous real interest rate and the expected rate of inflation, respectively. We assume that that the assimilated parameters live by the typical definitions. These processes unequivocally meet the regularity conditions imposed on $r_t,\pi_t$ in the generic financial model construction of section \ref{sec2.1}.

Specifically, we identify the $\RR^3$-valued standard Brownian motion as $z_t=\left[z_{s,t},z_{r,t},z_{\pi,t}\right]^{\top}$. Hence, $z_{u,t}$ represents an undiversifiable source of risk. Correspondingly, we introduce
\begin{equation}\label{eq:sec52pricandspd}
\begin{aligned}
\mathrm{d}\Pi_t&=\Pi_t\left(\pi_t\mathrm{d}t+\xi^{\top}z_t+\xi_u\mathrm{d}z_{u,t}\right), \ \Pi_0=1 \\ \mathrm{d}M_t&=M_t\left(-r_t\mathrm{d}t+\phi^{\top}\mathrm{d}z_t+\phi_{u}\mathrm{d}z_{u,t}\right), \ M_0=1,
\end{aligned}
\end{equation}
as the simplified price index and real SPD, where $\xi=\left[\xi_s,\xi_r,\xi_{\pi}\right]^{\top}\in\RR^3$ and $\phi=\left[\phi_{s},\phi_r,\phi_{\pi}\right]^{\top}\in\RR^3$ characterise the constant factor loadings on $z_t$. In agreement with these dynamics, we let $\mathrm{d}Z_t^{\widehat{\lambda}_u}=Z_t^{\widehat{\lambda}_u}[-R_{f,t}\mathrm{d}t-{\lambda}^{\top}\mathrm{d}z_t-\widehat{\lambda}_{u,t}\mathrm{d}z_{u,t}]$, $Z_0^{\widehat{\lambda}_u}=1$, where $Z_t^{\widehat{\lambda}_u}$ represents the artificial nominal SPD, for an endogenous $\widehat{\lambda}_{u,t}\in\DD^{1,2}\left(\left[0,T\right]\right)$. Here, $z_t$ is independent from $z_{u,t}$ and internally dependent, implied by the correlation matrix $\rho\in\left[-1,1\right]^{3\times 3}$, whose rows equate to $\left(1,\rho_{s,r},\rho_{s\pi}\right)$,$\left(\rho_{sr},1,\rho_{s\pi}\right)$, $\left(\rho_{s\pi},\rho_{s,r},1\right)$.

The actual nominal pricing kernel surfaces from $Z_t=M_t/\Pi_t$ and in part urges the nominal rate $R_{f,t}=r_t+\pi_t-\xi^{\top}\lambda-\xi_u\lambda_{u}$. The asset menu\footnote{The nominal bonds crop up as a consequence of the construction: $P_{i,t}=\EE[Z_{t}^{-1}Z_{T_i}|\eF_t]$, $T_i\geq t$. Observe that we also include $\mathrm{d}B_t=R_{f,t}B_t\mathrm{d}t, \ B_0=1$ in addition to $S_t$ and $P_{i,t}$ as part of the asset mix.} then consists of
\begin{equation}
\begin{aligned}
 \mathrm{d}S_t&=S_t\left[\left(R_{f,t}+\lambda_s\sigma_s\right)\mathrm{d}t+\sigma_s\mathrm{d}z_{s,t}\right],\ S_0=1 \\ \mathrm{d}P_{i,t}&=P_{i,t}\big[\big(R_{f,t}-\mathcal{\upsilon}_{t,T_i}^{\top}\widehat{\sigma}^P{\lambda}^P\big)\mathrm{d}t-\upsilon_{t,T_i}^{\top}\widehat{\sigma}^P\mathrm{d}\widehat{z}_t^P\big], \ P_i=1,
\end{aligned}
\end{equation}
which respectively specify a stock, and two nominal bonds that differ on account of two times to maturity $T_i\in\RR_+$, $i=1,2$. In these stochastic differential equations, we define $\mathrm{vec}(\widehat{\sigma}^P)=[\sigma_r,0,0,\sigma_{\pi}]$, ${\lambda}^P=[\lambda_r,\lambda_{\pi}]^{\top}$, and $\upsilon_{t,T_i}=[\frac{1}{\sigma_r}(1-e^{-\kappa(T-t)}),\frac{1}{\sigma_{\pi}}(1-e^{-\alpha(T-t)})]^{\top}$ for notational appliance. Moreover, the market construction mandates that the constant prices of financial risk obey $\lambda=[\lambda_s,\lambda_r,\lambda_{\pi}]^{\top}=\rho\left(\xi-\phi\right)$ and $\lambda_{u}=\xi_u-\phi_{u}$.

At last, we introduce the following two dynamic processes given that $\Lambda_t=\Sigma_t\lambda$:
\begin{equation}
\begin{aligned}
\mathrm{d}X_t&=X_t\left[\left(R_{f,t}+x_t^{\top}\Lambda_t\right)\mathrm{d}t+x_t^{\top}\Sigma_t\mathrm{d}z_t\right], \ X_0\in\RR_+,\\
\mathrm{d}P_{3,t}&=P_{3,t}\big[\big(R_{f,t}+\sigma_{P_3,t}\widehat{\lambda}_{u,t}\big)\mathrm{d}t+\sigma_{P_3,t}\mathrm{d}z_{u,t}\big],\ P_{3,0}=1.
\end{aligned}
\end{equation}
The first process comprises the agent's wealth dynamics, where we define $x_t$ as the $\RR^3$-valued $\eF_t$-measurable vector containing the proportions of $X_t$ that the agent allocates to the risky instruments. Here, we assume that the portfolio weights satisfy $x_t^{\top}\Sigma_t\Sigma_t^{\top}x_t\in L^1\left(\left[0,T\right]\right)$, where $\Sigma_t\in\RR^{3\times 3}$ accommodates the securities' loadings on $z_t$. Admissibility of $x_t$ holds, if additionally $X_t\geq 0$ for all $t\in\left[0,T\right]$. The second process comprises the fictitious asset, to which $x_{u,t}$ of $X_t$ is allocated. Admissibility of $x_{u,t}$ follows evenly from that of $x_t$.

\subsubsection{Approximate Dual Optimality}
First, we analyse the approximate dual side of the terminal wealth problem in the previous economy. For that purpose, consider Proposition \ref{prop1} and curtail the space of feasible dual controls $\widehat{\lambda}_{u,t}$ to $\mathcal{P}=\RR\subset \DD^{1,2}\left(\left[0,T\right]\right)$. We intuitively verify this repression of the space in the sequel. Tantamount to the results in section \ref{sec4.3}, this operation makes sure that the dual problem generates analytical solutions, which gives rise to Corollary \ref{col2}.  

\begin{corollary}\label{col2} Consider $\inf_{\widehat{\lambda}_{u,t}\in\RR,\eta\in\RR_+}\EE[V(\eta Z_T^{\widehat{\lambda}_u},\Pi_T )]+\eta X_0 $ for $\widehat{\lambda}_{u,t}\in \mathcal{P}$. Then\footnote{In what follows, we denote by $\mathcal{N}\left(\cdot\right)$ the CDF of a univariate standard normal random variable.},
\begin{equation}\label{eq:prop11}
\begin{gathered}
\sum_{i\in\mathcal{S}}X_0^{\gamma_i}\mathcal{N}\left(d_{0,T}^{M_i}\right)=X_0,\ \ \mathrm{\&}\ \
\widehat{\lambda}_{u,t}^{U,*}=\left(1-X_0\left(\sum_{i\in\mathcal{S}}\frac{1}{\gamma_i}\mathcal{N}\left(d_{0,T}^{M_i}\right)X^{\gamma_i}_0\right)^{-1}\right)\xi_u,
\end{gathered}
\end{equation}
for $\widehat{\phi}=\big[\phi,-\widehat{\lambda}_{u}^{U,*}+\xi_u\big]^{\top}$ such that $\widehat{\lambda}_{u,t}=:\widehat{\lambda}_{u,t}^{U,*}$. The value function hence agrees to
\begin{equation}\label{eq:upperbound}
\begin{aligned}
\widehat{J^{U,\mathrm{opt}}}\left(X_0\right)&=\frac{\eta X_0^{\gamma_d}}{1-\gamma_d}\mathcal{N}\left(d_{0,T}^{M_d}\right)+\frac{\eta X_0^{\gamma_u}}{1-\gamma_u}\mathcal{N}\left(-d_{0,T}^{M_u}\right)+H\left(\gamma_d,\gamma_u\right),
\end{aligned}
\end{equation} 
in which $H\left(\gamma_d,\gamma_u\right)=\frac{\mathcal{N}\left(d_{1}\right)}{1-\gamma_d}+\frac{\mathcal{N}\left(-d_{1}\right)}{1-\gamma_u}$, for $d_1={\sigma_T^M}^{-1}\left(\log\left(\eta^{U,*}\right)+\mu_T^M\right)$, and 
\begin{equation}
\begin{aligned}
 \mu^{M_i}_{0,T}&=\left(\bar{r}-r_0\right)\frac{1}{\sigma_r}B_{0,T}-\left(\bar{r}-\xi_u\left(\lambda_{u}-\widehat{\lambda}_{u}^{U,*}\right) +\frac{1}{2}\widehat{\phi}^{\top}\widehat{\rho}\widehat{\phi} \right)T+\Big(1-\frac{1}{\gamma_i}\Big){\sigma^{M}_T}^{2}\\ {\sigma^{M}_T}^{2}&=\widehat{\phi}^{\top}\widehat{\rho}\widehat{\phi} T-\frac{\sigma_r^2}{2\kappa^2}\left(2\left(\frac{1}{\sigma_r}{B_{0,T}}-T\right)+\sigma_r^{-2}\kappa B_{0,T}^2 \right)-\frac{2\sigma_r}{\kappa}e_2^{\top}\rho\phi\left(T-\frac{1}{\sigma_r}B_{0,T}\right)
\end{aligned}
\end{equation}
for $i\in\mathcal{S}$ and $B_{t,T}:=\frac{\sigma_r}{\kappa}\left(1-e^{-\kappa\left(T-t\right)}\right)$, where $\widehat{\rho}$ denotes the correlation matrix of $\widehat{z}_t$. Then, $d_{0,T}^{M_d}={\sigma_T^M}^{-1}\big(\log\left(\eta^{U,*}\right)+\mu_{0,T}^{M_d}\big)$ and  $d_{0,T}^{M_u}=-{\sigma_T^M}^{-1}\left(\log\left(\eta^{U,*}\right)+\mu_{0,T}^{M_u}\right)$. Furthermore,
\begin{equation}\label{eq:prop12}
X_0^{\gamma_i}=\eta^{U,*}\exp\left(\left(1-{1}/{\gamma_i}\right)\mu_T^M+2^{-1}{\sigma_T^M}^2\left(1-{1}/{\gamma_i}\right)^2\right), \ i\in\mathcal{S}.
\end{equation}
\end{corollary}
\begin{proof}
Note that $\mu_T^M=\mu_{0,T}^{M_i}-(1-1/\gamma_i){\sigma_T^M}^2$: Corollary \ref{col1} suffices to derive the results. 
\end{proof}

Analogous to the approximate results in section \ref{sec4.3}, Corollary \ref{col2} effectively discloses the results in Corollary \ref{col1} for $t=0$ applied to the financial market under study. Hence, closed-form approximate solutions emerge via the system of equations in \eqref{eq:prop11}. To certify the call for an approximation to $\widehat{\lambda}_{u,t}^{\mathrm{opt}}$ by $\widehat{\lambda}_{u,t}^{U,*}$ in \eqref{eq:prop11}, let us inspect the former's semi-analytic.\footnote{Let us remark that for example $X_T^{\mathrm{opt}}$ as well as $\eta^{\mathrm{opt}}$ are in this instance latent in Corollary \ref{col1}.} Accordingly, let us note that $\mathcal{G}'^{1}_{t,T}=-\xi_u$ and observe that $\mathcal{G}'^{2}_{t,T}$ represents a previsible process, which pursuant to Corollary \ref{col1} quarters $\{\widehat{\lambda}^{\mathrm{opt}}_{u,s}\}_{s\in\left[t,T\right]}$ as
\begin{equation}
\mathcal{G}'^{2}_{t,T}=\mathcal{R}^2_t\EE\left[\frac{X_T^{\mathrm{opt}}Z_T^{\widehat{\lambda}_u}}{\EE\big[X_T^{\mathrm{opt}}Z_T^{\widehat{\lambda}_u}\ \big| \ \eF_t\big]}\mathcal{R}_{x,T}^{-1}\left(\int_t^T\mathcal{D}^{z_u}_t\widehat{\phi}_{u,s}^{\mathrm{opt}}\mathrm{d}z_{u,s}-\int_t^T\mathcal{D}^{z_u}_t\widehat{\phi}_{u,s}^{\mathrm{opt}}\mathrm{d}s\right)\cond\eF_t\right].
\end{equation}
Let us recall that the parameter $\widehat{\phi}_{u,t}^{\mathrm{opt}}=\xi_u-\widehat{\lambda}_{u,t}^{\mathrm{opt}}$ lodges the shadow price of non-traded risk, in tandem with the fact that both $X_T^{\mathrm{opt}}Z_T^{\widehat{\lambda}_u}$ and $\mathcal{R}_{x,T}^{-1}$ analogously encircle integral representations enveloping $\widehat{\lambda}_{u,t}^{\mathrm{opt}}$, for any $\gamma_d\neq\gamma_u$. As a consequence, the possibly non-linear dependency of $\mathcal{G}'^{2}_{t,T}$ on the entire paths of the shadow price is indisputable. In virtue of the results in Corollary \ref{col1}, the analytical identity from which we may restore the truly optimal shadow price $\widehat{\lambda}_{u,t}^{\mathrm{opt}}$ immediately follows, similarly from \eqref{eq:sec33shadowprice2}, and ought to obey 
\begin{equation}\label{eq:sec52dualcontr}
\widehat{\lambda}_{u,t}^{\mathrm{opt}}\left(\mathcal{G}'^{2}_{t,T}+\xi_u\right)^{-1}=1-X_t\left(\sum_{i\in\mathcal{S}}\frac{1}{\gamma_i}\XX^{\lambda_2}_{\gamma_i}\left(\mathcal{A}_i\cond\eF_t\right)X_t^{\gamma_i}\right)^{-1},
\end{equation}
which we cannot solve for all $\gamma_d,\gamma_u\in\RR_+\setminus\left[0,1\right]$ such that $\gamma_d\neq\gamma_u$, due to the aforementioned path-dependency. In other words, $\widehat{\lambda}_{u,t}^{\mathrm{opt}}$ is solely acquirable for CRRA preferences, $\gamma_d=\gamma_u$. In that case, $\mathcal{G}'^{2}_{t,T}$ disappears alongside the path-dependent probability weights, culminating in $\widehat{\lambda}_{u,t}^{\mathrm{opt}}=\left(1-\gamma\right)\xi_u$; the approximation is exact under $\mathcal{P}=\RR$ for CRRA investors. For any $\gamma_d,\gamma_u\in\RR_+\setminus\left[0,1\right]$, we find that roughly $\widehat{\lambda}_{2,t}^{\mathrm{opt}}\in\left[\left(1-\gamma_d\right),\left(1-\gamma_u\right)\right]\left(\mathcal{G}'^{2}_{t,T}+\xi_u\right)$. On those grounds, letting $\mathcal{P}=\RR$ is sensible in consideration of approximation purposes. 
\subsubsection{Approximate Primal Optimality}
The suppression of $\DD^{1,2}\left(\left[0,T\right]\right)$ to $\mathcal{P}=\RR$ on the dual side sets out totally analytical expressions, amongst which an upper bound on the value function $\widehat{J^{U}}\left(X_0\right)$. In spite of this mathematical facility, it fails to procreate feasible portfolio decisions. In accordance with the technique in section \ref{sec4.3}, we use the optimality criterion implicit in Corollary \ref{col1} for the artificial portfolio rules $\widehat{x}_t^{\mathrm{opt}}$ and project it onto the feasible region of $x_t$'s:
\begin{equation}
X_T^{*}=X_0+\int_0^T\left(R_{f,t}+x_t^{\mathrm{opt},\top}\big|_{\widehat{\lambda}_{u,t}\in\RR}\Lambda_t\right)X_t^{*}\mathrm{d}t+\int_0^Tx_t^{\mathrm{opt},\top}\big|_{\widehat{\lambda}_{u,t}\in\RR}\Sigma_tX_t^{*}\mathrm{d}z_t,
\end{equation}
results, in which we let $x_t^{\mathrm{opt}}|_{\widehat{\lambda}_{u,t}\in\RR}$ be the baseline-optimal portfolio rules $x_t^{\mathrm{opt}}$ that arise out of Collary \ref{col2}, with the inclusion of a deterministic shadow price $\widehat{\lambda}_{u,t}:=\widehat{\lambda}_{u,t}^{L,*}\in\RR$. To obtain the foregoing admissible approximation, we rely on a projection kernel similar to \eqref{eq:sec43projectkern}, which nullifies any remaining allocation to the fictitious security $P_{3,t}$. Corollary \ref{col3} contains the congruous approximate portfolio weights for an unspecified $\widehat{\lambda}_{u,t}^{L,*}\in\RR$.

\begin{corollary}\label{col3}
Consider $\widehat{x}_t^{\mathrm{opt}}$ implicit in Corollary \ref{col2} for $x_{u,t}=0$. Consequently,
\begin{equation}\label{eq:optwqealthprop2}
\begin{gathered}
\bar{x}_t^{*}=\left({-A_{t,T}+{\lambda}}\right)\sum_{i\in\mathcal{S}}\frac{1}{\gamma_i}\frac{X^{\gamma_i}_t}{{X_t}^{*}}\mathcal{N}\left(d_{t,T}^{M_i}\left(1-2\Ind_{\left\lbrace i=u\right\rbrace}\right)\right)+A_{t,T},
\end{gathered}
\end{equation} 
where $\bar{x}_t^{*}=\Sigma_t^{\top}x_t^{*}$ for $x_t^{*}:=x_t^{\mathrm{opt}}|_{\widehat{\lambda}_{u,t}\in\RR}$, and $A_{t,T}=-B_{t,T}e_2+\xi$, along with $X^{\gamma_i}_t=\EE\big[\eta^{-1/\gamma_i}\widehat{M}_T^{1-1/\gamma_i}\ \big| \ \eF_t\big]$. More, let $d_{t,T}^{M_i}=({\sigma_{t,T}^M})^{-1}\left({\log\left(\eta^{L,*} M_t\right)+\mu_{t,T}^{M_i}}\right)$, and 
\begin{equation}
\begin{aligned}
 \mu^{M_i}_{t,T}&=\left(\bar{r}-\xi_u\left(\lambda_{u}-\widehat{\lambda}_{u}^{L,*}\right)+\frac{1}{2}\widehat{\phi}^{\top}\widehat{\rho}\widehat{\phi}\right)\Delta_{T,t}+\left(\bar{r}-r_t\right)\frac{B_{t,T}}{\sigma_r} +\Big(1-\frac{1}{\gamma_i}\Big){\sigma^{M}_{t,T}}^{2}\\ {\sigma^{M}_{t,T}}^{2}&=\widehat{\phi}^{\top}\widehat{\rho}\widehat{\phi} (T-t)-\frac{\sigma_r^2}{\kappa^2}\left(\frac{B_{t,T}}{\sigma_r}-\Delta_{t,T}+\frac{\kappa  B_{t,T}^2 }{2\sigma_r^2}\right)-\frac{2\sigma_r}{\kappa}e_2^{\top}\rho\phi\left(\Delta_{t,T}-\frac{B_{t,T}}{\sigma_r}\right)
\end{aligned}
\end{equation}
for $i\in\mathcal{S}$ and $\Delta_{t,T}:=T-t$. Then, $X_T^{*}$ is budget-feasible for all $\eta^{L,*}\in\RR$. Further,
\begin{equation}
X_t^{*}={X_t^{\gamma_d} Z_t^{\widehat{\lambda}_u}}\mathcal{N}\left(d_{t,T}^{M_d}\right)+{X_t^{\gamma_u} Z_t^{\widehat{\lambda}_u}}\mathcal{N}\left(-d_{t,T}^{M_u}\right),
\end{equation}
which singularises the market value of $I\big(B_T^{-1}Y_T^{\mathrm{opt}},\Pi_T\big)$ evaluated at any time $t\in\left[0,T\right]$. Notice that the value function $\widehat{J^L}\left(X_0\right)$ for \eqref{eq:optwqealthprop2} is not available in closed-form.
\end{corollary}
\begin{proof}
Here, $X_t^{\gamma_i}={\eta^{L,*}}^{-{1}/{\gamma_i}}\widehat{M}_t^{1-{1}/{\gamma_i}}\exp\big(\left(1-{1}/{\gamma_i}\right)\mu_{t,T}^{M_i}-\frac{1}{2}{\sigma_{t,T}^M}^2\left(1-{1}/{\gamma_i}\right)^2\big)$, $i\in\mathcal{S}$.
\end{proof}
The prior corollary particularises an analytically tractable approximation to the optimal portfolio rules concealed in an application of Corollary \ref{col1} to the economy at hand, without regard to the in closed-form non-existent $\eta^{L,*}\in\RR_+$ and $\widehat{\lambda}_{u,t}^{L,*}\in\RR$.\footnote{For financial intuition applicable to the previous two corollaries, see the analysis around Corollary \ref{col1}. Throughout, we set $\mathrm{D}_{\RR}\big(\widehat{\theta}^{L,*},\widehat{\theta}^{U,*}_t\big)=\widehat{J^{U,\mathrm{opt}}}-\widehat{J^{L,\mathrm{opt}}}$, for $\widehat{\theta}^{L,*}_t=\big(\widehat{\lambda}_{u,t}^{L,*},\eta^{L,*}\big)$ and  $\widehat{\theta}^{U,*}=\big(\widehat{\lambda}_{u,t}^{U,*},\eta^{U,*}\big)$.} To avoid numerical optimization approaches when identifying these controls, we may utilise the approximate optimal rules \eqref{eq:optwqealthprop2} in the sense of $\widehat{\mathrm{proj}}_{\widehat{\mathcal{A}}'^{\mathcal{P}}_{X_0}}$. Then, we insert the controls of Corollary \ref{col2} into the one above for full identification. Compared with ${\mathrm{proj}}_{\widehat{\mathcal{A}}'^{\mathcal{P}}_{X_0}}$, for which 
\begin{equation}
\widehat{J^{L,\mathrm{opt}}}\left(X_0\right)=\sup_{\widehat{\lambda}_{u,t}^{L,*}\in\RR,\eta^{L,*}\in\RR_+}\EE\left[U\left(X_T^{*},\Pi_T\right)\right]
\end{equation}
must provoke the optimal primal controls, we discern the potential numerical burden, cf. section \ref{sec4.2}. Subsequently, we compare the two approaches. In either case, we arrive at closed-form expressions for the optimal portfolio composition, which accompany lower and upper bounds, $\widehat{J^{L,\mathrm{opt}}}$ and $\widehat{J^{U,\mathrm{opt}}}$ on the optimal value function. The gap, $\mathrm{D}_{\RR}\big(\widehat{\theta}^{L,*}_t,\widehat{\theta}^{U,*}_t\big)$, ensues easily and the harmonious $\mathcal{CV}$ must be determined numerically, see section \ref{sec4.2}.

\subsection{Main Numerical Results}\label{sec5.3}
\begin{table}[!t]\centering
\begin{adjustbox}{max width=\textwidth}
\begin{threeparttable}
\def\sym#1{\ifmmode^{#1}\else\(^{#1}\)\fi}
\sisetup{table-space-text-post = \sym{***}}
\footnotesize
\begin{tabular}{l*{13}{S[table-align-text-post=false]}}

\toprule[2pt]
                &\multicolumn{1}{c}{Parameter}&\multicolumn{1}{c}{Value}&\multicolumn{1}{c}{}&\multicolumn{1}{c}{Parameter}&\multicolumn{1}{c}{Value}&\multicolumn{1}{c}{}&\multicolumn{1}{c}{Parameter}&\multicolumn{1}{c}{Value}&\multicolumn{1}{c}{}&\multicolumn{1}{c}{Parameter}&\multicolumn{1}{c}{Value}&\\
               
\midrule
 &  \multicolumn{2}{c}{$S_t$} & & \multicolumn{2}{c}{$Z_t$} & & \multicolumn{2}{c}{$r_t$} & & \multicolumn{2}{c}{$\pi_t$}\\
\cline{2-3}\cline{5-6}\cline{8-9}\cline{11-12}
\addlinespace
 & $\sigma_s$ & 0.158 &  &$\phi_S$  & -0.333 & & $\bar{r}$& 0.012 & & $\bar{\pi}$  & 0.054 \\
\addlinespace
  & $\lambda_s$ & 0.343 &  & $\phi_r$  & 0.170 & & $\kappa$ & 0.613 & & $\alpha$ &  0.027 \\
\addlinespace
  &  $\lambda_u$ & 0.027 &  & $\phi_{\pi}$  & 0.120 & & $\sigma_r$&  0.026 & & $\sigma_{\pi}$ &  0.014\\	 		
  \addlinespace
  & $\rho_{sr}$ & -0.129 &  &  $\phi_u$  & -0.014  & & $\lambda_r$&  -0.209 & & $\lambda_{\pi}$ & -0.105 \\	 		 
\addlinespace
&  & &  &   $\xi_u$  & 0.013  & & $\rho_{s\pi}$&  -0.024 & & $\rho_{r\pi}$ & -0.061 \\	 
\addlinespace
\midrule		 

\end{tabular}
\end{threeparttable}
\end{adjustbox}
\newline
\caption{\label{tab1}\textbf{Parameter input.} This table reports the benchmark parameter input that we employ in providing the numerical illustrations. The values in this table are identical to those that are documented in Table 1 of \citet{brennan2002dynamic}. We complement their baseline parameters with explicit values for $\lambda_u$ and $\phi_u$, that we set equal to the means of their counterparts $\lambda$ and $\phi$. Note that ${\xi}={0}_3$ holds. We keep the benchmark planning horizons fixed at $T=5$ and $T=10$, as unambiguously indicated. }
\end{table}

We progress by assessing the performance of the approximate method in the preceding economy for $\mathcal{P}=\RR$ and dual CRRA individuals on the basis of a couple of numerical examples. Towards this end, we put the benchmark parameter input as reported in \cite{brennan2002dynamic} to use. Tangibly, Table \ref{tab1} records these benchmark values. The results appearing in this subsection arise as a consequence of $N=10,000$ simulated paths for time increments $\Delta t_i=0.05$ as part of an Euler scheme, such that $M=T/0.05$. We consider two planning horizons, $T=5$ and $T=10$, because the inherently visible patterns therein carry over to extensions of these terminal dates. In addition to that, we study three dual CRRA investors whose risk profiles conform to $\left(\gamma_d,\gamma_u\right)\in\left\lbrace\left(5,5\right),\left(10,2\right),\left(15,3\right)\right\rbrace$ for a benchmark level and endowment equally identified with $\mathcal{K}=X_0=1$. The first element of the previous set mirrors the preferences of a standard CRRA agent, which we include for the sake of a validity check. That is to say, the approximation is exact for CRRA individuals, from where $\mathrm{D}_{\RR}\big(\widehat{\theta}^{L,*}_t,\widehat{\theta}^{U,*}_t\big)=0$ must hold, modulo numerical errors. For the remaining two investors, their risk profiles diverge in form of the coefficients of risk aversion, which leads to inexact approximations: $\mathrm{D}_{\RR}\big(\widehat{\theta}^{L,*}_t,\widehat{\theta}^{U,*}_t\big)>0$ ought to be true.\footnote{The width of the discrepancy between the uncoupled levels of risk aversion differs for these investors. We note that the approximation's precision practically decreases with this width according to \eqref{eq:sec52dualcontr}.} Last, in opposition to \citet{bick2013solving}, we do not simulate the upper bound $\widehat{J^{U}}\left(X_0\right)$ in order to preclude estimations biases concerning the compensating variations. Rather, we utilise the analytical bound to emphasise the method's accuracy.   

\begin{table}[!t]\centering
\begin{adjustbox}{max width=\textwidth}
\begin{threeparttable}
\def\sym#1{\ifmmode^{#1}\else\(^{#1}\)\fi}
\sisetup{table-space-text-post = \sym{***}}
\begin{tabular}{@{\extracolsep{15pt}}l*{7}{S[table-align-text-post=false]}}
\toprule[2.5pt]
                &\multicolumn{1}{c}{$\gamma_{1}$}&\multicolumn{1}{c}{$\gamma_{2}$}&\multicolumn{1}{c}{$\gamma_{3}$}&\multicolumn{1}{c}{$\gamma_{1}$}&\multicolumn{1}{c}{$\gamma_{2}$}&\multicolumn{1}{c}{$\gamma_{3}$}\\
\cline{2-4}\cline{5-7}
\addlinespace               
 & & {$T=5$}  &  &  & {$T=10$}  &  \\
\midrule
\addlinespace
\addlinespace
 LB & 0.135 & 0.234 & 0.178 & 0.193  & 0.416& 0.292 \\
   & {$\left(0.133,0.137\right)$} & {$\left(0.229,0.239\right)$} & {$\left(0.175,0.181\right)$}  &{$\left(0.192,0.195\right)$} & {$\left(0.410,0.422\right)$} & {$\left(0.289,0.295\right)$}\\
\addlinespace
\addlinespace
  UB & 0.136& 0.235 & 0.181 & 0.194  & 0.419& 0.295 \\
  \addlinespace
\addlinespace
 CV & 0.001 & 0.001 & 0.003 & 0.002  & 0.005 &0.004 \\
\addlinespace
\addlinespace
 AL & 1.664 & 2.314 & 4.308 & 1.885  & 4.600 & 4.219 \\	
\addlinespace
\cline{2-7}
\addlinespace
 $\bar{\theta}^{L}$ &{$\left(0.052,0.458\right)$} & {$\left(0.034,0.965\right)$} & {$\left(0.052,0.778\right)$}  &{$\left(0.052,0.225\right)$} & {$\left(0.031,0.745\right)$} & {$\left(0.048,0.495\right)$}\\
\addlinespace
\addlinespace
 $\bar{\theta}^{U}$  &{$\left(0.052,0.458\right)$} & {$\left(0.030,0.950\right)$} & {$\left(0.046,0.754\right)$}  &{$\left(0.052,0.225\right)$} & {$\left(0.026,0.720\right)$} & {$\left(0.041,0.478\right)$}\\
  \addlinespace
\midrule		 

\end{tabular}
\end{threeparttable}
\end{adjustbox}
\newline
\caption{\label{tab3}\textbf{Duality gaps, welfare losses and parameters.} The table reports the estimated lower bounds (rows marked LB) and analytically derived upper bounds (rows marked UB) on the optimal value function. In parentheses, the approximate 95\% confidence intervals are displayed. The rows marked CV and AL document the compensating variations and annual welfare losses, respectively. Additionally, the rows labelled $\bar{\theta}^{L}:=\big(-\widehat{\lambda}^{L,*}_{u,t},\eta^{L,*}\big)$ and $\bar{\theta}^{U}:=\big(-\widehat{\lambda}^{U,*}_{u,t},\eta^{U,*}\big)$ document, respectively, the primal and dual Lagrange multiplier and shadow price. The lower bounds materialise from inserting $\widehat{\theta}^{U,*}_t$ into the primal objective. The results rely on $N=10,000$ paths for $\Delta t_i=0.05$ time-steps, and $X_0=1$. The set $\left(\gamma_d,\gamma_u\right)\in\left\lbrace\left(5,5\right),\left(10,2\right),\left(15,3\right)\right\rbrace$ abbreviated by $\left\lbrace\gamma_1, \gamma_2, \gamma_3\right\rbrace$ comprises the dual CRRA risk profiles.}
\end{table}

Table \ref{tab3} resumes the lower and upper bounds on the optimal value function attributable to the approximating mechanism in Section \ref{sec4.2}, as well as the homologous compensating variations and annual welfare losses in basis points (bp) of the initial endowment; the `naive' and true procedures engender identical outcomes, due to which we also present the optimal parameters.\footnote{By the `naive' approach, we indicate the approximate method that involves the $\widehat{\mathrm{proj}}_{\widehat{\mathcal{A}}_{X_0}'^{\mathcal{P}}}$ projection.} Additionally, we accompany the estimated lower bounds with their agreeable 95\% confidence intervals. Preferably, these confidence bands occlude the upper bound, considering that $\widehat{J^{L}}\left(X_0\right)$ in \eqref{eq:mcobjective1} outlines a point estimate. Table \ref{tab3} in total gives a strong indication of our technique being near-optimal. Notably, the duality gaps vary from $0.001$ for the $T=5$ CRRA agent to $0.002$ in regard to the $T=10, \gamma_3$ dual CRRA agent. Apart from these negligible differences, the confidence intervals contain the upper bounds, alluding in unison to accuracy. Economically quantifying the degree of these gaps' widths, we observe that the breadth of the compensating variations and especially the annual welfare losses ratify this statement on the question of precision. The sizes of the annual welfare losses concretely range between $1.664$ bp for the $T=5$ CRRA individual to $4.600$ bp for the $T=10, \gamma_2$ dual CRRA agent.\footnote{From this perspective, these annual welfare losses construe the mere meaningful touchstones, because of their practical and financial pertinence, along with the dissimilarity in the planning horizons.} If we follow the interpretation of `AL' as a representative investor's annual management fee, the exiguity of these quantities is apparent. Furthermore, the insignificant optimality gaps for the CRRA investor verify the method's legitimacy. On a final note, the negligible dissimilarities between $\widehat{\theta}^{L,*}_{t}$ and $\widehat{\theta}^{U,*}_t$, jointly with identical performance, make the numerical effort redundant.

\begin{figure}[!t]
\begin{center}
\includegraphics[scale=1.13]{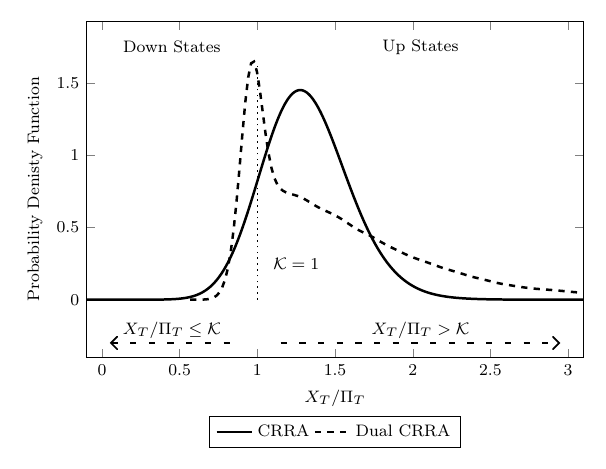}
\end{center}
\caption{\label{f3}\textbf{Probability density of horizon wealth.} The figure displays the probability density functions of the optimal and the `naively' approximated near-optimal terminal real wealth, respectively associated with an isoelastic agent and a dual CRRA agent. The coefficients of risk aversion adhere to $\gamma=5$ and $\left(\gamma_d,\gamma_u\right)=\left(10,2\right)$ for the respective CRRA and dual CRRA investors. The benchmark reads $\mathcal{K}=1$ (vertically dotted line). Both densities emanate from kernel density estimation with a bandwidth equal to 0.15 in application to the wealth equation \eqref{eq:mcobjective1}, based on $T=5$ for $N=10,000$ with $M=100$.}
\end{figure}

We note that the CVs and confidence intervals increase with $T$. However, the ALs remain unaffected, which suggests annual stability in the method's performance. Intuitively, the aberration of the approximate $x_t$ from the true $x_t^{\mathrm{opt}}$ is subjected to an extended stream of $z_{u,t}$ for higher $T$ and thus accumulates more outwardly for expanded $T$. Table \ref{tab3} furthermore reveals that besides the size of the dissimilitude in levels of risk aversion, the individual magnitudes of $\gamma_d$ and $\gamma_u$ affect the approximation's potency. Observe that the risk profile of $\gamma_3$ is quantitatively more widespread than that of $\gamma_2$, whereas the approximation to $x_t^{\mathrm{opt}}$ is moderately more accurate for $\gamma_3$. The overall coefficients of risk aversion comprised within the $\gamma_3$ profile outweighing those within the $\gamma_2$ profile exemplify the necessary nuance to the rule. In particular, Corollaries \ref{col2} and \ref{col3} show that the piecewise variety of the approximate rules abates with boosts in the risk profiles as a whole. Hence, $x_t^{\mathrm{opt}}$ for the $\gamma_3$ agent are less sensitive to economic shocks, making the constant approximation to $\widehat{\lambda}_{u,t}^{\mathrm{opt}}$ more eloquent. Last, in Figure \ref{f3}, we depict the densities for $X_T/\Pi_T$ of the CRRA and dual CRRA agents. The right-skewed nature of the latter shows that the dual CRRA function invokes a rough guarantee near the desired benchmark.\footnote{This is a logical consequence of the results in section \ref{sec5.1}, see \citet{kammaetalnetspar} for further details.} 

\section{Conclusion}\label{sec6}
This paper has proposed and assessed a computationally tractable method based on closed-form expressions allowing us to near-optimally approximate portfolio weights in imperfect environments that populate finite-horizon investors with non-trivial preferences. The procedure evades the unexceptional absence of analytical formulae in such circumstances by projecting the optimal closed-form portfolio composition in an artificial market onto the admissible region while confining the shadow prices to a suitable parametric family. As a consequence, the technique obtains explicit identities for the investments decisions that mimic the truly optimal rules up to the shadow prices. 

Optimality gaps unique to this method emerge as a result of convex duality and serve as concise evaluators of the technique's accuracy. We have accordingly tested the method in a model that accommodates unspanned inflation risk and occupies an agent with ratio dual CRRA preferences. Insignificant gaps and annual welfare losses varying between $1$ and $5$ basis points for different investors suggest precision of the method, cf. section \ref{sec4}. In spite of this simplified illustration, the general mechanism, cf. sections \ref{sec3} and \ref{sec4}, is applicable under investment constraints to a far-reaching class of market models that admittedly encompass non-Markovian dynamics and state-dependent preferences.

\begin{appendices}
\section*{Appendix}
\footnotesize
\section{Auxiliary Results}\label{sec:App:A}
Let us first introduce the notion of a Fr\'{e}chet derivative
\begin{defn}[Fr\'{e}chet Derivative]\label{frechet}
Let $V$ and $W$ be two Banach spaces, and consider a mapping
\begin{equation}
G:V\rightarrow W.
\end{equation}
In addition, introduce a set $U$ that complies with $V\subset U$. Then, $G$ is called Fr\'{e}chet differentiable at $X\in V$ if there exists an $A:U\rightarrow W$ that satisfies in addition to $\|Az\|_{W}\leq M\|z\|_{U}$ the following condition
\begin{equation}\label{eq:frechet}
\lim_{h\rightarrow 0}{\|h\|_{V}}^{-1}{\|G\left(X+h\right)-G\left(X\right)-Ah\|_{W}}=0,
\end{equation}
for some $M\in\RR_+$ and all $z\in U$. The Fr\'{e}chet derivative at hand reads $D_{X}F=A$.
\end{defn}
Second, let us turn to the Malliavin calculus. The next theorem embodies the Clark-Ocone formula
\begin{theorem}[Clark-Ocone]\label{clarkocone}
Suppose that $F:\mathcal{C}_0\left(\left[0,T\right]\right)\rightarrow\RR$ satisfies $F\in\DD^{1,2}$. Then,
\begin{equation}\label{eq:C0}
F=\EE\left[F\right]+\int_0^T\EE\left[\mathcal{D}_tF\cond\eF_t\right]\mathrm{d}W_t.
\end{equation}
\end{theorem}
Specifically regarding martingales, \eqref{eq:C0} furnishes an adequate tool to uniquely identify the integrand in their martingale representation. Ultimately, the next theorem introduces the Skorokhod operator.

\begin{theorem}[Skorokhod Operator]\label{skorokhod}
Consider some $F:\mathcal{C}_0\left(\left[0,T\right]\right)\rightarrow\RR$ such that $F\in\DD^{1,2}$. Then, 
\begin{equation}\label{eq:dualmall}
\EE\left[F\delta\left(h\right)\right]=\EE\left[F\int_0^Th_s\mathrm{d}W_s\right]=\EE\left[\int_0^T\mathcal{D}_sFh_s\mathrm{d}s\right]=\EE\left[\left\langle \mathcal{D}F,h\right\rangle_{L^2\left(\left[0,T\right]\right)}\right],
\end{equation}
for all $F\in\DD^{1,2}$ and $h\in\mathrm{Dom}\left(\delta\right)$, where the divergence operator, or Skorokhod integral lives by
\begin{equation}
\delta\left(Fh\right)=F\int_{0}^Th_s\mathrm{d}W_s-\int_0^T\mathcal{D}_sF,h_s\mathrm{d}s,
\end{equation}
in which its domain reads $\mathrm{Dom}\left(\delta\right):=\left\lbrace h\in L^2\left(\Omega;L^2\left(\left[0,T\right]\right)\right)\mid \ \left|\EE\left[\left\langle \mathcal{D}F,h\right\rangle_{L^2\left(\left[0,T\right]\right)}\right]\right|\leq c\|F\|_{L^2\left(\Omega\right)}\right\rbrace$.

\end{theorem}
Observe that these results hold for functionals on the Wiener space as described in the main text. For proofs, consider Chapter 4.41 of \cite{RogerWilli:2000:DiffusionsMarkovprocesses} or section 1.3.3 of \citet{nualart2006malliavin}.
\end{appendices}
\bibliographystyle{apalike}
{\footnotesize
\bibliography{APbib}}
\end{document}